\documentclass[12pt]{article}
\usepackage{amsmath,amsfonts,amsthm,amssymb,amscd,colordvi}
\usepackage{graphicx}
\usepackage[usenames]{color}
\renewcommand{\Re}{\mathop{\rm Re}\nolimits}
\renewcommand{\Im}{\mathop{\rm Im}\nolimits}
\newcommand{\p}{\partial}

\newcommand{\e}{\varepsilon}
\newcommand{\vk}{\varkappa}

\newcommand{\vp}{\varphi}

\newcommand{\gi}{\rho}

\newcommand{\lla}{\gamma}
\newcommand{\n}{{}_n}

\newcommand{\tm}{\theta_1^-}
\newcommand{\ts}{\tilde s}
\newcommand{\Vts}{V^{\tilde s}}
\newcommand{\ho}{\widehat\Omega}
\newcommand{\bk}{{\mathbf k}}
\newcommand{\Pho}{ \Phi}

\newcommand{\bb}{\mbox{\boldmath$\beta$}}
\newcommand{\R}{{\mathbb R}}
\newcommand{\C}{{\mathbb C}}
\newcommand{\IP}{{\bf P}}
\newcommand{\Z}{{\mathbb Z}}
\newcommand{\ZZ}{{\mathbb Z}^\infty_{+0}}
\newcommand{\E}{{\bf E}}
\newcommand{\cE}{{\cal E}}
\newcommand{\T}{{\mathbb T}}
\newcommand{\Tnr}{{\mathbb T}^{N-r}}
\newcommand{\N}{{\mathbb N}}
\newcommand{\PP}{{\bf P}}

\newcommand{\CC}{C_{pql}(2I)^pv^q\bar v^l}

\newcommand{\cA}{{\cal A}}

\newcommand{\cD}{{\cal D}} 
\newcommand{\cF}{{\cal F}}

\newcommand{\cH}{{\cal H}}

\newcommand{\cHR}{{\cal H}^{\text{res}}}

\newcommand{\cQ}{{\cal Q}}

\newcommand{\cT}{{\cal T}}

\newcommand{\strela}{\rightharpoonup}
\newcommand{\lc}{\lceil}
\newcommand{\rc}{\rceil}

\newcommand{\const}{\mathop{\rm const}\nolimits}
\newcommand{\as}{\quad\mathop{\rm as} \nolimits\quad}

\newcommand{\supp}{\mathop{\rm supp}\nolimits}

\newcommand{\Arg}{\mathop{\rm Arg}\nolimits}

\def\dbar{{\mathchar'26\mkern-12mu d}}
\def\12{\tfrac12}

\def\lan{\langle}
\def\ran{\rangle}

\def\eps{\varepsilon}
\theoremstyle{plain}
\newtheorem{theorem}{Theorem}[section]
\newtheorem{lemma}[theorem]{Lemma}
\newtheorem{proposition}[theorem]{Proposition}
\newtheorem{corollary}[theorem]{Corollary}
\theoremstyle{definition}

\theoremstyle{remark}

\numberwithin{equation}{section}

\begin{document}

\author{Sergei
Kuksin\footnote{CNRS and  I.M.J, Universit\'e Paris Diderot-Paris 7, Paris, 
 France, e-mail:
  kuksin@math.jussieu.fr },\addtocounter{footnote}{2} Alberto Maiocchi \footnote{Laboratoire de
Math\'ematiques,  Universit\'e de Cergy-Pontoise, 2  avenue Adolphe Chauvin,
Cergy-Pontoise, France,
e-mail: alberto.maiocchi@unimi.it}}

\title{Resonant averaging for
weakly nonlinear stochastic Schr\"odinger equations }
\date{}
%%\date{(preliminary version)}
\maketitle

\begin{abstract}
We consider the free linear Schr\"odinger  equation  on a torus
$\T^d$, perturbed by a hamiltonian  nonlinearity,
driven by a random force and damped by 
a linear damping:
$$
 u_t -i\Delta u +i\nu \rho |u|^{2q_*}u
= - \nu f(-\Delta) u
+ \sqrt\nu\,\frac{d}{d t}\sum_{\bk\in \Z^d} b_\bk\bb^\bk(t)e^{i\bk\cdot x} \ .
$$
Here $u=u(t,x),\ x\in\T^d$, 
 $0<\nu\ll1$,  $q_*\in\N\cup\{0\}$, 
 $f$ is a positive continuous function, $\rho$ is a positive parameter 
  and 
$\bb^\bk(t)$ are standard independent complex Wiener processes.  We are interested in limiting, as
$\nu\to0$, behaviour of  solutions for this equation 
and of its stationary measure. Writing the equation in the slow time $\tau=\nu t$, we
 prove that the limiting  behaviour of the  both % of solutions and of the stationary measure
  is  described by the   {\it effective equation}
$$
u_\tau+ f(-\Delta) u = -iF(u)+\frac{d}{d\tau}\sum  b_\bk\bb^\bk(\tau)e^{i\bk\cdot x} \ ,
$$
where the nonlinearity $F(u)$ is made out of the resonant terms of the monomial $ |u|^{2q_*}u$. 
We explain the relevance of this result for the problem of weak turbulence 
\end{abstract}

\tableofcontents
\section{Introduction}\label{s0}
\subsection
{ Equations}
We study the free Schr\"odinger equation on the torus $\T^d_L= \R^d/(2\pi L\Z^d)$,
\begin{equation}\label{*1}
u_t(t,x)-i\Delta u(t,x)=0,\quad x\in\T^d_L\,, 
\end{equation}
stirred by a perturbation, which comprises a
hamiltonian  term, a linear damping and a random force. That is, we consider the equation 
\begin{equation}\label{1.11}
\begin{split}
u_t-i\Delta u= - i\eps^{2q_*}|u|^{2q_*}u-\nu f(-\Delta) u
+\sqrt\nu \frac{d}{dt}\sum_{\bk\in\Z^d}  b_\bk\bb^\bk(t) e^{i{L^{-1}\bk}\cdot x}
\ ,\\
 u=u(t,x),\quad x\in \T^d_L\,,
\end{split}
\end{equation}
where   $\ q_*\in \N \cup\{0\}$ and
$  \eps,\nu>0$ are two small parameters, controlling  the size of the perturbation. 
 The damping  $-f(-\Delta) $ is the selfadjoint 
  linear operator in $L_2(\T^d_L)$ which acts on the exponents $e^{i{L^{-1}\bk}\cdot x}$, $\bk\in\Z^d$, 
 according to
\begin{equation}\label{f}
f(-\Delta) e^{i{L^{-1}\bk}\cdot x}= \lla_\bk e^{i{L^{-1}\bk}\cdot x}, \qquad \lla_\bk= f( \lambda_\bk  )
\quad\text{where}\quad \lambda_\bk= |\bk|^2L^{-2}.
\end{equation}
 The function $f$ is real positive and continuous. To avoid technicalities, not relevant for 
 this work, we assume that 
$\ 
f(t)\ge C_1 |t| +C_2$ for all $t$,
for suitable positive constants $C_1,C_2$ (for example, $f(-\Delta)u=-\Delta u+u$). 
The processes $\bb^\bk, \bk\in\Z^d$, are
standard independent complex 
Wiener processes, i.e., $\bb^\bk(t)=\beta^\bk_+(t)+i\beta^{\bk}_-(t)$, 
where $\beta_{\pm}^\bk(t)$ are standard independent  real Wiener
processes. The real numbers $b_\bk$ are all non-zero and decay  fast
when $|\bk|\to\infty$. 

The nonlinearity  in \eqref{1.11} is hamiltonian and may be written as
\begin{equation}\label{*ham}
-i\e^{2q_*} |u|^{2q_*}u=\e^{2q_*}i   \,\nabla \cH(u) ,\qquad
\cH(u) =-\frac1{2q_*+2}\int |u(x)|^{2q_*+2}dx.
\end{equation}
We {\it assume} that eq. \eqref{1.11} with sufficiently smooth initial data $u_0(x)$
is well posed. It is well known that this assumption holds (at least) under
some restriction on $d,q_*$ and the growth of $f(t)$ at infinity, see in  Section~\ref{s1.1}

Equation \eqref{1.11} with small $\nu$ and $\eps$  is important for
physics and mathematical physics, where it serves  as a universal
model. In particular, it is  used in the physics of plasma 
to describe small oscillations of the media on long  time scale, see 
 \cite{Fal, Naz, ZL75, ZLF}. 
 The parameters $\nu$ and $\eps$ measure, respectively, the inverse 
  time-scale of the forced oscillations, and their  amplitude. 
  Physicists consider different regimes, where the two parameters
  are tied in various ways. To do this they assume some relations between $\e$ and $\nu$,
  explicitly or  implicitly. In our work we choose
  $$
  \eps^{2q_*}=\gi \nu,
  $$
  where $\rho>0$ is a constant. This assumption is within the  usually imposed bounds,
  see  \cite{ Naz}.   Passing to the slow time
$\tau=\nu t$, we  get the rescaled equation 
\begin{equation}\label{1.111}
\begin{split}
\dot u+i\nu^{-1}\big(-\Delta u\big)= -f(-\Delta) u- i\gi|u|^{2q_*}u
+%\frac{d}{d\tau}
\sum_{\bk\in\Z^d}  b_\bk \dot\bb^\bk(\tau) e^{i{L^{-1}\bk}\cdot x}
\ ,
\end{split}
\end{equation}
where $u=u(\tau,x)$, $x\in \T^d_L$ and 
 the dot   $\ \dot{}\ $  stands for 
$\frac{d}{d\tau}$.  If we write $u(\tau,x)$ as Fourier series,
$\ 
u(\tau,x)=\sum_\bk v_\bk(\tau)e^{iL^{-1}\bk\cdot x},
$
then in view of \eqref{*ham},  eq. \eqref{1.111} may be written as the system
\begin{equation}\label{1.100}
\begin{split}
\dot v_\bk+i \nu^{-1}\lambda_\bk v_\bk=-\gamma_\bk v_\bk + 2\rho\, i   \,\frac{\p \cH(v)}{\p \bar v_\bk}  
+b_\bk \dot\beta^\bk(\tau),\quad \bk \in \Z^d.
\end{split}
\end{equation}
Here $\cH(v)$ is the Hamiltonian $\cH$, expressed in terms  of the Fourier coefficients
$v=(v_\bk, \bk\in\Z^d)$:
\begin{equation}\label{Hv}
\cH(v)=\frac1{2q_*+2}
\sum_{\bk_1,\dots \bk_{2q_*+2}\in\Z^d} v_{\bk_1}\dots v_{\bk_{q_*+1}} \bar v_{\bk_{q_*+2}} \dots
\bar v_{\bk_{2q_*+2}} \,  \delta^{1\ldots q_*+1}_{q_*+2\ldots 2q_*+2}\,,
\end{equation}
and we use a notation, standard in physics (see \cite{Naz}):
\begin{equation}\label{N1}
\delta^{1\ldots q_*+1}_{q_*+2\ldots 2q_*+2}=\left\{\begin{array}{cc}
1 & \mbox{if }\bk_1+\ldots+\bk_{q_*+1}-\bk_{q_*+2}-\ldots-\bk_{2q_*+2} = 0 \\
0 & \mbox{otherwise}
\end{array}
\right.\ .
\end{equation}
As before  we
are interested in  the limit $\nu\to0$, corresponding   to small oscillations in the original
non-scaled equation.\footnote{See \cite{KNer13} for a theory of this equation for the case
when $f(t)=t+1$ and $\nu=\infty$.}

\subsection
{Weak Turbulence} \label{sWT}
In physics equations \eqref{1.111} with $\nu\to0$ are treated by the theory of weak turbulence,
or 
 WT (this abbreviation also may stand for `Wave Turbulence', but the difference between  the two notions
seems for us negligible);  see the  works, quoted above as well as \cite{CZ00}.
 That theory either deals with  equation \eqref{1.111},
where  $L=\infty$ 
by formal replacing Fourier series for $L$-periodic functions with Fourier integrals and makes with them bold 
transformations, or  considers the limit  $\nu\to0$ simultaneously with the limit  $L\to\infty$ and 
treats the two of  them in an equally bold way.\footnote{
Alternatively (and more often) people, working on WT, consider the HPDE \eqref{1.111}${}_{f=0, b_\bk=0\,\forall \bk}$
and treat it in a similar formal way, see \cite{Fal, Naz, ZLF, CZ00}. 
The corresponding problems do not fit our technique. Some recent progress in 
their rigorous study may be found in \cite{FGH}.
}
Concerning this limit WT makes a number of remarkable 
predictions, based on tools and ideas, developed in the community, which can be traced 
 back to the work \cite{Peierls}. The most famous of them deals with the
energy spectrum of solutions $u(\tau,x)$.  To describe the corresponding claims,
 consider the quantity $\E|v_\bk(\tau)|^2$, average it in
 time\footnote{certainly this is not needed if we consider stationary solutions of the equation}
 $\tau$  %using certain time-scale $T_{\text{scale}}$, 
 and in wave-vectors 
$\bk\in \Z^d$  such that $|\bk|/ L \approx r>0$; next properly scale this and denote the result $E_r$.
The function $r\to E_r$ is called the {\it energy spectrum}.  It is predicted by WT that, in certain
{\it inertial range} $[r_1,r_2]$, which is contained in the spectral zone 
where the random force is negligible (i.e., $|b_\bk|\lll (\E|v_\bk|^2)^{1/2}$ if $r_1\le |\bk|/L \le r_2$), the
energy spectrum has an algebraic behaviour:
\begin{equation}\label{KZ}
E_r\sim r^{-\alpha}\quad \text{for}\quad  r\in[r_1,r_2],
\end{equation}
for a suitable $\alpha>0$.  The WT limit, in fact, deals with two iterated limits:
\begin{equation}\label{limits}
\begin{split}
L\to\infty,\quad \nu\to0  \, .\quad% T_{\text{scale}}\to\infty
\end{split}
\end{equation}
 Relations between the two parameters in \eqref{limits}
is not quite clear for us, and it may be better to talk about the WT limits (rather then about a single case). But 
all the limits should lead to  relations \eqref{KZ} with  finite $\alpha$'s.

We suggest to study  the WT limits (at least, some of them) by splitting the limiting process 
  in two steps:

I) prove that when $\nu\to0$, main characteristics of solutions $u^\nu$ have limits of order one, 
described by certain {\it effective equation} which is a nonlinear stochastic equation with coefficients 
of order one and with a hamiltonian  nonlinearity, made out the resonant terms of the 
nonlinearity $ i\gi|u|^{2q_*}u$.

II) Show that main characteristics of solutions for the effective equation have non-trivial limits of 
order one, when $L \to\infty$ and $\rho=\rho(L)$ is a suitable function of $L$.

In this work we perform Step I, postponing  Step II for the future. We stress  that the results of Step I
along cannot justify the predictions of WT since the latter  (e.g. the asymptotic \eqref{KZ}) cannot hold 
when the period $L$ is fixed and finite. We believe that a suitable choice of the function $\rho(L)$ leads to
non-trivial spectral asymptotic \eqref{KZ}, but it is open for discussion up to what extent the corresponding 
choice of the limits in \eqref{limits} agrees with physics and the tradition of WT.

 As the title of the paper suggests, our argument is a form of averaging.
The latter is a tool which is used by the WT community on a regular basis, either explicitly (e.g. see \cite{Naz}), or
implicitly.

\subsection
{  Inviscid limits for damped/driven hamiltonian PDE}  Equation \eqref{1.111} is the linear hamiltonian PDE
\eqref{*1}, driven by the random force and damped by the damping $-f(-\Delta u)-i\rho|u|^{2q_*}u$.  Damped/driven
hamiltonian PDE (HPDE) and the inviscid limits in these equations when the random force and the 
damping go to zero, are very important for physics. In particular, since the $d$-dimensional Navier-Stokes
equation (NSE) with a random force can be regarded as a damped/driven Euler equation (which is an HPDE), and
the inviscid limit for  the NSE describes the $d$-dimensional turbulence. The NSE with random force, 
 especially when
$d=2$, was intensively studied last years, but the corresponding inviscid limit turned out to be very 
complicated even for $d=2$, see \cite{KS}. The problem of the inviscid limit 
becomes feasible when the underlying HPDE is integrable or linear. The most famous integrable PDE is the KdV
equation. Its damped/driven perturbations and the corresponding inviscid limits were studied   in \cite{KP08, K10}.  In \cite{K12} the method of those works was applied to 
the situation when the unperturbed HPDE is the Schr\"odinger equation
\begin{equation}\label{*5}
u_t + i(-\Delta u + V(x) u)=0,\qquad x\in\T^d_L,
\end{equation}
where the potential $V(x)$ is in general position. Crucial for the just mentioned works  is that there the 
unperturbed HPDE is free from strong resonances. For \cite{KP08, K10} it means that all solutions of KdV
are are almost-periodic functions of time, and for typical solutions the corresponding frequency vectors are 
free from resonances; while for  \cite{K12}  it means that for the typical potentials $V(x)$, considered in 
\cite{K12}, the spectrum of the linear operator in \eqref{*5} is non-resonant. 

In contrast, now the linear operator in the unperturbed equation \eqref{*1}  has the eigenvalues 
$\lambda_\bk, \bk\in\Z^d$ (see \eqref{f}), which are highly resonant. To explain the corresponding 
 difficulty, we rewrite equation
\eqref{1.111}=\eqref{1.100} as a fast-slow system, denoting 
$
I_\bk=\tfrac12 |v_\bk|^2, \;\; \vp_\bk=\Arg v_\bk
$.
In the new variables  eq.~\eqref{1.111} reads 
\begin{equation}\label{*6}
\dot I_\bk(\tau) = v_\bk\cdot P_\bk(v) +b_\bk^2 +b_\bk(v_\bk\cdot\dot\bb^\bk), 
\end{equation}
\begin{equation}\label{*7}
\dot \vp_\bk(\tau) =-\nu^{-1} \lambda_\bk +\dots, 
\end{equation}
where $\bk\in\Z^d$. Here the dots stand for a term of order one (as $\nu\to0$) which  has a singularity 
when $v_\bk=0$. 

Finite-dimensional systems of the form \eqref{*6}, \eqref{*7}, where $\bk$ belongs to a finite set $\bold K$
and $\dots$ stand for  smooth functions, are considered in the classical {\it stochastic averaging}.  If the fast
dynamics $\dot\vp_\bk = -\nu^{-1}\lambda_\bk, \ \bk\in \bold K$, is ergodic on the torus $\T^{\bold K}$
(i.e., the vector $(\lambda_\bk, \bk\in\bold K)$ is non-resonant), then the stochastic averaging theorem  from 
\cite{Khas68, FW03}  states that distributions of solutions for \eqref{*6} converge, as $\nu\to0$, to those of solutions 
for the averaged system
\begin{equation}\label{*8}
\dot I_\bk(\tau) =\lan v_\bk\cdot P_\bk\ran(I) +b_\bk^2 +b_\bk \sqrt{2I_\bk}\,\dot\beta^\bk(\tau)\,,
\quad \bk\in\bold K. 
\end{equation}
Here $\lan\cdot\ran$ signifies the averaging in  angles $\vp\in \T^{\bold K}$ (in this case 
 the averaged system is relatively 
simple since the noises in the systems  \eqref{*6} and  \eqref{*7} are diagonal). 
If the system of equations  \eqref{*6},
\eqref{*7} with  $\bk\in\Z^d$ corresponds to damped/driven perturbations of the Schr\"odingier  equation 
\eqref{*5} with non-resonant spectrum,  then the Khasminski scheme formally  applies,  but its realisation 
encounters significant difficulties since equations \eqref{*7} and the dispersions (i.e. the stochastic parts) of 
equations \eqref{*8} have singularities at the locus 
\begin{equation}\label{game}
\Game=\{ I: I_\bk =0 \;\;\text{for some} \; \bk\},
\end{equation}
which is dense in the space of infinite sequences $\{I_\bk\ge0, \bk\in\Z^d\}$.  Moreover, the drift in the 
averaged system \eqref{*8}${}_{\bk\in\Z^d}$ is a non-Lipschitz vector-field, so its unique solvability 
 is unclear. A way to overcome this difficulty was developed in \cite{KP08, K10, K12}. 
Namely it was shown that, in the non-resonant case, there exists a regular system of equations 
\begin{equation}\label{*eff}
\dot v_\bk = R_\bk(v) + b_\bk\dot\bb^\bk(\tau), \quad \bk\in\Z^d, 
\end{equation}
obtained from the non-singular (as $\nu\to0$) part of the original damped/driven HPDE
 by a kind of averaging, such that under the natural mapping
$$
v_\bk \mapsto I_\bk = \tfrac12 |v_\bk|^2
$$
 solutions of \eqref{*eff} transform to weak (in the sense of the stochastic calculus) solutions of 
the averaged system \eqref{*8}.
In   \cite{K12} this algebraical fact was used to study the inviscid limits in the damped/driven perturbations of 
nonresonant equations \eqref{*5}. The system \eqref{*eff} is called the {\it effective equation}. 

The argument in \cite{K12} uses crucially the non-resonance assumption.
If  it is applied to 
eq.~\eqref{*5}, perturbed by the terms, forming the r.h.s. of \eqref{1.11}, then the corresponding 
effective equation is
 linear. That is, {\it in the non-resonant situation the limiting as $\nu\to0$ dynamics is linear}.

\subsection
{  Effective equation for \eqref{1.111}} 
 Consider equations \eqref{*6}, \eqref{*7} which come from eq.~\eqref{1.100}. Now
 the fast motion is resonant, and  if
$
\lambda_{\bk_1} l_1+\dots \lambda_{\bk_m} l_m =0,
$
where $l_j$'s  are non-zero integers, then the corresponding linear combination of phases 
$
\Phi=
\vp_{\bk_1} l_1+\dots \vp_{\bk_m} l_m 
$
is not a fast,  but a slow variable since  $\dot\Phi\sim1$ when $\nu\to0$.  This difficulty is well known in
finite-dimensional  systems (see \cite{AKN} and Section~\ref{s3.1} below). There it is resolved  based 
on a lemma, stating that if $\bk\in \bold K$, $|\bold K|<\infty$,  then a suitable unimodular linear operator 
transforms the torus  $\T^{\bold K}$ to itself in such a way that in 
the transformed system the $\vp$-equations become $\dot \vp'_\bk=-\nu^{-1} \lambda'_\bk$,  $\bk\in \bold K$,
and among the new frequencies $(\lambda'_\bk,  \bk\in \bold K)$ some $r$ of them are zero
(where  $r\le |\bold K|$), while the remaining $ |\bold K| -r$ components are non-resonant. The new
system has $ |\bold K| -r$  fast motions and $ |\bold K| +r$ slow motions, and the usual non-resonant 
averaging may be used to study it. 

 This approach does not apply 
directly to
 the infinite-dimensional system \eqref{*6}, \eqref{*7} with $\bk\in\Z^d$ since the unimodular linear transformation above has no infinite-dimensional analogy (at least,
we do not see any), and  the finite-dimensional averaged equations have no good limit 
as $|\bold K|\to \infty$.  But some {\it formulas} for the resonant averaging do have limits as 
$ |\bold K|\to \infty$. By analogy with the construction of the non-resonant effective equation
\eqref{*eff} (which are linear for the perturbations we consider), they allow us in Section~\ref{s2.3} to 
guess right resonant effective equation. It turned out to be a %nonlinear stochastic equations of the 
 damped/driven hamiltonian system
\begin{equation}\label{*eff1}
\begin{split}
\dot v_\bk=-\gamma_\bk v_\bk + 2\rho\, i   \,\frac{\p \cH^{\text{res}}(v)}{\p \bar v_\bk}  
+b_\bk \dot\beta^\bk(\tau),\quad \bk\in\Z^d,
\end{split}
\end{equation}
where $\cHR$ is obtained as the resonant average of the Hamiltonian $\cH(v)$: 
\begin{equation}\label{hres}
\cHR(v)= \frac1{2q_*+2}
\sum_{\bk_1,\dots \bk_{2q_*+2}\in\Z^d} v_{\bk_1}\dots v_{\bk_{q_*+1}} \bar v_{\bk_{q_*+2}} \dots
\bar v_{\bk_{2q_*+2}} \,  \delta^{1\ldots q_*+1}_{q_*+2\ldots 2q_*+2}\,
 \delta(\lambda^{1\ldots q_*+1}_{q_*+2\ldots 2q_*+2})\,,
\end{equation}
and we use another physical notation:
\begin{equation}\label{N2}
\delta(\lambda^{1\ldots q_*+1}_{q_*+2\ldots 2q_*+2})=\left\{\begin{array}{cc}
1 & \mbox{if }\lambda_{\bk_1} +\ldots+\lambda_{\bk_{q_*+1}}
-\lambda_{\bk_{q_*+2}}- \ldots- \lambda_{\bk_{2q_*+2}}
 = 0 \\
0 & \mbox{otherwise}
\end{array}
\right.\ .
\end{equation}
That is, the effective equation is obtained from the system 
 \eqref{1.100} by a simple procedure: we drop the fast rotations
and replace the Hamiltonian $\cH$ by its resonant average $\cHR$. In difference with the non-resonant case,
this is a nonlinear system.
It turns out that equations \eqref{*eff1} with sufficiently smooth initial data $v(0)$ are well posed.
 We obtain  this result  in Section~\ref{s5.2} as a consequence of our  main theorem.  
 
 Since the Hamiltonian
 $\cHR$ is obtained by averaging, it has infinitely many commuting quadratic integrals of motion,
 including $|u|^2_{L_2}$ and $|\nabla u|^2_{L_2}$ (where $u=u(x)$ is the function with the Fourier 
 coefficients $v=\{v_s\}$), see Lemma~\ref{l.symm}.
 The corresponding symmetries of the hamiltonian equation $\dot v=i\nabla \cHR(v)$
 also are (weak) symmetries of the effective equation, see Lemma~\ref{l.invar}. This 
 makes the effective equation a bit similar to the stochastic 2d~Navier-Stokes equations on the torus,
 cf.~\cite{KS} (in fact, the former is significantly simpler than the latter). 
 
 The construction of the resonant Hamiltonian $\cHR$ is in the spirit of WT, and the corresponding 
 hamiltonian equation is known there as the {\it equation of discrete turbulence},
 see \cite{Naz},~Chapter~12.  Similar equations were  considered by  mathematicians, interested in related problems
 (see \cite{GG12}), and were used by them for intermediate arguments
  (e.g., see \cite{FGH}).

The stochastic equation \eqref{*eff1} was not considered before our work.

\subsection
{ Results}  Main results of our work  are stated in Sections~\ref{s5.2}-\ref{s2.5}
and are proved in Section~\ref{s3}. They establish that long-time behaviour of solutions for equations 
\eqref{1.111}, when $\nu\to0$, is controlled  by solutions for the effective equation. 
We start with the results on the Cauchy problem for eq.~\eqref{1.111}.  So, let $v^\nu(\tau)$ be a solution of 
\eqref{1.100} such that 
$$
v^\nu(0)=v_0,
$$
where $v_0=(v_{0\bk}, \bk\in\Z^d)$ corresponds to a sufficiently smooth function $u_0(x)$. Let us fix any $T>0$.

Consider the list  $\cA$ of resonances in eq.~\eqref{*1}.  That is,  the set of all nonzero integer 
vectors $\xi=(\xi_\bk, \bk\in\Z^d)$ of finite length, satisfying  $\sum_{\bk \in\Z^d} \xi_k\lambda_k=0$.  For
$\xi\in\cA$ consider the corresponding resonant combination of phases of solutions $v^\nu(\tau)$,
$\ 
\Phi^\xi(v^\nu(\tau)):=  \sum_{\bk \in\Z^d} \xi_k \vp_k(v^\nu(\tau))\in S^1$, $ 0\le\tau\le T.
$
Consider also the vector of actions 
$
I(v^\nu(\tau))=\{ I_\bk(v^\nu(\tau)), \bk\in \Z^d\}.
$

\noindent
{\bf Theorem 1}. When $\nu\to0$, we have the weak convergence of  measures 
$$
\cD\big(I(v^\nu(\tau)) \big) \strela  \cD\big(I(v^0(\tau))\big),
$$
where $v^0(\tau),\ 0\le\tau\le T$, is a unique solution of equation 
\eqref{*eff1} such that $v^0(0)=v_0$. 
\medskip

The resonant combinations of phases $\Phi^\xi(v^\nu(\tau))$
also, in certain sense, converge in distribution to $\Phi^\xi(v^0(\tau))$. See Section~\ref{s5.2} for exact 
statement, which is more involved then Theorem 1. 
 On the contrary, if a finite vector
$s=(s_\bk, \bk\in\Z^d)$ is non-resonant, i.e.  $\sum s_\bk \lambda_\bk\ne0$, then the measure
$
\cD(\Phi^{(s)}(v^\nu(\tau))=: \mu^{(s)}(\tau),
$
mollified in $\tau$, converges when $\nu\to0$ to the Lebesgue measure on $S^1$. 
That is, solutions of the effective equation \eqref{*eff1} approximate in law slow components of solutions
$v^\nu(\tau)$, but not their fast components. These assertions certainly are related to the random phase approximation, accustomed in the WT, but we found it difficult to be here more specific.

The limiting behaviour of solutions $v^\nu(\tau)$ can be described without evoking the effective equation.
Namely, denote by $\cA_m$ the set of resonances $\xi\in\cA$ of length $|\xi|\le m:=2q_*+2$. Then 
the vectors $I^\nu(\tau)=I(v^\nu(\tau))$ and $\Phi^\nu(\tau)=\big(\Phi^\xi(v^\nu(\tau)), \xi \in\cA_m\big)$
converge in distribution to  limiting processes $I^0(\tau)$ and $\Phi^0(\tau)$, which are weak solutions 
of the corresponding averaged equations. Those equations depend only on $I$ and $\Phi$, but the equations
for $\Phi$ have strong singularities at the locus $\Game$ and rigorous formulation of this convergence is
involved, see Proposition~\ref{p.slow}. We see no way to prove that the system of averaged equations is
well posed, so this description is not unique (and we see no way to prove this convergence without using
the effective equation). 

\medskip

Now consider a stationary measure $\mu^\nu$ for  equation \eqref{1.111}  (it always exist). We have 

\noindent
{\bf Theorem 2}. Every sequence $\nu'_j\to0$ has a subsequence $\nu_j\to0$ such that 
$$
I\circ \mu^{\nu_j}\strela  I \circ \mu^0,\qquad
\Phi^{(\xi)}\circ \mu^{\nu_j}\strela \Phi^{(\xi)}\circ \mu^0\quad \forall\, \xi\in\cA,
$$
where $\mu^0$ is a stationary measure for equation \eqref{*eff1}. If a vector $s$ is non-resonant, then the
measure $\Phi^{(s)}\circ\mu^\nu$ converges, as $\nu\to0$, to the Lebesgue measure on $S^1$.
\medskip

\noindent
{\bf Theorem 3}.  If equation \eqref{*eff1} has a unique stationary measure $\mu$, then
$\mu^\nu\strela\mu$ as $\nu\to0$.

Existing technique allows to prove that the stationary measure is unique if $q^*=1$ and the function $f(\lambda)$
grows fast enough as $\lambda\to\infty$, see in \cite{KS}.  We do not discuss here the corresponding results. Due
to  relatively simple structure,  the effective equation has a number of 
nice properties (see \eqref{p4}-\eqref{invar}). Accordingly, 
we  believe that the uniqueness of the stationary measure $\mu$ may be established under mild
 restrictions on $f$. Also, those properties of the effective equation give us  good hope that the Step~II from Section \ref{sWT}
may be rigorously performed. The method of the recent work \cite{FGH} may be relevant for that. 
\medskip

\noindent
{\bf Notation and Agreement.} The  {\it stochastic terminology} we use agrees with \cite{KaSh}.
 All filtered probability spaces we work with satisfy the {\it usual condition} (see \cite{KaSh}).\footnote{
I.e., the corresponding filtrations $\{\cF_t\}$ are continuous from the right, and each $\cF_t$ contains all negligible sets.}
 Sometime we forget to mention that 
a certain relation holds a.s. 
\\
{\it Spaces of integer vectors.} We denote by $\Z^\infty_0$ the set of vectors in $\Z^\infty$ 
of finite length,  and denote $\ZZ=\{s\in\Z^\infty_0: s_k\ge0\  \forall\,k\}$. Also
see \eqref{notation} and \eqref{hren}. 

\noindent 
{\it Infinite vectors.} For an infinite vector $\xi=(\xi_1,\xi_2,\dots)$ (integer, real or complex) and $N\in\N$ 
we denote by $\xi^N$ the vector $(\xi_1,\dots,\xi_N)$, or the vector $(\xi_1,\dots,\xi_N,0,\dots)$.  For a complex
vector $\xi$ and $s\in\ZZ$ we denote  $\xi^s=\prod_j \xi_j^{s_j}$. 

\noindent
  {\it Norms.} We use $|\cdot|$ to denote the Euclidean norm in $\R^d$ and in 
$\C\simeq\R^2$, as well as the $\ell_1$-norm in $\Z^\infty_0$. For the norms
$|\cdot|_{h^m}$ and $|\cdot|_{h^m_I}$
 see \eqref{vnorm} and below that.

\noindent
{\it Scalar products.}  The notation  ``$\cdot$'' stands for  the
scalar product in $\Z^\infty_0$, the paring of $\Z^\infty_0$ 
with $\Z^\infty$, the Euclidean scalar product in $\R^d$ and in $\C$. The latter means that if $u,v\in\C$, then 
$u\cdot v=\Re(\bar u v)$.  The $L_2$-product is denoted $\lan \cdot, \cdot\ran$.

\noindent
{\it Max/Min.} We denote $a\vee b=\max(a,b)$, $a\wedge b=\min(a,b)$. 
\bigskip          

\medskip\par
\noindent{\it Acknowledgments.} We wish to thank for discussions and advice 
Sergey Nazarenko, Anatoli Neishtadt  and Vladimir Zeitlin. 
 This work was supported by ANR through the grant STOSYMAP (ANR 2011 BS01 015 01).

\section{Preliminaries}\label{s1}
Since in this work 
we are  not interested in the dependence  of the
results on  $L$, 
 from now on it will be kept fixed and equal to 1, apart from Section~\ref{s2.5}. There 
 we make explicit calculations, controlling how their results depend 
  on $L$.
\subsection{Apriori  estimates.}\label{s1.1}
In this section we discuss preliminary properties of solutions for \eqref{1.111}. We found it
convenient to parametrise the vectors from the 
 trigonometric basis $\{e^{i{ \bk}\cdot x}\}$ by natural 
numbers and to normalise them.
 That is, to use the basis $\{e^j(x), j\ge 1 \}$, where $e^j(x) = (2\pi )^{-d/2} e^{i{ \bk}\cdot x}$,
 $\bk=\bk(j)$.  The functions $e^j(x)$ 
are eigen--vectors of the Laplacian,  $-\Delta e^j=\lambda_j e^j$, so
ordered that $0=\lambda_1< \lambda_2\le \ldots$.  Accordingly eq.~\eqref{1.111} reads 
\begin{equation}\label{1.1}
\begin{split}
\dot u+i\nu^{-1}\big(-\Delta u\big)= -f(-\Delta) u- i\gi|u|^{2q_*}u
+\frac{d}{d\tau}\sum_{j=1}^\infty b_j\bb^j(\tau)e^j(x)\ ,
% u=u(\tau,x),\quad x\in \T^d\ ,
\end{split}
\end{equation}
$u=u(\tau,x)$, where 
$
f(-\Delta) e^j= \lla_je^j$ with $ \lla_j= f(\lambda_j). 
$
The processes $\bb^j = \beta^j + i\beta^{-j}  , j\ge 1$, are
standard independent complex Wiener processes. 
 The real numbers $b_j$ are such that  for a suitable sufficiently large 
$r \in\N$ (defined below in \eqref{rr}) 
 we have 
$$
B_r:=2\sum_{j=1}^\infty \lambda_j^r
b_j^2  < \infty .
$$

By $\cH^r$, $r\in \R$, we denote the Sobolev space $\cH^r=H^r(\T^d, \C)$,
%of order $r$,formed by $2\pi$-periodic 
% $H^r([0,2\pi]^2,\C)$functions, % 
 regarded as a real Hilbert space,
 and denote by $\lan \cdot,\cdot \ran$ 
 the real $L^2$--scalar product on $\T^d$ . We provide $\cH^r$  with the  norm $\|\cdot \|_r$, 
$$
\left\| u\right\|_r^2=\sum_{j=1}^\infty|u_j|^2  ( \lambda_j\vee 1)^r\quad
\mbox{for } u(x)=\sum_{j=1}^\infty u_j e^j(x)\  .%\quad
                                %\left\|u\right\|_0 = \|u\|\ , 
$$
%where $\langle a \rangle=$max$\,(a,1)$. 
%and with the corresponding inner product  $\lan \cdot,\cdot \ran_r$.
 We abbreviate $\cH^0=\cH$ and $\|\cdot\|_0=
\|\cdot\|$ (so $\|\cdot\|$ is the $L_2$-norm). 

Let $u(t,x)$ be a solution of \eqref{1.1} such that
$u(0,x)=u_0$. It satisfies standard a-priori estimates which we now
discuss, following \cite{K12}. Firstly,  for a suitable $\eps_0>0$, uniformly in $\nu>0$
one has
\begin{equation}\label{1.3}
\E   e^{\eps_0 \|u(\tau)\|^2} \leq C(B_0,\|u_0\|) \quad \forall \tau\geq 0\,.
\end{equation}
 Assume that 
\begin{equation}\label{2.4}
q_*<\infty\;\;\;\text{if}\;\;d=1,2,\qquad
q_*<\frac{2}{d-2}\;\;\;\text{if}\;\;d\ge3\ . 
\end{equation}
Then, the following bounds on the Sobolev norms of the solution hold for each $2m\le r$ and every
 $n$:
\begin{equation}\label{2.5}
 \begin{split}
\E\left(\sup_{0\le\tau\le T}  {\|u(\tau)\|}_{2m}^{2n}+
\int_0^T{\|u(s)\|}^2_{2m+1} { \|u(s)\|}_{2m}^{2n-2}ds\right)\\ 
\le {\|u_0\|}_{2m}^{2n}+
C(m,n,T)\big(1+\|u_0\|^{c_{m,n}}\big),
\end{split}
\end{equation}
 \begin{equation}\label{2.05}
 \E\,
{\|u(\tau)\|}_{2m}^{2n} \le  C(m,n)\qquad  \forall\,\tau\ge0, 
\end{equation}
 where $C(m,n,T)$  and $C(m,n)$  also depend on
 $B_{2m}$.
 \medskip
 
  Estimates \eqref{2.5}, \eqref{2.05} are assumed everywhere in our work.  As we have explained, 
 they are fulfilled under the assumption  \eqref{2.4}, but if  the function $f(t)$ grows super-linearly, then the 
 restriction  \eqref{2.4} may be weakened.
  \medskip
 
 Relations \eqref{2.5} in the usual way
 (cf. \cite{Hai01b,KS04J, Od06, Sh06})
 imply that eq.~\eqref{1.1} is {\it regular in the space
   $\cH^{2m}\cap L_{2q+2}$} in the sense that for any $u_0\in
 \cH^{2m}\cap L_{2q+2}$ it has a unique strong solution $u(t,x)$,
 equal $u_0$ at $t=0$, and satisfying estimates \eqref{1.3},
 \eqref{2.5} for any $n$. By the Bogolyubov-Krylov
 argument, applied to a solution of \eqref{1.1}, starting from the origin at $t=0$,
  this equation has a stationary measure $\mu^\nu$, supported
 by the space $\cH^{2m}\cap L_{2q+2}$, and a corresponding stationary
 solution $u^\nu(\tau)$, $\cD u^\nu(\tau)\equiv\mu^\nu$, also
 satisfies \eqref{1.3} and \eqref{2.05}.

\subsection{Resonant averaging  }\label{s3.1}
 For a vector $v\in\R^N$, $N\le\infty$,
we will denote
$\ 
|v|_1=\sum |v_j| \le\infty.
$ 
Given a vector ${W}\in\R^n$ , $n\ge1$, and a positive
integer $m$, we  call the set 
%l $\cA\subset \Z^n$ its {\it
  %resonant set  of  order $m$}, where 
\begin{equation}\label{4.2}
\cA=\cA(W,m):=\{ s\in \Z^n: |s|_1\le m, \,{W}\cdot s=0\}\ 
\end{equation}
 %the {\it  resonant set  of $W$  order $m$}. 
the {\it set of resonances for $W$ of order $m$}. 
We denote 
%by $r$ the rank of $\cA$, denote
 by  $\cA^Z$ the $\Z$-module, generated by 
$\cA$ (called the {\it  resonance module}),  denote its rank by $r$ 
 and set  $\cA^R=\,$span$\,\cA$ (so dim$\,\cA^R=r$). 
Here and everywhere below the finite-dimensional vectors
are regarded as column-vectors
 and ``span" indicates the linear envelope over real numbers.

The following fundamental lemma provides the space  $\cA^R$ with a very convenient integer basis.
For its  proof see, for example,  \cite{bour}, Section~7:
\begin{lemma}
There exists a system $\zeta^1,\ldots,\zeta^n$ of integer vectors in
$\Z^n$ such that
span$\,\{\zeta^1,\ldots,\zeta^r \}= \cA^R$,  and the
$n\times n$ matrix $R=(\zeta^1 \zeta^2\dots\zeta^n)$
% formed by the columns $\zeta^{j}$, $j=1,\ldots, n$, 
 is unimodular (i.e., $\det R=\pm1$).
\end{lemma}

The $\Z$-module, generated by $\zeta^1,\dots,\zeta^r$, contains $\cA^Z$ and 
 may be bigger  than $\cA^Z$, but the factor-group of the former by the latter always
  is finite. 
 \smallskip

We will write  vectors $y \in \R^n$ as $y=  \left(\begin{array}{c}\!\!y^I\!\! \\ \!\!y^{II}\!\!\end{array}\right), %(y^I,y^{II})^T,
  y^I\in\R^r, y^{II}\in\R^{n-r}$.
Then clearly $\cA^R=
\Big\{R   \left(\begin{array}{c}\!\!y^I\!\! \\ \!\!0\!\!\end{array}\right)
%(y^I,0)^T 
 : y^I\in \R^r \Big\}$. Therefore 
\begin{equation}\label{0.0}
s\in\cA  \Rightarrow \text{ $
 R^{-1}s= \left(\begin{array}{c}\!\!y^I\!\! \\ \!\!0\!\!\end{array}\right)
$ for some $y^I\in\Z^r$  }.
\end{equation}
Since $s\in\cA^R$  implies that $ W\cdot s =0$, then also 
\begin{equation}\label{0.00}
\big\{s\in\Z^n,\;\; |s|_1\le m 
%s\in\cA  \Leftrightarrow 
\text{ and $ R^{-1}s= \left(\begin{array}{c}\!\!y^I\!\! \\ \!\!0\!\!\end{array}\right)
$ for some $y^I\in\Z^r \big\}    \Rightarrow s\in\cA$
  }.
\end{equation}

Let us provide $\R^n$ with the standard basis $\{e^1,\dots,e^n\}$, where 
$e^j_i=\delta_{ij}$, $1\le i\le n$, and consider the vectors 
\begin{equation}\label{eta}
\eta^j=(R^T)^{-1} e^j\ , \quad j=r+1,\ldots,n.
\end{equation}
Then 
\begin{equation}\label{baasis}
\text{vectors $\eta^{r+1},\dots,\eta^n$ form a basis of $(\cA^R)^\bot$.
}
\end{equation}
Indeed, these $n-r$ vectors are  linearly independent,  and 
 for each $ j>r$,
  $s\in\cA$ in view 
of \eqref{0.0} we have 
\begin{equation}\label{-11}
\lan s,\eta^j\ran = \lan s,(R^T)^{-1}e^j\ran = \lan R^{-1}s, e^j\ran=0 \,.
\end{equation}

Consider a liner mapping, dual to the space $\cA^R$:
$$
\R^n\to\R^r,\qquad \xi\mapsto (\lan\xi,Re^1 \ran,\dots, \lan\xi,Re^r \ran).
$$
It defines a mapping 
\begin{equation}\label{f0}
L_\cA: \T^n=\R^n/2\pi\Z^n \to \T^r=\R^r/2\pi \Z^r.
\end{equation}
For any $\xi\in \T^n$ consider the fiber $F_\xi$ through $\xi$, $F_\xi=L_\cA^{-1} (L_\cA(\xi))$. 
Clearly the points 
\begin{equation}\label{f1}
\xi+\sum_{j=r+1}^n \zeta_j \eta^j\in\T^n  , \qquad \zeta\in\T^{n-r}, 
\end{equation}
belong to $F_\xi$. The r.h.s. of \eqref{f1} defines an embedding of $\T^{n-r}$ to $F_\xi$. 
Indeed, if two points $\zeta^1$ and $\zeta^2$ correspond to the same point of $\T^n$, then
for $l=r+1,\dots, n$ their difference $\zeta = \sum \zeta_j \eta^j$ satisfies 
$$
0 %\stackrel{\begin{smallmatrix}\textrm{mod $2\pi$ }\end{smallmatrix}}{=}  
\overset{\text{mod}\, 2\pi}{=} \Big\lan\sum_j \zeta_j \eta^j, Re^l \Big\ran
=\sum_j \zeta_j\big\lan \eta^j,Re^l\big\ran=\zeta_l.
$$
In fact, this embedding is an isomorphism. Indeed, let $\xi^1\in F_\xi$, i.e. 
$L_\cA\xi^1 = L_\cA\xi$.  Modifying $\xi^1$ by a suitable linear combination of vectors $\{\eta_j, r+1\le j\le n\}$, 
we get a vector $\xi^2$ such that $\lan\xi^2,Re^j\ran= \lan\xi,Re^j\ran$  for all $j$. Therefore $\xi^2=\xi$ in $\T^n$,
and $\xi^1$ has the form \eqref{f1}.  That  is,
\begin{equation}\label{f2}
F_\xi = L_\cA^{-1}(L_\cA(\xi))=: \T_\xi^{n-r},
\end{equation}
where $\T_\xi^{n-r}$ is the set \eqref{f1}. The coordinate $\zeta$ is well defined on $\T_\xi^{n-r}$ modulo
translations. So the normalised Lebesgue measure  $\, \dbar\zeta$ is well defined on $\T_\xi^{n-r}$, and the translation
of a torus $\T_\xi^{n-r}$ by a vector $s\in\R^{n-r}$ is a well defined operator. 

\medskip

 For a continuous 
function $f$ on $\T^n$ we define its {\it resonant average of order $m$
with respect to the vector $W$}  as the function\footnote{
To understand this formula, consider the mapping $\T^d\to \T^d$, $\vp\to \psi=R^T\vp$.
In the $\psi$-variables the function $f(\vp)$ becomes $f^\psi(\psi)=f((R^T)^{-1}\vp)$, and the equation 
$\dot\vp=W$ becomes $\dot\psi=R^T W$. So $\dot\psi_1=\dots=\dot\psi_r=0$, and the averaging of $f^\psi$
should be 
$\int_{\T^{n-r}}f^\psi\left( \psi_1, \dots,\psi_{r+1}+\theta_{r+1},\dots, \psi_n+\theta_n
%\vp+\sum_{j=r+1}^n \theta_j\eta^j
\right) \dbar\theta,$
which equals \eqref{usred}.
}
\begin{equation}\label{usred}
\langle f \rangle_W(\vp) :=
\int_{\T^{n-r}}f\left( \vp+\sum_{j=r+1}^n \theta_j
\eta^j\right) \dbar\theta\ ,
\end {equation}
where we have set $\dbar \theta_j:= \tfrac{1}{2\pi}d\theta_j$.
The importance of the resonant averaging  is due to the {\it resonant 
version of the Kronecker-Weyl theorem}. In order to state it, for a continuous function $f$ on 
$\T^n$,  $f(\vp)=\sum_s f_s e^{is\cdot \vp}$, we define its 
 degree     as 
$ \sup_{s\in \Z^n:f_s\neq 0}{|s|_1}$ (it is finite or infinite). 

\begin{lemma}\label{l.aver}
Let $f:\T^n\to \C$ be a continuous function of  degree at most  $m$. Then  % one has 
\begin{equation}\label{aver}
\lim_{T\to\infty}\frac1{T}\int_0^T f(\vp+t{W})\,dt = \langle
f\rangle_W (\vp)\ ,
\end {equation}
uniformly in $\vp\in\T^n$. The rate of convergence  in the l.h.s. 
depends on $n,m,   |f|_{C^0} $  and  $W$.
More specifically, it depends on ${W}$ only through the quantity 
$$
\kappa({W},m)=\min\{ |s\cdot{W}|:  s\cdot{W}\ne0, |s|_1\le m\}>0
$$
(the bigger $\kappa$, the faster is the convergence). 
\end{lemma}

Note that $\kappa\ge1$ if ${W}$ is an integer vector.\smallskip

\noindent
{\it Proof.} Let us denote the l.h.s. in \eqref{aver} as $\{f\}(\vp)$. As $f(\vp)$ is a finite sum
of harmonics $f_se^{is\cdot \vp}$, where the number of the terms is $\le C(n,m)$ and 
$|f_s|\le |f|_{C^0}$ for each $s$, it suffices to  prove the assertions for  $f(\vp)=e^{is\cdot \vp}, |s|\le m$.

i) Let $s\in\cA$. Then $s\cdot W=0$ and
 $\{e^{is\cdot \vp}\}=e^{is\cdot \vp}$ since $e^{is\cdot (\vp+t{W})} \equiv e^{is\cdot \vp}$.
 By \eqref{-11},
the integrand in \eqref{usred} equals
$\ 
e^{is\cdot \vp} \prod_{j=r+1}^n e^{i\theta_j \eta^j\cdot s}= e^{is\cdot \vp} .
$
So in this case 
$\ 
\{e^{is\cdot \vp}\}=e^{is\cdot \vp}=\langle e^{is\cdot \vp}\rangle_W.
$
\smallskip

ii) Now let $s\notin \cA$.  Since $|s|_1\le m$, then $s\cdot W\ne0$. Consider
$\ 
\{e^{is\cdot \vp}\}=\lim
 \frac1{T}
\int_0^Te^{is\cdot(\vp+t{W})}\,dt 
$.
The modulus of the expression under the lim-sign is $\, \le  2( {T|s\cdot
  {W}|})^{-1}$.  Therefore $\{e^{is\cdot \vp}\}=0$ and the rate of convergence to zero
  depends only on $\kappa({W},m)$. By \eqref{0.00}, $(R^T)^{-1}s\cdot e^j\ne0$
  for some $j>r$. So the integrand in \eqref{usred} is a function $Ce^{i\xi\cdot\theta}$,
  where $\xi$ is a non-zero integer vector, and
  $\langle e^{is\cdot \vp} \rangle_W=0$. We see that in this case \eqref{aver} also holds.
\qed\medskip

The proof above also demonstrates the following result:
\begin{lemma}\label{r.aver} Let $f$ be a  finite trigonometrical polynomial $f(\vp)=\sum f_se^{is\cdot \vp}$. Then
the l.h.s. of \eqref{aver} equals $\sum_{s\cdot W=0} f_s %\delta_{0,\,s\cdot W}
\, e^{is\cdot \vp}$. If
degree of $f$ is $\le m$, then also 
\begin{equation}\label{yy}
\lan f\ran_W(\vp)=
\sum%_{|s|\le m} 
f_s\delta_{0,\,s\cdot W}\, e^{is\cdot \vp}=
\sum_{s\in \cA(W,m)} f_s\, e^{is\cdot \vp}.
\end{equation}
\end{lemma}

\subsection{Resonant averaging in  Hilbert space  }\label{s3.2}
Consider the Fourier transform for complex functions on $\T^d$ 
which we write as the
  mapping
$$
\cF: \cH \ni u(x)\mapsto v=(v_1,v_2,\ldots) \in \C^\infty\ ,
$$
defined by the relation $u(x)=\sum v_k e^k(x)$. In the space of
complex sequences we introduce the norms
\begin{equation}\label{vnorm}
\left|v\right|^2_{h^m}=\sum_{k\ge 1} |v_k|^2 (\lambda_k\vee1)^m\,, \quad m\in\R\,,
\end{equation}
and set $h^m=\{v|\, \left| v\right|_{h^m} <\infty\}$. Then 
$$
|\cF u|_{h^m}=\|u \|_m\qquad \forall\,m.
$$

For $k\ge1$ let us denote $I_k=I(v_k)=\tfrac12|v_k|^2$ and
$\vp_k=\vp(v_k)$, where $\vp(v)=\Arg v\in S^1$ if $v\ne0$ and
$\vp(0)=0\in S^1$. For any $r\ge0$ consider the
mappings\begin{equation*}%\label{5.01} 
\Pi_I:h^r\ni v\mapsto I=(I_1,I_2,\dots)\in h^r_{I+},\qquad
%\end{equation*}
%\begin{equation}\label{5.02}
\Pi_\vp:h^r\ni v\mapsto \vp=(\vp_1,\vp_2,\dots)\in\T^\infty.
\end{equation*}
Here $h^r_{I+}$ is the positive octant $\{I: I_k\ge 0 \  \forall k\}$ 
 in the space $h^r_{I}$, where 
$$
 h^r_{I}=\{I\mid | I|_{h^r_I}=2\sum_k (\lambda_k\vee1)^r|I_k|<\infty\}.
$$
Abusing a bit notation we will  write 
$$
\Pi_I (\cF(u))=I(u),\qquad \Pi_\vp(\cF(u))=\vp(u).
%\qquad (\Pi_I\times\Pi_\vp)(\cF(u))=(I\times\vp)(u). 
$$
The mapping $I:\cH^r\to h_I^r$ is  2-homogeneous 
continuous, while the mappings $\vp:\cH^r\to \T^\infty$ 
and $(I\times\vp):\cH^r\to h^r_I\times\T^\infty$ are Borel-measurable
 (the torus 
$\T^\infty$ is given the Tikhonov topology and the corresponding   Borel sigma-algebra). 

For integer vectors $s=(s_1,s_2,\dots)$  (and only for them) we will abbreviate
$$
|s|_1=|s|.
$$
We denote $
\Z_0^\infty=\{s\in\Z^\infty:  |s|<\infty\},
$
and for a vector $s=(s_1,s_2,\dots)\in \Z_0^\infty$ write
\begin{equation}\label{notation}
\Lambda\cdot s=\sum_k\lambda_ks_k, \quad \supp s=\{k: s_k\ne 0\}, \quad
\lc s\rc=\max\{k: s_k\ne0\}.
\end{equation}
Similar for  $\vp\in\T^\infty$ and $s\in \Z_0^\infty$ we write 
$\vp\cdot s=s\cdot\vp=\sum_k\vp_k s_k\in S^1=\R/2\pi\Z$. 

Let us  fix some
 $m\in\N$ and, by analogy with \eqref{4.2}, define the set  of resonances of order $m$
   for the frequency-vector
 $\ 
 \Lambda=(\lambda_1,\lambda_2,\dots)
 $ 
  as
\begin{equation}\label{resset}
\cA=
\cA(\Lambda,m)=\{ s\in \Z^\infty_0 : |s|\le m,\Lambda\cdot s=0\}\ .
\end{equation}
We have 
$$
\cA(\Lambda,m)=
\cup_n \big(\n\cA (\Lambda,m)\big),
$$
where $\n\cA(\Lambda,m)=\{s\in\cA(\Lambda,m): \lc s\rc\le n\}$. Clearly, 
\begin{equation}\label{*An}
 \n \cA(\Lambda,m)=
\cA (\Lambda^n,m), 
\end{equation}
where %$\n\cA(\Lambda,m)=\{s\in\cA(\Lambda,m): \lc s\rc\le n\}$ and 
$ \Lambda^n:=(\lambda_1,\ldots,\lambda_n)\in  \R^n $
(in the sense that the vectors in the set on the right are formed by first $n$ components
of the vectors on the left).

We order  vectors in  the set $\cA$,  that is write it as 
$\cA=\{ s^{(1)},s^{(2)},\ldots\}$, 
  in such a way %$(s^{(1)},s^{(2)},\ldots)$ 
    that 
  $\lc s^{(j_1)} \rc \le  \lc s^{(j_2)}\rc$ if $j_1\le j_2$, and for $N\ge1$ denote 
  \begin{equation}\label{NN}
 J(N)=\max\{j: \lc s^{(j)}\rc \le N\}.
\end{equation}

  For a vector $v\in h^0$ and any $n$ we denote by $v^n$ the vector $(v_1,\dots,v_n)\in \C^n$, or 
  the vector   $(v_1,\dots,v_n,0,\dots)\in h^0$. -- It will be always clear from the context which interpretation of $v^n$ 
  we use. 
  \medskip

\noindent
{\bf Definition.}  A Borel-measurable complex or real function 
  $f$ on $h^r$ is called {\it cylindric} if for some $n\ge1$ it
satisfies $f(v)\equiv f(v^n)$. It is called {\it cylindric continuous} if it is continuous
on some space $h^p$ (then it is continuous on each space $h^r$). 
\smallskip

Let $f$ be a continuous cylindric function.
We regard it as a function on  some  $\C^{n}=\{v^n\}$, denote by 
$\C^n_{*}$ the subset of $\C^{n}$ formed by vectors $v^n$ such  that $I_k(v^n)\ne0$
for each $k\le n$, and identify vectors $v\in \C^n_{*}$ with $(I,\vp)(v)\in \R_+^n\times\T^n$.
Assume that the order of $f(I,\vp)$ in $\vp$ is $\le m$.

Let $r_n$ be the rank of $\cA(\Lambda^n,m)$.   We define the $n$-dimensional resonant averaging
$\lan f(I,\vp)\ran^n_\Lambda$
of $f$, using \eqref{usred}:
\begin{equation}\label{defmedia}
\lan f(I,\vp)\ran^n_\Lambda=
 \lan f(I^n,\vp^n) \ran_{\Lambda^n}= \int_{\T^{n-r_n}}f\left(I_1,\ldots,I_n,\vp^n+
\sum_{j=r_n+1}^n\theta_j \eta^j\right)
\dbar \theta\ ,
\end{equation}
where $\eta^1,\dots,\eta^{n-r_n}$ are the integer $n$-vectors \eqref{eta}, corresponding to $\Lambda^n$.

Now let us take any $n'>n$. Since we can  regard $f$ as a function of $v^{n'}\in\C^{n'}$, then we can also
 construct the averaging   ${}_{n'}\lan f\ran_\Lambda$. But since  both  $ \lan f\ran^n_\Lambda$ and
  ${}_{n'}\lan f\ran_\Lambda$ equal the
limit in the l.h.s. of \eqref{aver}, then they coincide. So 
\begin{equation}\label{invarr}
\text{
$ \lan f\ran_\Lambda^n$ does not depend on $n$, provided that 
$f(v)\equiv  f(v^n)$. }
\end{equation}

The averaging \eqref{defmedia} can be conveniently  written in the $v$-coordinate. To see this, 
for any $\theta\in\T^\infty$ let us   denote by  $\Psi_\theta$  the linear operator in $h^0$:
$$
\Psi_\theta(v)=v',\qquad v'_k=e^{i\theta_k}v_k.
$$
 Clearly this is a unitary isomorphism of every space $h^r$.
  Noting  that
   \begin{equation}\label{Psi}
  (I\times\vp)(\Psi_\theta v)\equiv (I(v), \vp(v)+\theta)
  \end{equation}
  we re-write the r.h.s. of \eqref{defmedia} as 
  \begin{equation}\label{Aver}
 \int_{\T^{n-r_n}}   f(\Psi_{\big(\sum_{j=r_n+1}^n\theta_j \eta^j \big)}(v))
\,\dbar \theta\ ,
\end{equation}
This formula for $ \lan f\ran_\Lambda^n$ is meaningful  everywhere in $\C^n$ (not only in $\C^n_{*}$).

Finally we give the following 
\medskip

\noindent
{\bf Definition.}  Let $f(v)$ be a continuous cylindric function 
such that its order in $\vp$ is $\le m$. 
% such that $f(v)\equiv  f(v^n)$ for a suitable $n$.
 Then its     resonant average of order $m$ with respect to the vector $\Lambda$, $\lan f\ran_\Lambda(v)$, is the 
function \eqref{Aver}, where $n$ is sufficiently large. 
\medskip

We have seen that   \eqref{Aver}   is well defined  (i.e. this function will not change if we replace $n$ by a bigger $n'$).   Clearly it is continuous. 
 
 Let us take any $f(v)=f(v^n)$ as in the definition, 
  write it as a function of $(I^n,\vp^n)$, and decompose in  Fourier series in $\vp^n$,
 $f=\sum _{s\in\Z^n, |s|\le m} f_s(I^n)e^{is\cdot \vp^n}$. Then using \eqref{yy} we see  that 
  \begin{equation}\label{yyy}
  \lan f\ran_\Lambda(I^n,\vp^n)=\sum_{s\in\cA} f_s(I^n)e^{is\cdot \vp^n}\,.
  \end{equation}

 Let us denote 
\begin{equation*}  %\label{ZZ}
\ZZ=\{s\in\Z^\infty_0: s_k\ge0\;\; \forall k\},
  \end{equation*}
and consider a series on some space $h^r, r\ge0$:
  \begin{equation}\label{xx}
  F(v)=\sum_{p,q,l\in\ZZ} \CC\,, 
  \end{equation}
  where $I=I(v)$,
  $C_{pql}=0$ if $\supp q\cap \supp l\ne\emptyset$
  and for $v\in h^r$, $q\in\ZZ$ we write $v^q=\prod v_j^{q_j}$. 
  We assume that the series  converges normally in $h^r$ in the sense that for each $R>0$ we have 
    \begin{equation}\label{xxx}
    \sum_{p,q,l\in\ZZ}  |C_{pql} | \sup_{|v|_{h^r}, |w|_{h^r}\le R} |v^pw^p v^qw^l|<\infty.
    \end{equation}
    Clearly $F(v)={\bold F}(v,\bar v)$, where ${\bold F}$ is a (complex) analytic function on $h^r\times h^r$.
    Abusing language and following a physical tradition we will say that 
    {\it $F$ is analytic in $v$ and $\bar v$}. In particular, $F$ is continuous, and the series \eqref{xx}
     converges absolutely.

    Consider a monomial $f=(2I)^pv^q\bar v^l$. This is a cylindric function and $f(v)=f(v^n)$ if 
     $n\ge \lc p\rc\vee \lc q\rc\vee  \lc l\rc$. By \eqref{yyy}  we have
     $$
     \lan (2I)^p v^q\bar v^l\ran_\Lambda= 
     (2I)^pv^q\bar v^l \delta_{0, (q-l)\cdot\Lambda}\quad \text{if}\quad  |q|+|l|\le m. 
     $$ 
     %Let us arrange the series \eqref{xx} as 
     %$$
     %F(v)= \sum_{n=1}^\infty\sum_{ \lc s\rc\wedge  \lc l\rc=n } C_{sl}v^s\bar v^l. 
     %$$
     Assume that $F$ has degree $\le m$ in the sense that 
     $
     C_{pql}=0$ unless $ |q|+|l|\le m. 
     $
     If \eqref{xx}  is a finite series, then 
  $$
    \sum_{p,q,l\in\ZZ} \lan   \CC    
    \ran_\Lambda = \sum_{q-l\in\cA(\Lambda,m)} \CC=
    \sum_{(q-l) \cdot\Lambda=0} \CC
  $$
  If the series  \eqref{xx}   converges normally,  then the series in the r.h.s. above
  also does. It defines an analytic function, which we  take for a resonant average of \eqref{xx}. 
  More precisely,  denote by $\Pi^N$ the projection
  \begin{equation}\label{PiN}
  \Pi^N(v)=v^N=(v_1,\dots,v_N,0,\dots).
   \end {equation}
    Then the function 
 $F\circ\Pi^N$ is cylindric.

\noindent
{\bf Definition.}  If a function $F \in C(h^r)$ is given by a normally converging series 
\eqref{xx}, where $|q|+|l|\le m$,  then its 
     resonant average of order $m$ with respect to $\Lambda$ is the function
  \begin{equation}\label{La_aver}
  \lan F\ran_\Lambda(v)=
  \lim_{N\to\infty}\lan F\circ \Pi^N\ran_\Lambda(v)=
   \sum_{(q-l) \cdot\Lambda=0} \CC =   \sum_{q-l\in\cA(\Lambda,m)}  \CC. 
   \end {equation}
   \medskip
   
   Note that if $v=v^n$, then $\Pi^N v=v$ for $N\ge n$. So $\lan F\circ  \Pi^N\ran_\Lambda(v)$ is independent 
   from $N\ge n$, is given by \eqref{defmedia} and equals $\lan F\ran_\Lambda(v)$. 
   
\begin{lemma}\label{l.aaver}
Let the function $F$ in  \eqref{xx}  satisfies \eqref{xxx} and has degree $\le m$.
 Then  for each $v\in h^r$ we have 
\begin{equation}\label{aaver}
\lim_{T\to\infty}\frac1{T}\int_0^T F (\Psi_{t\Lambda} v)\,dt = \langle
F\rangle_\Lambda (v)\  .
\end {equation}
\end{lemma}

The assertion follows from Lemma \ref{l.aver} and the normal convergence of series \eqref{xx}.  

Note that in view of \eqref{La_aver}
%if $F$ is a cylindric function, $F(v)=F(v^n)$, and $I_j=I(v_j), \vp_j=\vp(v_j)$, then 
\begin{equation}\label{AAver}
\lan F\ran_\Lambda \text{ is a function of $I_1,I_2\dots$ and the variables 
$\{s\cdot \vp, s\in \cA(\Lambda,m)\}$}.
\end {equation}

We will use a trivial generalisation of the construction above. Namely, let $O$ be a domain
in $h^r_I$, $g(I)$ be a bounded continuous function on $O$ and $F(v)$ be a function as above. 
Consider $F_g(v)=g(I(v))\, F(v)$.  We set $\lan F_g\ran_\Lambda = g\,\lan F\ran_\Lambda$. 
Then Lemma~\ref{l.aaver} remains true for $F_g$.

\section{Averaging theorems for equation \eqref{1.1}.}\label{s5}
Everywhere below $T$ is a  fixed positive number.

\subsection{Equation \eqref{1.1} in the $v$-variables.}\label{s5.1}

Let us pass in eq. \eqref{1.1} with $u \in\cH^r, \ r>d/2$,  to the $v$-variables,
$v=\cF(u)\in h^r$:
 \begin{equation}\label{5.1}
 \begin{split}
dv_k+i\nu^{-1}\lambda_kv_kd\tau=P_k(v)\,d\tau +\,
b_kd\bb^k(\tau),\quad k\ge1;\qquad v(0)=\cF(u_0)=:v_0. 
\end{split}
\end{equation}
Here
\begin{equation}\label{5.00}
P_k=P_k^1+P_k^0,
\end{equation}
where $P^1$ and $P^0$ are, correspondingly, the linear and nonlinear  hamiltonian 
parts of the perturbation.  So $P^1_k$ is the Fourier-image of $-f(- \Delta)$, 
$P^1_k=\,$diag$\,\{-\lla_k,k\ge1\}$, while the  operator $P^0$ is the vector-field
$u\mapsto -i\rho|u|^{2q_*}u$, written in the $v$-variables. I.e., 
 $$
 P^0(v)=-i\gi \cF (|u|^{2q_*}u)\,, \;\; u=\cF^{-1}(v).
$$
Every its component $P^0_k$ is a sum of monomials:
\begin{equation}\label{P^0}
P_k^0(v)=\sum _{ {p,q,l\in\Z_{+0}^\infty }}  C_k^{pq l}(2I)^p
v^q{\bar v}^l = \sum _{ {p,q,l\in\Z_{+0}^\infty }}  P_k^{0pql}(v), 
\qquad k\ge1, 
\end{equation}
where % $\Z_{+0}^\infty=\{s\in \Z_{0}^\infty: s_j\ge0\; \forall\,j\}$,  $v^s=\prod_j v_j^{s_j}$ and 
$$
\text{ $C_k^{pql}=0$  unless $2|p|+|q|+|l|=2q_*+1$ and $|q|=|l|+1$\ .
%$|s|=q+1$ and  $|l|=q$.
}
$$
%Moreover, $C^{pql}_k$ depends on $l$ and $k$ only through $l':=l+e_k$,
%i.e., one has
%\begin{equation}\label{P^0.1}
%C_k^{pql}= C^{pql'}\ , \quad \mbox{with } \bar{C}^{pql}= - C^{plq}\ .
%\end{equation}

In Section \ref{s2.5} we  check  that the series for the functions $P_k^0$
% satisfy \eqref{xxx}. 
are  normally converging. 
It is straightforward that  this is a function of $\vp=(\vp_j, j\ge1)$, of order $2q_*+1$, and that 
 the mapping $P^0$ is analytic of polynomial growth:

\begin{lemma}\label{l.P^0}
The nonlinearity $P^0$ defines a real-analytic transformation of 
$ h^r$  if  $r>\tfrac{d}{2}$.
 The mapping
 $P^0(v)$ and its differential $dP^0(v)$ both have  polynomial growth in $|v|_{h^r}$.
 %The decompositions \eqref{P^0}  satisfy the assumptions of Lemma~\ref{l.an_aver}.
\end{lemma}
%The proof of the assertions is straightforward. 
\smallskip

 We will refer to equations \eqref{5.1} as to the {\it
  $v$-equations.}
We fix
 $$
 m=2q_*+2,
 $$
  and find the  set of  resonances or order $m$ 
$$
\cA=\cA(\Lambda,m)= (s^{(1)},s^{(2)},\ldots)$$
 (see \eqref{resset}). 
  For any $s^{(j)}\in \cA $ consider the corresponding resonant combination of phases $\vp(v)$:
   \begin{equation}\label{f-j}
\Phi_j: h^0\to S^1,\quad v\mapsto 
s^{(j)}\cdot \vp(v),
\end{equation}
and   introduce the Borel-measurable   mapping 
$$
\Pi_\Phi:h^r\ni v\mapsto \Phi=(\Phi_1,\Phi_2,\dots)\in  S^1\times
S^1 \times \cdots=:\cT^\infty\ .
$$
 We will also   write
$$
\Pi_\Phi(\cF(u))=\Phi(u),\qquad
(\Pi_I\times\Pi_\Phi)(\cF(u))=(I\times\Phi)(u). 
$$
Note that the system $\Phi$ of resonant combinations is highly over-determined:
there are many linear relations between its components $\Phi_j$. 

Let us pass in eq. \eqref{5.1} from the complex variables $v_k$ to
the action-angle  variables
$I_k\ge0,\  \vp_k\in S^1$:
 \begin{equation}\label{5.2}
dI_k(\tau)
=(v_k\cdot P_k)(v)\,d\tau + b_k^2\,d\tau  +b_k(v_k\cdot
d\bb^k)
\end{equation}
(here $\cdot$ indicates the real scalar product in  $\C\simeq\R^2$), 
 and 
 \begin{equation}\label{5.3}
 \begin{split}
d\vp_k(\tau)=\Big(-\nu^{-1}\lambda_k+
|v_k|^{-2}(iv_k\cdot P_k(v))
\Big)\,d\tau+|v_k|^{-2}b_k(iv_k\cdot d\bb^k)
%&=:(\nu^{-1}\lambda_k +G_k(v))\,d\tau +
%|v_k|^{-1}b_k\left(\frac{iv_k}{|v_k|}\cdot d\bb^k(\tau)\right). 
%\,d\beta^l.
\end{split}
\end{equation}
(to derive \eqref{5.3} it is better to use the complex Ito formula
 from Lemma~\ref{l.a1}, with $f(v)= \frac1{2i}\log\frac v{\bar v}=\vp(v)$). 
The equation for the actions are slow, while equations for the angles are fast since $d\vp_k \sim\nu^{-1}$. 
The function 
$|v_k|^{-2}(iv_k\cdot P_k)$
 is singular at $v_k=0$. By Lemma~\ref{l.P^0} it 
 meets the  estimate 
\begin{equation}\label{5.6}
\big||v_k|^{-2}(iv_k\cdot P_k(v))
 \chi_{\{|v_k|>\delta\}} \big|\le \delta^{-1}  Q_k(|v|_{h^r}),
\end{equation}
where $Q_k$  is a polynomial.  Since  for  any vector $v$ we have 
$
|v-v^M|_{h^{r-1/3}}\le CM^{-\frac1{3d}} |v|_{h^r} 
$
(because  $\lambda_k\sim |k|^{2/d}$), then
\begin{equation*}  %\label{5.5}
|P(v)-P(v^M)|_{h^{r-2-1/3}}\le M^{-\frac1{3d}}  Q(|v|_{h^r}),
\end{equation*}
where $Q$ is a polynomial.

As $\|u\|_r=|v|_{h^r}$, then relations \eqref{2.5}, \eqref{2.05} imply upper estimates for solutions $v(\tau)$
of \eqref{5.1}.  
Repeating for equations \eqref{5.1} and \eqref{5.2}
the argument from Section~7 in \cite{KP08} (also see Section~6.2 in \cite{K10} and Remark in 
Section~\ref{s.pr_couple} below), we get low bounds for the norms of  the components
$v_k(\tau)$ of $v(\tau)$:

\begin{lemma}\label{l5.1} Let $v^\nu(\tau)$ be a solution of \eqref{5.1} and $I^\nu(\tau)=I(v^\nu(\tau))$. Then for 
 any $k\ge1$ the following convergence holds uniformly in $\nu>0$:
\begin{equation}\label{5.8}
\int_0^T\IP\{I^\nu_k(\tau)\le\delta\}\,d\tau\to0\qquad \text{as}\;\; \delta\to0
\end{equation}
(the rate of the convergence  depends on $k$). 
\end{lemma}

We emphasise that this lemma does not imply 
that for a given $k$  we have 
$\PP \{I_k^\nu(\tau)=0$ for some $0\le\tau\le T\}=0$ (if it did, the proof in Section \ref{s3}
below would simplify significantly). 

\subsection{Resonant monomials and combinations of phases.
}  \label{s22}

Due to \eqref{5.3}
the resonant combinations  $\Phi_j$ of  angles (see \eqref{f-j}) 
satisfy slow equations: 
\begin{equation}\label{5.ris}
 \begin{split}
d\Phi_j(\tau)=\sum_{k\ge 1}s^{(j)}_k\Bigl(|v_k|^{-2}(iv_k\cdot
P_k)\,d\tau+ |v_k|^{-2}b_k(iv_k\cdot d\bb^k)\Bigr), \quad  j\ge1. 
%&=\sum_{k\ge 1}s^{(j)}_k\Bigl(G_k(v)\,d\tau +
%|v_k|^{-1}b_k\left(\frac{iv_k}{|v_k|}\cdot
%d\bb^k(\tau)\right)\Bigr)\,, \quad j\ge1. 
%\,d\beta^l.
\end{split}
\end{equation}
Now we define and study corresponding resonant monomials of $v$. 

For any $s\in\Z_0^\infty$, vectors $s^+, s^-\in\ZZ$  such that $s=s^+-s^-$  and 
$
\supp s=\supp s^+\cup \supp s^-$,  $\supp s^+\cap \supp s^-=\emptyset$
are uniquely 
defined.
Denote by $V^s$ the monomial 
\begin{equation}\label{hren}
V^s(v)= v^{s^+}\cdot \bar v^{s^-}=
\prod _j v_j^{s_j^+}  \cdot \prod_j \bar v_j^{s_j^-}.
\end{equation}
This is a real-analytic function on every space $h^l$, and $\vp \big(V^s(v)\big)=\sum_j s_j \cdot \vp(v_j)$. 
{\it Resonant monomials} are the functions 
$$
V_j(v)= V^{s^{(j)}}(v),  \qquad j=1,2,\dots. 
$$
Clearly they satisfy
\begin{equation}\label{ugly}
I(V_j(v))=(2I)^{\frac12 s^{(j) } }:=\prod_l (2I_l)^{ \frac12 s_l^{(j)}},\qquad 
\vp (V_j(v))=\Phi_j(v).
\end{equation}
(It may be better to call $V_j(v)$ a {\it minimal resonant monomial}  since for any $l\in\ZZ$  the monomial 
$I^lV_j(v)$ also is resonant and corresponds to the same resonance.)

Now  consider the mapping
$$
V: h^l\ni v\mapsto (V_1,V_2,\dots)\in \C^\infty,
$$
where $\C^\infty$ is given the Tikhonov topology. It is continuous for any $l$. For $N\ge1$ denote
$$
V^{(N)}(v)=\big(V_1,\dots,V_{J}(v)\big) \in \C^{J},
$$
where $J=J(N)$, see \eqref{NN}. 

For any $s\in\Z_0^\infty$, applying the complex Ito formula (see
Appendix~\ref{a1}) to the process
 $V^s(v(\tau))$,  we get that 
 \begin{equation} \label{7.1}
 \begin{split}
 d\,V^s= V^s\Big(-i\nu^{-1}(\Lambda\cdot s) d\tau+ &\sum_{j\in\supp s^+} s_j^+
             v_j^{-1}(P_j(v)\,d\tau+b_j\,d\bb_j)\\
             +&\sum_{j\in\supp s^-} s_j^-\,
             {\bar  v_j}^{-1}(\bar P_j(v)\,d\tau+b_j\,d\bar\bb_j)\Big).
 \end{split}
 \end{equation}
 If $\tilde s\in\Z_0^\infty$ is perpendicular to $\Lambda$, 
  then the first term in the r.h.s. vanishes. So  $V^{\tilde s}(\tau)$ is a slow process,
  $dV^{\tilde s}\sim1$. In particular, the processes 
 $dV_j, j\ge1$, are slow. 
 
 Estimates \eqref{2.5} and equation \eqref{7.1} readily imply 
 
 \begin{lemma}\label{l7.1}
 For any $j\ge1$ we have 
 $\ 
 \E\big|V_j(v(\cdot))\big|_{C^{1/3}[0,T]} \le C_j(T)<\infty,
 $ 
 uniformly in $0<\nu\le1$. 
 \end{lemma}

Let us provide the space $C([0,T];\C^\infty)$ with the Tikhonov topology, identifying it with
the space $C([0,T];\C)^\infty$. This topology is metrisable by the Tikhonov distance.  Fix any even integer $r$, 
\begin{equation} \label{rr}
r\ge \frac d2+1\,,
\end{equation} 
 and
abbreviate 
$$
h^r=h,\quad h^r_I=h_I,\quad C([0,T], h_{I+})\times 
C([0,T],\C^\infty)=:  \cH_{I,V}.
$$
We provide $\cH_{I,V}$ with Tikhonov's distance, the corresponding Borel $\sigma$-algebra and the natural
 filtration of the sigma-algebras
$\{\cF_t, 0\le t\le T\}$. 

Let us consider a solution $u^\nu(\tau)$ of eq. \eqref{1.1}, satisfying
$
u(0)=u_0,
$
denote $v^\nu(\tau)=\cF(u^\nu(\tau))$ and abbreviate
$$
I(v^\nu(\tau))=I^\nu(\tau),\quad V(v^\nu(\tau))= V^\nu(\tau)\in\C^\infty. 
$$

 \begin{lemma}\label{l7.2} 1) 
 Assume that $u_0\in\cH^r$. Then the 
  set of laws $\cD(I^\nu(\cdot), V^\nu(\cdot)), \  0<\nu\le1$, is tight in $\cH_{I,V}$. 
  
  2) Any limiting measure $\cQ$ for the set of laws in 1) satisfies
   \begin{equation} \label{apriori}
 \begin{split}
 \E^\cQ |I |^n_{C([0,T], h^r_I)} \le C_n &\quad \forall\, n\in\N,\qquad
 \E^\cQ \int_0^T|I(\tau)|_{h_I^{r+1}}d\tau \le C',\\
 & \E^\cQ e^{\e_0 |I(\tau)|_{h^0_I}} \le C^{''} \quad \forall\,\tau\in[0,T]. 
 \end{split}
 \end{equation}

  \end{lemma}
  \begin{proof} 1) 
  Due to Lemma \ref{l7.1} and the Arzel\`a Theorem, the laws of processes $V_j(v^\nu(\cdot))$,
  $0<\nu\le1$, are tight in $C([0,T],\C)$, for any $j$. Due to estimates \eqref{2.5} with $n=1$
   and since the 
  actions $I_k^\nu$ satisfy slow equations \eqref{5.2}, the laws of processes $I^\nu(\tau)$ are 
  tight in $C([0,T], h_{I+})$ (e.g. see in \cite{VF}).
  Therefore, for every $N$, 
   any sequence $\nu_\ell\to0$ contains a subsequence 
  such that the laws $\cD\big(I^\nu(\cdot), V^{(N)}(v^\nu(\cdot))\big)$ converges along it to
  a limit.  Applying the diagonal process we get another subsequence $\nu'_\ell$ such 
  that the convergence holds for each $N$. The corresponding limit is a measure $m^N$ on the 
  space 
  $
  C([0,T], h_{I+})\times C([0,T],\C)^{J(N)}.
  $
  Different measures $m^N$ agree, so by Kolmogorov's theorem they correspond to some measure $m$ on the 
  sigma-algebra, generated by cylindric subsets of the space
  $
  C([0,T], h_{I+})\times C([0,T],\C)^\infty,
  $
  which coincides with the Borel sigma-algebra for that space. 
  It is not hard to check that 
  $
  \cD(I^\nu(\cdot), V^\nu(\cdot))
  %\mid_{\nu=\nu'_j} 
  \strela m
  $
  as $\nu=\nu'_\ell\to0$. 
  This proves the first  assertion. 
  
  2) Estimates \eqref{apriori} follow from \eqref{1.3}, \eqref{2.5}, the weak convergence to $\cQ$
  and the Fatou lemma; cf. Lemma~1.2.17 in \cite{KS}. 
  \end{proof}

The variables $(I^\nu,V^\nu)(\tau)$ are slow components of the process  $v^\nu(\tau)$. 
Due to \eqref{ugly}
 the variables  $(I^\nu,\Phi^\nu)(\tau)$ make another set of slow components. Their disadvantage is 
that 
$\Phi_j^\nu$'s have singularities at $\Game$ (see \eqref{game}). 
 %when $I^\nu_k=0$ for some $k\in\supps^{(j)}$.
   Still, below we use the resonant 
combinations of angles $\Phi_j$ as well as resonant monomials $V_j$ since outside $\Game$ equations for the former are better   then equations for the latter.

\subsection{Averaged equations and effective equation. }
\label{s2.3}

  Fix  $u_0\in \cH^r$ and consider  any limiting measure $\cQ^0$ 
  for the laws 
   %as $\nu=\nu_j\to0$:
\begin{equation} \label{5.88}
\cD(I^{\nu_\ell}(\cdot), V^{\nu_\ell}(\cdot))\strela \cQ^0
\as \nu_\ell\to0,
\end{equation}
existing by Lemma \ref{l7.2}. 
Our goal is to show that the limit $\cQ^0$ does not depend on the
sequence $\nu_\ell \to0$
and develop tools for its study. We begin with writing down averaged
equations for the slow components $I$ and $\Phi$ of the process $v(\tau)$,
using the  rules of the stochastic calculus (see \cite{Khas68, FW03}), and 
formally replacing there the  usual averaging in $\vp$ by the resonant 
averaging $\lan\cdot\ran_\Lambda$.
Let us first consider the $I$-equations \eqref{5.2}. The drift in the $k$-th equation is
$$
b_k^2+v_k\cdot P_k=b_k^2+v_k\cdot P_k^1+ v_k\cdot P_k^0,
$$
where $v_k\cdot P_k^1=- 2\lla_k I_k$ and 
$
v_k\cdot P_k^0(v) =  \sum _{ {p,q,l\in\Z_{+0}^\infty }} v_k\cdot  P_k^{0pql}(v),  
$
see \eqref{P^0}. By Section \ref{s2.5} the sum converges normally, so the resonant averaging of the drift
is well defined. The dispersion matrix for eq.~\eqref{5.2} with respect to the real Wiener processes 
$(\beta^1,\beta^{-1}, \beta^2,\dots)$ is diag$\,\{b_k(\Re v_k \,\Im v_k), k\ge1\}$ (it is formed by $1\times2$-blocks). 
The diffusion matrix equals the 
dispersion matrix times its conjugated  and equals diag$\,\{b_k^2|v_k|^2, k\ge1\}$. It is independent 
from the angles, so the averaging does not change it. 
 For its symmetric square-root we take diag$\,\{b_k \sqrt{2I_k} \}$, and accordingly 
write the $\Lambda$-averaged $I$-equations as 
\begin{equation}\label{5.21}
dI_k(\tau) =%\chi_{\tau\le \theta_1^-\wedge \tau_P}\left(
\langle v_k\cdot
P_k\rangle_{\Lambda}(I,V)\,d\tau + 
b_k^2\,d\tau  +{b_k}\sqrt{2I_k}\, d\beta^k(\tau), \quad  k\ge 1. 
\end{equation}

Now consider equations \eqref{5.ris} for  resonant combinations $\Phi_j$ of the angles. The corresponding 
dispersion  matrix  $D=(D_{jk})$ is formed by $1\times2$-blocks 
$$
D_{jk} = -s_k^{(j)} b_k(2I_k)^{-1} (\Im v_k  \, -\Re v_k).
$$
Again the diffusion matrix does not depend on the angles and equals $M=(M_{j_1 j_2})$,
$M_{j_1 j_2} =\sum_k s_k^{(j_1)}s_k^{(j_2)} b_k^2 (2I_k)^{-1}$. The matrix $D^{new}$ with the entries 
$D^{new}_{jk}=s_k^{(j)}b_k(2I_k)^{-1/2}$ satisfies $|D^{new}|^2=M$, and we write the averaged equations for 
$\Phi_j$'s as 
\begin{equation}\label{5.22}
d{\Phi_j}(\tau)= %\chi_{\tau\le \theta_1^-\wedge \tau_P}
 \sum_{k \ge
  1}s^{(j)}_k\Bigl(\frac{\langle i v_k\cdot 
  P_k\rangle_{\Lambda}(I,V)}{ 2I_k}\,d\tau +\frac{{b_k}}{\sqrt{2I_k}}\,
d\beta^{-k}(\tau) \Bigr)\ , \quad  j\ge 1
\end{equation}
(we have to use here Wiener processes, independent from those in eq.~\eqref{5.21}
since the differentials  $v_k\cdot d\bb^k$ and $iv_k\cdot d\bb^k$, corresponding to the noises in equations 
\eqref{5.2} and \eqref{5.3}, are independent). 

Note that eq. \eqref{5.21} has a weak singularity at the set 
$$
\Game (h)= \cup_k \{v\in h:  v_k=0\},
$$
while  eq.~\eqref{5.22} has there a strong singularity. This feature of the averaged equations 
makes it difficult to use them as a tool to study the original equation \eqref{1.1}=\eqref{5.1}.

\medskip
Let us return to the averaged $I$-equations. Consider the component $\lan v_k\cdot P_k^0\ran_\Lambda$
 of the averaged drift:
\begin{equation*}
\begin{split}
\lan v_k\cdot P_k^0\ran_\Lambda=
&\sum_{p,q,l} \int_{\T^{n-r_n}} (v'_k\cdot P_k^{0pql}(v'))\mid_{v'= \Psi_{{\sum \theta_j  \eta^j }}  (v) } 
 \,\dbar\theta_{r_n+1}\dots \dbar\theta_n\\
& =\sum_{p,q,l}
  \int_{\T^{n-r_n}} v_k \cdot \Big ( e^{-i ({\sum   \theta_j  \eta^j_k })}
  P_k^{0pql}(\Psi_{{\sum \theta_j  \eta^j }}  (v) )
  \Big)\,\dbar\theta_{r_n+1}\dots \dbar \theta_n\\
 & =:  v_k\cdot \Big( \sum_{p,q,l}
 R^{0\, pql}_k(v)\Big),
  \end{split}
\end{equation*}
where $n=n(k,p,q, l)$, \ 
 ${\sum \theta_j  \eta^j }={\sum_{j=r_n+1}^n  \theta_j  \eta^j}$
and 
$$
 R^{0\, pql}_k(v)=  \int_{\T^{n-r_n}} e^{-i ({\sum   \theta_j  \eta^j_k })}
  P_k^{0pql}(\Psi_{ {\sum \theta_j  \eta^j}}  (v) )\,\dbar\theta.
$$

Denote
$
R^0_k(v)= \sum_{p,q,l} R_k^{0pql}(v).
$
Then 
\begin{equation}\label{xa}
\lan v_k\cdot P_k^0\ran_\Lambda= v_k\cdot R_k^0.
\end{equation}
Repeating the derivation of \eqref{La_aver} and using that $|q|+|l|\le
m-1$, 
 we see that 
\begin{equation}\label{ura}
R^0_k(v) = \sum_{\substack{p,q,l\in\ZZ \\ q-l\in\cA(\Lambda,m)+e_k\\ |q|+|l|+1\le m}} 
C_k^{pql}  (2I)^pv^q\bar v^l. 
\end{equation}
Clearly
\begin{equation}\label{ura2}
R^0_k(v)=\frac{v_k}{2I_k}
\sum_{ q-l-e_k\in\cA} 
C_k^{pql}  (2I)^pv^q\bar v^l\bar v_k=
\frac{v_k}{2I_k} R^-_k(v),
\end{equation}
where $R^-_k$ may be written as a function of $(I,V)$, and the series for $R^-_k(v)$ converges 
absolutely in bounded sets of the space $h$.

If $v=\Pi^N(v)=v^N$ and $k\le N$, 
then the terms $P_k^{0pql}(v^N)$ with $\lc p\rc\vee\lc q\rc\vee\lc l\rc>N$
vanish, and in the remaining finitely many terms we can replace $n$ by $N$ (see \eqref{invarr}). 
Therefore
\begin{equation}\label{xi}
\Pi^N
 R^{0}(v^N)=  
 \int_{\T^{N-r_N}} \Pi^N \Psi_{ - {\sum   \theta_j  \eta^j }} 
  P^{0}(\Psi_{ {\sum  \theta_j  \eta^j}}  v^N )\,\dbar\theta,\quad \forall\, N. 
\end{equation}

Now we set 
$$
R=R^0+R^1,
$$
where 
\begin{equation}\label{R1}
R^1=P^1=\text{diag}(-\lla_k, k\ge1)
\end{equation}
(i.e., $R^1_k(v)=-\lla_kv_k$). 
Since 
$\ 
\lan v_k\cdot P^1_k\ran_\Lambda=\lan-\sum 2\lla_kI_k\ran_\Lambda=v_k\cdot P^1_k
=v_k\cdot R^1_k,
$
then in view of \eqref{xa} we have 
\begin{equation}\label{R_k}
\lan v_k\cdot P_k\ran_\Lambda(v)=v_k\cdot R_k(v). 
\end{equation}
For further usage we note that by the same argument we have 
$\lan iv_k\cdot P_k^0\ran_\Lambda= iv_k\cdot R_k^0$ and 
$\lan iv_k\cdot P_k^1\ran_\Lambda=0= iv_k\cdot R_k^1$. So also 
\begin{equation}\label{R_kk}
\lan i v_k\cdot P_k\ran_\Lambda(v)=iv_k\cdot R_k(v). 
\end{equation}

The vector field $R^0$ defines locally-Lipschitz operators in the spaces $h^p$, $p>d/2$:
 \begin{equation}\label{p4}
|R^0(v)-R^0(w)|_{h^p}\le C_p\big( |v|_{h^p}\vee |w|_{h^p}\big)^{2q_*}|v-w|_{h^p}.
  \end{equation}
  Indeed, in view of \eqref{xi}, for any $v,w$ such that $|v|_{h^p}, |w|_{h^p}\le R$ we have
    \begin{equation}\label{p44}
    \begin{split}
|  \Pi^N(
 R^{0}(v^N) &- R^0(w^N))|_{h^p} \\
& \le 
 \int_{\T^{N-r_N}} \Big
 | \Psi_{ - {\sum   \theta_j  \eta^j }} \big(
  P^{0}(\Psi_{ {\sum  \theta_j  \eta^j}}  v^N ) - P^0(\Psi_{ {\sum  \theta_j  \eta^j}}  w^N )\big) \Big|_{h^p}
  \,\dbar\theta.
  \end{split}
   \end{equation}
   Since $P^0(v)=-i\rho\cF(|\hat v|^{2q_*}\hat v)$, where $\hat v=\cF^{-1}v$, then 
  denoting 
  $  \Psi_{ {\sum  \theta_j  \eta^j}}  v^N =v^N_\theta\,$
  etc, and using that the operators $\Psi_\theta$ define isometries of $h^p$, we bound the r.h.s. 
  of \eqref{p44}   by
  
   \begin{equation*}
   \begin{split}
    \int_{\T^{N-r_N}} \big| P^0(v^N_\theta) - P^0(w^N_\theta)\big|_{h^p}d\theta =
    \rho  \int_{\T^{N-r_N}}  \big\||  \widehat {v^N_\theta}|^{2q_*} \widehat{ v^N_\theta} - 
    | \widehat{ w^N_\theta}|^{2q_*} \widehat{ w^N_\theta}   \big\|_p \,\dbar \theta
   \\
  \le \rho C_p R^{2q_*} \|\widehat{v^N_\theta} - \widehat{w^N_\theta}\|_p 
   \le \rho C_p R^{2q_*}|v-w|_{h^p}.
  \end{split}
  \end{equation*}
   Passing to the limit as $N\to\infty$ 
  using the Fatou lemma   we get  \eqref{p4}.

The vector field
$R^0$  is hamiltonian:
\begin{equation}\label{p2}
R^0=i \rho\nabla \cH^{\text{res}}(v),
\end{equation}
where $\cH^{\text{res}}(v) = \lan \cH\ran_\Lambda(v)$  and $\cH$ is the Hamiltonian  \eqref{*ham}. Indeed, 
since $P^0(v)=i\rho\nabla \cH(v)$, then
  denoting, as usual, 
   $R^{0N}=(R^0_1,\dots, R^0_N)$  we have
  \begin{equation*}
\begin{split}
 R^{0N}(v^N)&=  \int_{\T^{N-r_N}}  \Psi_{- {\sum  \theta_j  \eta^j}}   \Big(
  %e^{-i ({\sum   \theta_j  \eta^j })}
  i\rho \nabla
  \cH
  (\Psi_{ {\sum  \theta_j  \eta^j}}  (v^N) )\Big)\,\dbar\theta\\
  &= i\rho \nabla_{v^N}  \int_{\T^{N-r_N}}  \cH
  ( \Psi_{ {\sum  \theta_j  \eta^j}}  (v^N) )\,\dbar\theta
  =i\rho \nabla_{v^N}  \cH^{\text{res}}(v^N),
  \end{split}
\end{equation*}
as $\Psi_\theta^*\equiv\Psi_{-\theta}$. 
So \eqref{p2} is valid for $v$ replaced by $v^N$, for any $N$. By continuity it also holds for $v$.

Since the Hamiltonian $\cH^{\text{res}}$ is obtained by averaging, then $\cHR$ and the 
corresponding
vector field $R^0$ have many symmetries, given by various rotations $\Psi_m, m\in\R^\infty$:

  \begin{lemma}\label{l.symm}
  i) Let  ${\mathbf 1}=(1,1,\dots)$.  Then $\cHR(\Psi_{t{\mathbf 1}}v)=\,$const  (i.e., it does
  not depend on $t$); 
  
  ii) $\cHR(\Psi_{t\Lambda}v)=\,$const;
  
  iii) if $s^\perp = s^{\perp N}\in\R^\infty$ is a vector, orthogonal to the 
   set ${}_M\cA(\Lambda,m)$ (see \eqref{*An}),  the $\cHR(\Psi_{ts^\perp}v)=\,$const. 
   \end{lemma}
\begin{proof} By continuity, it suffices to prove i) - iii) for any $N\in\N$ large enough and 
every $v=v^N$. So below we assume that $v=v^N$, where $N\gg1$.

i) We have 
$$
\cHR(\Psi_{t{\mathbf 1}}v^N)=
\int_{\T^{N-r_N}} \cH\big( \Psi_{\sum\theta_j\eta^j}(\Psi_{t{\mathbf 1}}v^N)\big)\,\dbar\theta=
\int_{\T^{N-r_N}} \cH\big(  \Psi_{t{\mathbf 1}}(
\Psi_{\sum\theta_j\eta^j}v^N) 
\big)\,\dbar\theta.
$$
Let us denote  $\Psi_{t{\mathbf 1}}(\Psi_{\sum\theta_j\eta^j}v^N)  = v^N(t;\theta)$. Then 
$
(d/dt) v^N(t;\theta)=iv^N.
$
The flow of this hamiltonian  equation commutes with that of the equation with the Hamiltonian $\cH$. So
$\cH(v^N(t;\theta))$ is independent from $t$ for each $\theta$, and i) follows.

ii) As above, we have to check that the integral
$
\int_{\T^{N-r_N}} \cH\big(  \Psi_{t{\Lambda}}(
\Psi_{\sum\theta_j\eta^j}v^N) 
\big)\,\dbar\theta
$
does not depend on $t$. The argument of $\cH$ equals $\Psi_{(t\Lambda^N+\sum\theta_j\eta^j)}v^N$.
Since clearly, $\Lambda^N\in({}_N\cA)^\perp$, then by \eqref{baasis} $\Lambda^N=\sum\theta_j^0\eta^j$. 
Therefore the integral equals
$$
\int_{\T^{N-r_N}} \cH\big(  
\Psi_{(\sum
(\theta_j+ t\theta_j^0) \eta^j )}v^N) 
\big)\,\dbar\theta.
$$
It is $t$-independent, and ii) follows.

iii)   Assuming that $N\ge M$, we have  ${}_M\cA\subset {}_N\cA$. So by 
\eqref{baasis} $s^\perp = \sum \theta'_j\eta^j$, and the assertion follows by the same argument as in ii).
  \end{proof}

  Noting that for any $m\in\R^\infty$ the transformations $\Psi_{tm}$, $t\in\R$, form the flow of the hamiltonian 
  equation with the Hamiltonian $\tfrac12\sum m_j|v_j|^2$, we recast the assertions of the last lemma as follows:
  \begin{equation}\label{integrals}
  \{\cHR,H_0\}=0,\quad  \{\cHR,H_1\}=0,\quad  \{\cHR,H_{s^\perp}\}=0\;\; \forall\,s^\perp
  =s^{\perp M}\in({}_M\cA)^\perp. 
  \end{equation}
  Here $\{\cdot,\cdot\}$ signifies the Poisson bracket and 
  \begin{equation*}
  \begin{split}
  &H_0(v)=\tfrac12\sum|v_j|^2 = \tfrac12 |v|^2_{h^0},
  \quad H_{s^\perp}(v)= \tfrac12 \sum_{j=1}^M  s_j^\perp |v_j|^2,\\
&  H_2(v)=\tfrac12\sum \lambda_j|v_j|^2
%=\tfrac12 |v|^2_{h^1} - \tfrac12 |v_0|^2=
 =  \tfrac12\int_{\T^d} |\nabla u(x)|^2dx\,.
   \end{split}
   \end{equation*}
   Since the Hamiltonians above depend only on the actions $\tfrac12|v_j|^2=I_j$, they commute. So the Hamiltonian
  $\cHR$ has infinitely many commuting quadratic integrals of motion. Two of them, $H_0$ and $H_2+H_0$, 
  are nonnegative and coercive.
  
  Finally, since the transformations $\Psi_{m},m\in\R^\infty$, are symplectic, then the symmetries above also preserve the 
  hamiltonian vector field $R^0$. In particular,
  \begin{equation}\label{invar}
  \Psi_{-s^\perp}R^0(\Psi_{s^\perp}v)=R^0(v)\qquad 
  \forall\,s^\perp   =s^{\perp M}\in({}_M\cA)^\perp. 
  \end{equation}
  
\medskip

Motivated by the averaging theory for equations without  resonances  in \cite{K10, K12}, we
consider the following {\it  effective equation} for the slow dynamics in
 eq.~\eqref{5.2}:
\begin{equation}\label{5.eff}
d v_k= R_k(v) d \tau +{b_k}d \bb^k\ ,
\qquad k\ge1. 
%R_k=\sum_{p,s,l\in\ZZ} R_k^{0\,psl}.
\end{equation}
In difference  with the averaged equations \eqref{5.21} and \eqref{5.22}, the effective equation is  regular,
i.e. it does not have  singularities at $\Game(h)$.  Since  $R^0: h\to h$ is locally Lipschitz,  then strong solutions for \eqref{5.eff} exist locally in
  time and are unique:
 
  \begin{lemma}\label{l.uniq} 
  A strong solution of eq. \eqref{5.eff} with a specified initial data $v(0)=v_0\in h$ is unique, a.s.
 \end{lemma}

Since the transformations  $\Psi_{m},m\in\R^\infty$, obviously preserve the vector field $R^1$ as well as  the law of 
the random force in \eqref{5.eff} (see the proof of the lemma below), 
then those  $\Psi_m$ which are symmetries of $R^0$ (equivalently,  which are 
symmetries of the Hamiltonian $\cHR$), preserve weak solutions of \eqref{5.eff}:

  \begin{lemma}\label{l.invar}
  If $v(\tau)$ is a solution of equation \eqref{5.eff} and $m\in\R^\infty$ be either a vector $m=t{\mathbf 1}, t\in\R$,
  or a vector $m=t\Lambda, t\in\R$, or  $m=m^M\in ({}_M\cA)^\perp$, 
  then $\Psi_{m}v(\tau)$ also is a weak solution.
   \end{lemma}
   \begin{proof}
   Denote $\Psi_{m}v(\tau)=v'(\tau)$. Applying $\Psi_{m}$ to eq.~\eqref{5.eff}, using Lemma~\ref{l.symm}
   (cf.     \eqref{invar})
 and the invariance of the operator $R^1$ with respect to $\Psi_{m}$,  we get
 $$
 dv'_k = \big(\Psi_{m} R(v(\tau)\big)_kd\tau + e^{im_k} b_kd\bb^k= (R(v'(\tau))_k+b_k( e^{im_k} d\bb^k).
 $$
 So $v'(\tau)$ is a weak solution of \eqref{5.eff}, where the set of standard independent Wiener processes 
 $\{\bb^k(\tau), k\ge1\}$ is replaced by the equivalent set of processes $\{e^{im_k} \bb^k(\tau), k\ge1\}$.
   \end{proof}
   
    \begin{corollary}\label{c.invar}
    If $\mu$ is a stationary measure for equation \eqref{5.eff} and a vector $m$ is as in
    Lemma~\ref{l.invar}, then the measure $\Psi_{m}\circ\mu$ also is stationary. 
       \end{corollary}

 Importance of the effective equation  for the study of 
the long-time dynamics in equations \eqref{1.1}=\eqref{5.1} is clear from the next lemma:

 \begin{lemma}\label{l.ef_eq}
 Let a continuous process $v(\tau)\in h$ be a weak solution of \eqref{5.eff}. Then $I(v(\tau))$ is a 
 weak solution of \eqref{5.21}. Let stopping times $0\le\tau_1<\tau_2\le T$ and $\delta_*>0, N\in\N$ be such that
 \begin{equation}\label{stop}
 I_k(v(\tau)) \ge\delta_*\quad \text{for $\tau_1\le\tau\le\tau_2$ and $k\le N$.}
 \end{equation}
 Then the process 
 $\big(I(v(\tau)), \Phi_j(v(\tau)), j\le J(N)\big)$ is a weak solution of  the system of  averaged 
 equations \footnote{This system is heavily under-determined. }
  \eqref{5.21}, \eqref{5.22}${}_{j\le J}$.
 \end{lemma}
 \begin{proof} Let $v(\tau)$ satisfies \eqref{5.eff}. 
 % with some independent standard Wiener  processes $\bb^k$. 
  Applying Ito's formula \eqref{ito} 
  to $I_k(v(\tau))$ and $\Phi_j(v(\tau)), \ j\le J$, we get that 
  \begin{equation}\label{dIk}
 dI_k = v_k\cdot R_k\,d\tau + b_k^2\,d\tau +b_kv_k\cdot d\bb^k
 \end{equation}
 and 
 \begin{equation}\label{dPhi}
 d\Phi_j =\sum_{k\in \supp s^{(j)}} s^{(j)}_k \left(\frac{iv_k\cdot R_k}{|v_k|^2}d\tau+ 
  \frac{b_k }{|v_k|^2} iv_k\cdot d\bb^k\right).
\end{equation}
Using \eqref{R_k} and \eqref{R_kk} 
we see that \eqref{dIk} has the same drift and diffusion as \eqref{5.21}. 
So $I(v(\tau))$ is a weak solution of  \eqref{5.21} (see \cite{Yor74, MR99}). Similar, for 
$\tau\in[\tau_1,\tau_2]$, in view of \eqref{R_kk}, the process $(I,\Phi_j, j\le J)$, is a weak solution
of  the system  \eqref{5.21}, \eqref{5.22}${}_{j\le J}$.   \end{proof}

The result, presented in   the next section,  in a sense is inverse to Lemma~\ref{l.ef_eq}. Namely it implies 
that the limiting measures $\cQ^0$ as in \eqref{5.88} are weak solutions of the system of
averaged equations  \eqref{5.21}, \eqref{5.22}, and that these solutions may be lifted to solutions of
the effective equation.

\subsection{Averaging for solutions of the initial-value problem. 
}\label{s5.2}

Let $v^\nu(\tau)$ be a solution of \eqref{5.1} such that $v^\nu(0)=v_0=\cF(u_0)\in h^{r}$,
and $\cQ^0$ be a measure in $\cH_{I,V}$ as in \eqref{5.88}.

\begin{theorem}\label{t5.22}
%Let $\cQ^0$ be a measure on $\cH_{I,V}$ as in \eqref{5.88}. Then
% for any vector $\theta\in \T^\infty$ such that $\Phi(\theta) =\Phi_0$
  There exists a unique weak solution $v(\tau)$ of effective equation \eqref{5.eff}, satisfying 
 $v(0)=   v_0$  a.s., and the law of $(I, V )(v(\cdot))$ in the space $\cH_{I,V}$ coincides
  with $\cQ^0$. The convergence \eqref{5.88} holds as $\nu\to0$. 
\end{theorem}

If $\tilde s\in \Z_0^\infty$ is any vector, perpendicular to $\Lambda$ (possibly $|s|>m$), 
then $\Vts(v^\nu(\tau))$ is a slow process
(see eq. \eqref{7.1}), and the $\Lambda$-averaged equation for $\vp\big(\Vts(v^\nu(\tau))\big) =\vp(v^\nu(\tau))\cdot \tilde s$
 still depend only on $(I,V)$ (not on 
$\Vts$). Therefore one can literally repeat the argument, proving the theorem -- the only difference would be a
bit  heavier  notation --  and get the following result:
\medskip

\noindent
{\bf Amplification.} Let $\tilde s\in \Z_0^\infty$ be any vector, perpendicular to $\Lambda$. 
Then under the assumptions  of Theorem \ref{t5.22} we have the convergence
$
\cD(I,V,\Vts)(v^\nu(\cdot))\strela \cD(I,V,\Vts)(v (\cdot)).
$
\medskip

By this result, for any $\tilde s\in \Z_0^\infty$, perpendicular to $\Lambda$, and any $\tau$ 
we have
$\ 
\cD(V^{\tilde s}(v^\nu(\tau))\strela  \cD( V^{\tilde s}(v(\tau))$ as  $\nu\to0$.
Now consider $\vp(v^\nu(\tau))\cdot {\tilde s}=\vp(V^{\tilde s}(v^\nu(\tau))\in S^1$. 
Since $\vp(V)$ is a discontinuous function of $V\in \C$, then to pass to
a limit as $\nu\to0$ we do the following. We identity $S^1$ with
$\{v\in \R^2: |v|=1\}$, denote $\lc \tilde s \rc =N$,  and
 approximate the discontinuous function 
 $V^N =(V_1,\dots,V_N)
 \mapsto \vp (V^{\tilde s})$
  by continuous functions 
$$
V^N \mapsto  f_\delta([I(V^N)]) \,
\vp (V^{\tilde s}) 
 \in \R^2\,,
\qquad [I]=\min_{1\le k\le   N } I_k, \quad 0<\delta\ll1. 
$$
where $f_\delta$ is continuous, $0\le f_\delta\le 1$, $f_\delta(t)=0$ 
for $t\le \delta/2$ and $f_\delta=1$ for $t\ge \delta$.

For any measure $\mu_\tau$ in a complete metric space, which weakly continuously depends on $\tau$,
and any $\tau_1<\tau_2$ we will denote
$$
\lan\mu_\tau\ran_{\tau_1}^{\tau_2}=\frac1{\tau_2-\tau_1}
\int_{\tau_1}^{\tau_2} \mu_\tau\, d\tau.
$$
Then the argument above jointly with Lemma~\ref{l5.1} imply:

\begin{corollary}\label{c1}
Let $\ts\in\Z_0^\infty$ be any non-zero vector, orthogonal to $\Lambda$, and let $0\le \tau_1<\tau_2\le T$.
Then 
$$
\lan\,\cD(\vp(u^\nu(\tau))\cdot \ts)\,\ran _{\tau_1}^{\tau_2} \strela 
\lan\,\cD(\vp(v(\tau))\cdot \ts)\,\ran _{\tau_1}^{\tau_2} \quad\text{as}\quad \nu\to0.
%\frac1{\tau_2-\tau_1} \int_{\tau_1}^{\tau_2} \cD(\vp( v (\tau))\cdot \ts)\,d\tau 
$$
\end{corollary}

On the contrary, if $s\cdot\Lambda\ne0$, then by Proposition~\ref{r34} 
we have the convergence 
$
\lan\,\cD(\vp(u^\nu(\tau))\cdot s)\,\ran _{\tau_1}^{\tau_2}
%\frac1{\tau_2-\tau_1} \int_{\tau_1}^{t_2} \cD(\vp(u^\nu(\tau))\cdot s)\,d\tau 
\strela \dbar \vp.
$
More generally, if vectors $\ts_1,\dots,\ts_M$ from $\Z_0^\infty$ are perpendicular to $\Lambda$ and
a vector $s$ is not, then
$$
\big\lan\cD(I,\vp\cdot\ts_1,\dots,\vp\cdot \ts_M,\vp\cdot s)(u^\nu(\tau))\big\ran _{\tau_1}^{\tau_2}\strela
\big\lan\cD(I,\vp\cdot\ts_1,\dots,\vp\cdot \ts_M)(v(\tau))\big\ran _{\tau_1}^{\tau_2}\times  \dbar \vp.
$$
\smallskip

The proof of Theorem  \ref{t5.22} is  inspired by that of Theorem~3.1 in
\cite{K10} and is presented  in the next section.
In difference with \cite{K10}, 
we do not know an equivalent description of the measure $\cQ^0$ only  in terms
of the slow variables $(I,V)$ of equation \eqref{5.1}. But the following result
holds true:

\begin{proposition}\label{p.slow} Consider the natural process on the space $\cH_{I,V}$ with
the measure $\cQ^0$.  If for some  $N\in\N$ and $\delta_*>0$, 
 stopping times $0\le\tau_1<\tau_2\le T$ satisfy
 \eqref{stop}, then for $\tau\in[\tau_1,\tau_2]$ the process $\big(I, \Phi^{(N)}\big)  \big((I,V)(\tau) \big)$ is a
weak solution of the averaged 
equations \eqref{5.21} and \eqref{5.22}${}\mid_{j\le J}$. Here 
$ \Phi^{(N)}=(\Phi_1,\dots, \Phi_{J(N)})$. 
\end{proposition}

 Since    the averaged 
quantities $\langle v_k\cdot P_k\rangle_{\Lambda} $ and  $\langle iv_k\cdot
P_k\rangle_{\Lambda} $ are functions of $I$ and $\Pho$ (see \eqref{AAver}), then equations 
 \eqref{5.21} and \eqref{5.22}${}\!\mid_{j\le J}$  form  an under-determined 
 system of equations for the variables   $(I,\Pho)$.%$(I(\tau), \Phi_j(\tau), j\le J, j\ne\tilde\jmath) $.  
 
 %Let  $\ts\in\Z_0^\infty$ be a non-zero vector, orthogonal to $\Lambda$, and such that 
 %$\lc \tilde s \rc \le N$. Denote $\Phi^{\tilde s} (u(\tau))=\vp(u(\tau))\cdot \tilde s$. 
  %For the same reason  as in the Amplification above, 
 %the process $(I, \Phi^{(N)}, \Phi^{\tilde s} )(v^\nu(\tau))$

By
Theorem \ref{t5.22} the Cauchy problem for the effective equation has a weak solution.
Using 
 Lemma~\ref{l.uniq} and the Yamada-Watanabe argument  (see \cite{KaSh, Yor74, MR99})
we get that  the  equation is well posed:

\begin{corollary}\label{c2}
For any $v_0\in h^r$, eq. \eqref{5.eff} has a unique strong solution $v(\tau)$ such that $v(0)=v_0$. 
Its law  satisfies \eqref{apriori}. 
\end{corollary}

\subsection{Averaging for stationary solutions.
}\label{s.stat}
Let $u^\nu(\tau)$ be a stationary solution of eq. \eqref{1.1} as at the end of Section~\ref{s1.1}.\footnote{  Under certain 
restrictions on the equation it is  known that this solution is unique (more exactly, unique is its law).
E.g., see \cite{Sh06}, but we will not discuss this now.}  Solutions $u^\nu$  inherit the a-priori estimates 
\eqref{1.3}, 
\eqref{2.5}, \eqref{2.05}, so still the set of laws 
$\cD(I(u^\nu(\cdot)), V(u^\nu(\cdot)))$, $0<\nu\le1$, is tight in $\cH_{I,V}$ (cf. Lemma~\eqref{l7.2}).  Consider 
any limit
\begin{equation}\label{Conv}
\cD\big(I(u^{\nu_\ell}(\cdot)), V(u^{\nu_\ell}(\cdot)) \big)\strela \cQ\quad \text{as $\nu_\ell\to0$}. 
\end{equation}
As before, the measure $\cQ$ satisfies \eqref{apriori} (with the constants $C_n, C', C^{''}$,  
corresponding  to $v_0=0$).  Moreover, it is stationary in $\tau$. 
Arguing as when proving Theorem \ref{t5.22} in Section~\ref{s3}  we establish

\begin{theorem}\label{t.stat}
 There exists a stationary solution $v(\tau)$ of the effective equation \eqref{5.eff} such that 
$\cQ=\cD\big(I(v(\cdot)), V(v (\cdot)) \big)$. 
\end{theorem}

If ${s}\in \Z_0^\infty$ is such that ${s}\cdot\Lambda\ne0$, then Proposition \ref{r34} readily implies that 
\begin{equation}\label{zz}
\Big( \big( I^M\times \Phi^{J(M)}\big) \times (\vp\cdot {s})\Big)\circ \mu_\nu\strela
\Big( \big( I^M\times \Phi^{J(M)}\big)\circ\mu\Big)\times \dbar \theta
\end{equation}
(note that  since solutions $u^\nu$ and $v$ are stationary, then
$\lan \cD u^\nu(\tau)\ran_{\tau_1}^{\tau_2}= \cD u(\tau)$ and the same holds for $v(\tau)$).

\begin{theorem}\label{t.univ}
Let $u^\nu$ be a stationary solution of \eqref{1.1}, $\cD (u^\nu(\tau))\equiv \mu^\nu$, 
and assume that \eqref{5.eff} has a unique stationary measure $\mu$. Then 
 $\mu^\nu\strela\mu$ as $\nu\to0$. 
\end{theorem}
\begin{proof}
Let $\{\nu_l\}$ be the sequence in \eqref{Conv}. It suffices to verify that $\mu^{\nu_l}\strela\mu$ as 
$\nu_l\to0$. Since the stationary solutions satisfy \eqref{2.5} with $2m=r, n=1$, then 
$\int \|u\|^2_{r+1}\mu^\nu(du)\le C$ uniformly in $\nu$. So the set of measures $\mu^\nu$ is tight in $h=h^r$,
and replacing $\{\nu_l\}$ by a suitable subsequence $\{\nu_j\}$ we have
$\ 
\mu^{\nu_j}\strela \mu_*,
$
for some measure $\mu_*$.  We have to check that $\mu=\mu_*$. For this it suffices to verify that 
$\Pi^N\circ\mu=\Pi^N\circ\mu_*$, for each $N$ (see \eqref{PiN}).  Consider the operators 
$$
\Pi^N_I(v)=I^N(v), \quad \Pi^N_\vp(v)=\vp^N(v). 
$$
Then by Theorem~\ref{t.stat}, $\Pi^N_I\circ\mu = \Pi^N_I\circ\mu_*$.  Since $\mu(\Game)=\mu_*(\Game)=0$
by Lemma~\ref{l5.1},  then we can regard $\Pi^N\circ\mu$ and $\Pi^N\circ\mu_*$ as measures on $\R^N_+\times\T^N$. 
Let us take any bounded continuous function $p(I^N)\ge0$ such that 
$
\int p(I^N)\,d(\Pi^N_I\circ\mu)=1,
$
then also 
$
\int p(I^N)\,d(\Pi^N\circ\mu_*).
$
Consider two measures on $\T^N$:
$$
m=\int p(I^N)(\Pi^N_I\circ\mu)\, dI^N,\qquad  m_*=\int p(I^N)(\Pi^N_I\circ\mu_*)\, dI^N.
$$
It remains to check that they are equal (for any $N$ and any $p(I^N)$ as above).

Let us consider the mapping 
$$
L_{{}_N\cA} :\T^N\to \T^r=\{\xi\},\qquad r=r(N)
$$
(see \eqref{f0}). Its fibers are the tori $\T_\xi^{N-r}\sim \Tnr$ as in \eqref{f2}. Accordingly, we can decompose 
the measures $m$ and $m_*$
as 
$$
m=\rho^\xi(d\zeta)\,\pi(d\xi),\qquad m_*=\rho_*^\xi(d\zeta)\,\pi_*(d\xi),
$$
where $\pi=L_{{}_N\cA} \circ m$, $ \pi_*= L_{{}_N\cA} \circ m_*$
are measures on $\T^r$ and $\rho^\xi, \rho^\xi_*$ are measures on the fibers 
 $\Tnr_\xi$, $\xi\in\T^r$. 
Using again Theorem~\ref{t.stat} we get for any vector $s\in {}_N\cA$ and any bounded 
continuous function $f$ that 
$$
\int_{\T^N} f(s\cdot\vp)\big( \Pi^N_\vp\circ\mu\big)(d\vp)=
\int_{\T^N} f(s\cdot\vp)\big( \Pi^N_\vp\circ\mu_*\big)(d\vp).
$$
So
$
\pi=L_{{}_N\cA}\circ m=L_{{}_N\cA}\circ m_*=\pi_*.
$

Now consider the fibers $\Tnr_\xi$ and measures $\rho^\xi$ and $\rho^\xi_*$ on them. Recall that translations
of the tori $\Tnr_\xi$ by vectors from $\R^{N-r}$ are well defined operators. 
  By \eqref{zz} translation of the measure $m_*$ by any 
vector $s^\perp\in({}_N\cA)^\perp$ will not change it. Therefore 
$$
\rho_*^\xi(d(\zeta+s^\perp))\, \pi_*(d\xi) = \rho_*^\xi(d\zeta)\, \pi_*(d\xi) \qquad \forall\, s^\perp\in({}_N\cA)^\perp\simeq\R^{N-r}. 
$$
That is,  $\rho_*^\xi$ is a translation-invariant measure on $\Tnr_\xi$. So   $\rho_*^\xi= \dbar\zeta$, for each $\xi$. 

Similar, since by Corollary \ref{c.invar} the unique stationary measure $\mu$ is invariant with respect to rotations
$\Psi_{s^\perp}, s^\perp\in ({}_M\cA)^\perp$, then the measure $m$ is invariant with respect to translations by
vectors from $({}_M\cA)^\perp$. As before, this implies that  also $\rho^\xi\equiv  \dbar\zeta$. So $m=m_*$,
which proves the theorem.
\end{proof}

\noindent
{\it On the energy spectrum of stationary measures for \eqref{5.eff}.}  From \eqref{integrals} we know that 
$\cHR$ has infinitely many quadratic integrals of motion  of the form
\begin{equation}\label{Hint}
H(v)=\frac12\sum_{j=1}^\infty \alpha_j |v_j|^2=\sum \alpha_jI_j.
\end{equation}
Let $v(\tau)$ be a stationary solution of \eqref{5.eff} as  in Section~\ref{s1.1},
$\cD v(\tau)\equiv \mu$.  Applying the Ito formula to $H(v(\tau))$, taking the expectation and using 
\eqref{integrals} we get 
$$
0= \E\big(dH(v)(R^1(v)+R^0(v))  +\sum_j b_j^2d^2 H(v)(e_j,e_j)\big)=
-\sum_j\alpha_j \gamma_j \E |v_j|^2 + \sum \alpha_j b_j^2.
$$
Denoting ${\cE}_j=\E\tfrac12|v_j|^2$, we obtain
$$
2\sum_j \cE_j\gamma_j\alpha_j =\sum_j b_j^2\alpha_j\quad \text{for any $(\alpha_1,\alpha_2,\dots)$ as in \eqref{Hint}. 
}
$$
This is an under-determined system of linear equations for the vector $(\cE_1,\cE_2,\dots)$. The vector
 defines the energy spectrum  $E_r$ of $\mu$ after the additional averaging along the shells 
 $\{ |\bk| \approx\const\}$). The system above impose restrictions on the energy spectrum, but does not
 determine it.

\section{Explicit calculation}\label{s2.5}
We intend here to calculate explicitly the  effective equation
\eqref{5.eff}, keeping  track of the dependence on the size $L$ of the
torus. To do that, it is
convenient to use the natural parametrisation of the exponential basis 
by  vectors $\bk\in\Z^d$; that is, decompose functions $u(x)$ to Fourier series, 
$
u(x)= \sum_{\bk\in \Z^d}v_\bk e^{i L^{-1}\bk\cdot x}\  .
$
We re-define the norms $|\cdot |_{h^p}$ accordingly :
$$
\left\| u\right\|^2_p =(2\pi L)^d\sum_{\bk\in\Z^d}\left(\frac
{|\bk|\vee 1}{L}\right)^{2p} |v_\bk|^2=:
\left|v\right|^2_{h^p}\ .
$$
Now, as in the Introduction,  the eigenvalues of the minus-Laplacian are
$\lambda_\bk=|\bk|^2 /L^2$.

In the $v$-coordinates the nonlinearity becomes the mapping $v\mapsto
P^0(v)$, whose $\bk$-th component is
$$
P^0_\bk(v)=-i\gi \sum_{\bk_1,\ldots \bk_{2q_*+1}\in \Z^d}
v_{\bk_1} \cdots  v_{\bk_{q_*+1}} \bar v_{\bk_{q_*+2}}\cdots \bar v_{\bk_{2q_*+1}}
\delta^{1\ldots q_*+1}_{q_*+2\ldots 2q_*+1\, \bk}\ 
$$
(see \eqref{N1}). Accordingly, 
\begin{equation}\label{eq:example}
%\begin{split}
v_\bk\cdot P^0_\bk
%&= \frac {i\gi}2 \sum_{r_1,\ldots r_{2q_*+1}\in \Z^d} \left(
%\bar v_{r_1} \cdots  \bar v_{r_{q_*+1}}v_{r_{q_*+2}}\cdots v_{r_{2q_*+1}}
%v_r -  c.c. 
% v_{r_1} \cdots   v_{r_{q_*+1}} \bar v_{r_{q_*+2}}\cdots \bar v_{r_{2q_*+1}}\bar v_r
%\right) 
%\delta^{1\ldots q_*+1}_{q_*+2\ldots 2q_*+1\, r}\\  &
= \gi \sum_{\bk_1,\ldots \bk_{2q_*+1}\in \Z^d} \Im\,(v_{\bk_1} \cdots
v_{\bk_{q_*+1}} \bar v_{\bk_{q_*+2}}\cdots \bar v_{\bk_{2q_*+1}} 
\bar v_\bk )
\delta^{1\ldots q_*+1}_{q_*+2\ldots 2q_*+1\, \bk}
%\end{split}
\end{equation}
In order to calculate the resonant average, we first  notice that $v_\bk\cdot
P^0_\bk$ can be written as a series \eqref{xx}, where $|C_{pql}| \le 1$ and
$|q|+|p|+|l|=2q_*+2$. In this case the sum
in the l.h.s. of \eqref{xxx} is bounded by
$$
C \left( \sum_{\bk\in\Z^d}|v_\bk|\right)^{2q_*+2}\le C_1(L)|v|_p^{q_*+1} \left(
\sum_{\bk\in\Z^d}|\bk|^{-2p} \right)^{q_*+1}.
$$
So the condition \eqref{xxx} is met if $2p>d$. 

Since the order of the resonance $m=2q_*+2$, then  $\lan v_\bk\cdot P^0_\bk\ran_\Lambda(v)$
equals 
$$
 \gi \sum_{\bk_1,\ldots \bk_{2q_*+1}\in \Z^d}
\Im\,(v_{\bk_1} \cdots  v_{\bk_{q_*+1}} \bar v_{\bk_{q_*+2}}\cdots \bar
v_{\bk_{2q_*+1}}\bar v_\bk)
\delta^{1\ldots q_*+1}_{q_*+2\ldots 2q_*+1\, \bk}\delta(\lambda^{1\ldots
  q_*+1}_{q_*+2\ldots 2q_*+1\, \bk})\ ,
$$
(see \eqref{N2}).
This  follows from \eqref{eq:example} and
\eqref{La_aver} if one notes  that  appearing  there
 restriction
$(q-l)\cdot \Lambda=0$ is now  replaced 
by
 the factor $\delta(\lambda^{1\ldots q_*+1}_{q_*+2\ldots 2q_*+1\, \bk})$.
In a similar way, we see that the quantity $R^0_k$ , entering  equation \eqref{5.eff},  takes the form
$$
R^0_\bk(v)= -i\gi \sum_{\bk_1,\ldots \bk_{2q_*+1}\in \Z^d}
v_{\bk_1} \cdots  v_{\bk_{q_*+1}} \bar v_{\bk_{q_*+2}}\cdots \bar
v_{\bk_{2q_*+1}}
\delta^{1\ldots q_*+1}_{q_*+2\ldots 2q_*+1\, \bk}\delta(\lambda^{1\ldots
  q_*+1}_{q_*+2\ldots 2q_*+1\, \bk})\ .
$$
Taking into account that $R^1_\bk=-\lla_\bk v_\bk$, we finally arrive at an explicit formula
for the effective equation 
\eqref{5.eff}:
 \begin{equation}\label{explicit}
\begin{split}
&dv_\bk= \Bigl(-\lla_\bk v_\bk \\
&-i\gi \sum_{\bk_1,\ldots \bk_{2q_*+1}\in \Z^d}
v_{\bk_1} \cdots  v_{\bk_{q_*+1}} \bar v_{\bk_{q_*+2}}\cdots \bar
v_{\bk_{2q_*+1}}
\delta^{1\ldots q_*+1}_{q_*+2\ldots 2q_*+1\, \bk}\delta(\lambda^{1\ldots
  q_*+1}_{q_*+2\ldots 2q_*+1\, \bk})\Bigr)d\tau \\
&+b_\bk d\bb^\bk\ , \qquad \bk\in \Z^d\ .
\end{split}
\end{equation}
Due to \eqref{p2},
$$
R^0_k(v) = i\rho \nabla_{v_k}\cH^{\text{res}}(v)=
 2i\rho \frac{\p}{\p \bar v_k}
 \cH^{\text{res}}(v).
$$
Therefore eq. \eqref{explicit}  can be written as the damped--driven hamiltonian
system  \eqref{*eff1}.
\medskip

\noindent 
{\it Examples.} 
a) If $q_*=1$, then \eqref{explicit}  reads
\begin{equation*}
\begin{split}
dv_\bk= \Big(-\lla_\bk v_\bk 
-i\gi \sum_{\bk,\bk',\bk''\in \Z^d}
v_{\bk}  v_{\bk'} \bar v_{\bk''}
\delta_{\bk+\bk'\,,\,\bk''+
  r}\,\delta_{\lambda_\bk+\lambda_{\bk'}\,,\,\lambda_{\bk''}+
  \lambda_\bk}\Big)  d\tau
+b_\bk d\bb^\bk\ , 
\end{split}
\end{equation*}
where  $\bk\in \Z^d$. 
If $f(t)=t+1$, then this equation  looks similar to the CGL equation 
$$
\dot u-\Delta u+u= i|u|^{2}u +\frac{d}{d\tau}\sum b_\bk \bb^\bk(\tau)e^{i\bk\cdot x}, 
$$
written in the Fourier coefficients. The last equation 
 possesses nice analytical properties; e.g.  its stationary 
measures is unique for any $d$, see \cite{KNer13}. 

 In order to get simple but nontrivial
examples, consider one-dimensional equations.

b) Let in \eqref{1.111}   $q_*=2$. Then the effective equation reads
\begin{equation*}
\begin{split}
dv_k= &-\lla_k v_k d\tau +b_kd\bb^k 
\\
&-i \gi\sum_{k_1,k_2,k_3,k_4,k_5\in \Z}
v_{k_1}  v_{k_2} v_{k_3} \bar v_{k_4} \bar v_{k_5}
\delta_{k_1+k_2+k_3\,,\,k_4+k_5+k}\,\delta_{k_1^2+k_2^2+k_3^2\,,\,k_4^2+
  k_4^2+k^2}  d\tau\,.
\end{split}
\end{equation*}

c) Our results remain true if the Hamiltonian $\cH$, corresponding to the nonlinearity in \eqref{1.111},
has variable coefficients. In particular, let $d=1$ and  the  nonlinearity  in   \eqref{1.111}  is replaced by 
$-i p(x) |u|^{2}u$ with a sufficiently smooth function $p(x)$.
Then the effective equation is 
$$
dv_k= \Big(-\lla_k v_k -i \sum_{k_1,k_2,k_3,k_4\in \Z}
v_{k_1}  v_{k_2} \bar v_{k_3} p_{k_4}
\delta_{k_1+k_2+k_4\,,\,k_3+k}\,\delta_{k_1^2+k_2^2\,,\,k_3^2+
  k^2}\Big)  d\tau\,
+b_kd\bb^k\ ,
$$
where  $k\in \Z$ and  $p_k$'s are the Fourier coefficients of $p(x)$.

\section{Proof of  Theorem \ref{t5.22} and Proposition \ref{p.slow}
}\label{s3}

\subsection{Proof  of the theorem and the proposition }\label{s5.22}
First step of the proof is traditional:

\noindent
\textbf{Step 1.} {\it Redefining the equations for large amplitudes.}
For any $P\in \N$ consider the stopping time
\begin{equation}\label{ta_P}
\tau_P= \inf\left\{\tau\in[0,T]\Big| |v(\tau)|^2_{h^r}=P\right\} 
\end{equation}
(here and in similar situations below $\tau_P=T$ if the set is
empty). For $\tau\ge \tau_P$ and each $\nu > 0$ we redefine equations
\eqref{5.1} to the trivial system
$$
dv_k= b_kd\bb^k\ ,\qquad k\ge 1\ .
$$
Thus we have modified equations \eqref{5.1} to  the system 
\begin{equation}\label{5.11}
\begin{split}
&dv_k= \chi_{\tau\le \tau_P}\left(-i\nu^{-1}\lambda_kv_kd\tau +P_k(v)\,d\tau  \right)  +\,
b_kd\bb^k(\tau)
%+\chi_{\tau\ge \tau_p}  b_kd\bb^k 
\ ,\quad k\ge 1\ ,\\
&v(0)=v(u_0). 
\end{split}
\end{equation}
Similar to $v^\nu$, a solution  $v^\nu_P$  of \eqref{5.11}
meets the estimates
\begin{equation}\label{gafa3.10}
\E \sup_{0\le\tau\le T}|I(\tau)|_{h^m_I}^M= \E \sup_{0\le\tau\le
  T}|v(\tau )|_{h^m}^{2M}\le C(M,m,T)\ ,
\end{equation}
uniformly in $\nu\in(0,1]$.

Replacing the sequence  $\nu_\ell\to 0$  in \eqref{5.88} by a 
 suitable subsequence we
achieve that also $\cD(I_P^{\nu_\ell}(\cdot),
V_P^{\nu_\ell}( \cdot))\strela \cQ_P$ for some law $\cQ_P$, for any
$P\in\N$. Clearly, for the limiting  laws $\cQ_P$ estimates \eqref{gafa3.10} still hold.

\begin{lemma}
For each $P\in\N$, one has $\cQ_P=\cQ^0$ for $\tau\le \tau_P$ (that
is, images of the two measures under the mapping
$(I,V)(\tau)\mapsto (I,V)(\tau\wedge \tau_P)$ are equal), and
$\cQ_P((I,V)(\cdot))\strela \cQ^0((I,V)(\cdot))$ as $P\to\infty$.
\end{lemma} 
\noindent
{\it Proof. } Since $\cD(v^\nu_P)=\cD(v^\nu)=: \PP^\nu$ for $\tau\le
\tau_P$, then passing to the limit as $\nu_\ell\to 0$ we get the first
assertion of the lemma. As $\PP^\nu\{\tau_P  < T\}\le CP^{-1}$
uniformly in $\nu$ (cf. \eqref{gafa3.10}), then the last assertion follows.
\qed
\medskip

\noindent 
\textbf{Step 2.} {\it Galerkin-like approximation for the effective equation.}
Let us fix  any $N\in\N$. Given a process $(I(\tau), V(\tau))$ such that its law equals $\cQ^0$, we
 consider a Galerkin-like approximation of order $N$
 for the effective equation \eqref{5.eff} with $\tau\ge0$ 
  of the form 
$v(\tau)=(v^N(\tau),  v^{>N}(0) )$, 
where $ v^{>N}=( v_{N+1}, v_{ N+2},\dots)$,  and
 $v^N(\tau)=(v^N_1,\ldots,v^N_N)  (\tau)\in\C^N$ solves
\begin{equation}\label{5.effgal}
d v_k^N(\tau)= R_k(v) d \tau +{b_k}d \bb^k(\tau)\,, 
%v=\big(v^N, {\cal V}_\theta^{>N} \big),
\qquad k\le N\ .
\end{equation}
Here $R_k(v)=R^1_k(v^N)+ +R^0_k(v)$, where $R^1_k(v)=-\gamma_kv_k$ and
$R^0_k(v)=\frac{v^N_k}{2I_k} R^-_k(I,V)$ (see \eqref{ura2}).  The argument $(I,V)$ is understood as follows:
\begin{equation}\label{arg}
(I,V)=\big(I^N(v(\tau)), I^{>N}(\tau), V^{(N)}(v(\tau)), V^{>(N)}(\tau)\big). 
\end{equation}
Here $I^N(v(\tau))=\big(I_1(v(\tau)), \dots, I_N(v(\tau))\big)$, while $ I^{>N}(\tau) =(I_{N+1}(\tau), dots)$ components 
of the process $(I(\tau), V(\tau))$, and similar with $ V^{(N)}(v(\tau))$ and $V^{>(N)}(\tau)$.

Our   goal is to construct a process $\big((I(\tau),V(\tau)), v^N(\tau)\big)$ such that  

(a) $\cD(I(\cdot), V(\cdot))=\cQ_P$, 

(b) $(I^N,V^{(N)}) (v^N(\tau))=(I^N(\tau), V^{(N)}(\tau))$ a.s.,

(c) $v^N$ is a weak solution of the modified  system 
\begin{equation}\label{5.effgalP}
d v_k^N(\tau)=  \chi_{\tau\le \tau_P}  R_k(v) d \tau 
+ b_kd\bb^k\ ,% \quad v=\big(v^N, {\cal V}_\theta^{>N} \big).
\quad k\le N\ ,
\end{equation}
where the agreement \eqref{arg} holds. 
Denote $\PP(N,P)=\cD\big((I(\cdot), V(\cdot)), v^N(\cdot) \big)$. After these measures are 
constructed for each $N$ and $P$, we will obtain a required weak solution of \eqref{5.eff}
 as a $v$-component of the  limit of  $\PP(N,P)$  as
 $N\to \infty$ and $P\to \infty$.

 Transition to the limit is not complicated, but  construction of a process $v^N$ as above is 
 rather technical. We start with its sketch.  Let us denote 
$$
[I]=\min_{1\le k\le N} I_k 
$$
and set $\tau_*=\min\{\tau\le T:[I(\tau)]=0\}$. Construction of $\PP(N,P)$ for $\tau<\tau_*$ is 
sufficiently natural:
using Ito's formula we rewrite   \eqref{5.effgal} as a system of equations for $I^N$ plus a system of 
equations for $\vp^N$, see below \eqref{_1},~\eqref{_2}, where $k\le N$. These equations
are equivalent to the system \eqref{_4},~\eqref{_5}, $k\le N$.  The $I$-
equations \eqref{_4} do not include the angles $\vp^N$ and their coefficients depend only
on $I$ and $V$, which are known quantities with the law $\cQ_P$.
The process $I^N$ turns out to be 
 its solution (this is proven in Lemma \ref{l5.media} below). 
The $\vp$-equations \eqref{_5}, $k\le N$, make a  functional SDE on $\T^N$ 
with Lipschitz coefficients, and they have a unique solution $\vp^N$. So we have constructed a weak
solution $v^N$ for $\tau<\tau_*$. 

A posteriori we know that $I_k(\tau)$ vanishes when vanishes $v_k(\tau)\in\R^2$. It seems  unlikely that 
$v_k(\tau)=0$ for some $\tau>0$ since this is a stochastic process in $\R^2$, and it is well known that a
Wiener process in $\R^2$, starting from the origin, will never vanish for the second time, a.s.  Same is true 
for two-dimensional diffusion processes.
But
$v_k$ is an Ito process, and  is more complicated then a two-dimensional diffusion. 
A celebrated 
example due to  N.~Krylov (see in \cite{Kry77}) shows that such processes may pass through the origin 
many times with a positive probability. So we cannot rule out that $\tau_*<T$ with a positive 
probability, and  have to make a construction of $v^N$ more complicated. Namely, we
 fix any $\delta>0$, denote
by $\Lambda(\delta)$ the random set $\{[I]\ge\delta\}$ and denote $\Delta(\delta)=[0,T]\setminus \Lambda(\delta)$ (exact 
definition is a bit more complicated, see below). On segments, forming $\Delta(\delta)$, we construct  $v^N$ as above, and on intervals, forming $\Delta(\delta)$, define $v^N$ ``somehow", keeping the 
process $v^N$ continuous and satisfying  the property (b).  Thus constructed process $v^N$ will be called 
a {\it $\delta$-solution of \eqref{5.effgalP}. }

We know from Lemma~\ref{l5.1} that the set 
$\Delta(\delta)$ in average disappears 
with $\delta$. So  going to a limit as $\delta\to0$  we get that $\delta$-solutions converge 
to a process $v^N$ which solves
 eq.~\eqref{5.effgal} on the whole segment $[0,T]$ and satisfies (a) and (b). Rigorous realisation 
 of this construction turns out to be involved. It is given below.
\medskip

\noindent 
\textbf{Step 3.} {\it Construction of $\delta$-solutions.}
We denote
$$
\ho=\cH_{I,V}\times  C([0,T],\C^N)=\{(I,V),v^N\}.
$$
This is a complete separable metric space, and we 
  provide it with  the Borel sigma-algebra and with 
 the natural filtration of sigma-algebras $\{\hat\cF_\tau\}$.

 Find a  measurable process $(I(\tau), V(\tau))$ on $\widehat\Omega$, adapted to the 
 filtration, which has  the law $\cQ_P$. Define the stopping times $0\le\theta_n^\pm\le T$
such that $\ \ldots<\theta_n^-<\theta_n^+<\theta_{n+1}^-<\ldots$  and 
$\theta_0^+=0$ if $[I(0)]>\delta$ while $\theta_1^-=0$ if  $[I(0)]\le\delta$. Their construction
is clear from Fig.~1, where we assume that $[I(0)]>\delta$. It is easy to see that a.s. these 
stopping times stabilise at $T$ when $n\to\infty$.

  %  \includegraphics[height=3.0cm]{1.pdf}
% \includegraphics{1.pdf}

%\includegraphics{f1.pdf}

%%%%%%%%%%%%%%%%%%%%%%%%%%
%%%%%%%%%%%%%%%%%%%%%%%
\begin{figure}
\centering
\includegraphics[width=4.5in]{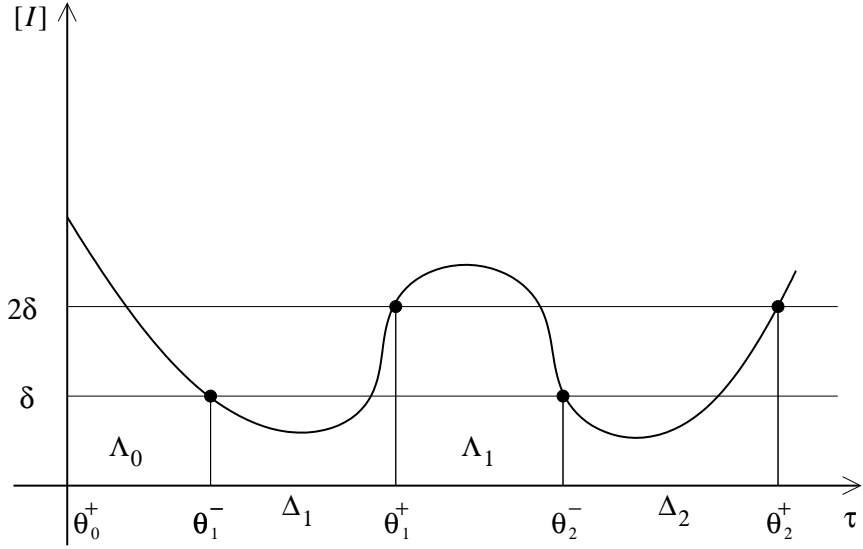}
\caption{Stopping times $\theta_n^\pm$}
\end{figure}

%\begin{itemize}
%\item If $[I(0)]\le \delta$, then $\theta^-_0=0$, otherwise
 % $\theta^+_0=0$.
%\item If $\theta_j^-$ is defined, then $\theta_j^+$ is the first
%  moment after $\theta_j^-$ when $[I(\tau)]\ge 2\delta$ (if this never
 % happens, then we set $\theta_j^+=T$; similar in item below).
%\item If $\theta_j^+$ is defined, then $\theta^-_{j+1}$ is the first
 % moment after $\theta_j^+$ when $[I(\tau)]\le \delta$.
%\end{itemize}

We denote $\Delta_n=[\theta_n^-,\theta_n^+]$, $\Lambda_n= [\theta_n^+,
  \theta_{n+1}^-]$ and set
\begin{equation}\label{sets}
\Delta=\Delta(\delta)=\cup\Delta_n\ , \quad \Lambda=\Lambda(\delta)=
\cup \Lambda_n\ .
\end{equation}
Our goal at this step is to get a measure $\PP_\delta$  on the space $\widehat \Omega$,
which is a law of a process $((I,V),v^N)$ satisfying (a), (b) on $[0,T]$ and (c) on $\Lambda(\delta)$, 
as well as the following property:

(d) $I(\tau)$ is a weak solution of \eqref{5.21}  for $0\le\tau\le \tau_P$, while for 
$\tau\in\Lambda(\delta)\cap[0,\tau_P]$ 
the process $\big(I,\Phi^{(N)}(V)\big)$ is a weak solution of equations \eqref{5.21}, 
\eqref{5.22}${}\mid_{j\le J(N)}$.

For segments $[0,\theta_n^-]$ and $[0,\theta_n^+]$, which we denote
below $[0,\theta_n^\pm]$, we will iteratively construct   processes 
${}^{\pm, \,n}\Upsilon=(( I,  V),  v^N)$, where $v^N$ is 
stopped at $\tau=\theta_n^\pm$, 
 such that
 \begin{equation}\label{prop}
 \begin{split}
 \text{
 the process} &\text{ $\ {}^{\pm, \,n}\Upsilon$  satisfies (a) and (d),  satisfies (b) for $\tau\le\theta_n^\pm$},\\
&\text{ and satisfies (c)   for $\tau\in[0,\theta_n^\pm]\cap \Lambda\cap [0,\tau_P]$.
} \end{split}
\end{equation}
  Next we will obtain a desired  measure
$\PP_\delta$ as a limit of the laws of  $\ {}^{\pm, \,n}\Upsilon$  as $n\to \infty$.

For the sake of definiteness assume that $0=\theta^+_0$.

\textbf{a)}  We first show how to construct the
process $\,  {}^{-,\, 1\!}\Upsilon =((I, V),v^N) $ (where $v^N$ is stopped at $\tau=\theta_1^-$):

\begin{lemma}\label{l.2.5}
For any positive $\delta$ there exists a process $\,  {}^{-,\, 1\!}\Upsilon =((I, V),v^N)$ 
 such that  $v^N(0)=v_0^N$, $v^N(\tau)$ is stopped at $\tau=\theta_1^-$, and \eqref{prop}
 holds.
\end{lemma}
\noindent
{\it Proof. }  On the filtered probability space $\ho$  let us  choose  continuous processes
$v_P^{' \nu}(\tau)$ and $ (I_P', V_P')(\tau)$, adapted to the filtration,
such that 
$\cD(v_P^{'\nu})=\cD( v_P^{\nu})$ and $\cD((I_P^{'}, V_P^{'}))=Q_P$. Then 
\begin{equation}\label{converg}
\cD\big(
(I,V)({v'_P}^\nu(\cdot)) \big) \strela \cD( (I'_P,V'_P)(\cdot))\quad \mbox{as }
  \nu=\nu_\ell\to 0.
%\mbox{ in } \cH_{I,V}\mbox{ a.s.,}
\end{equation}
To simplify notation we will usually drop the prime and  sub-index $P$.

%We have to construct   process $((I, V),v^N)$, where $v^N$ is a lift for $I^N$.
%Keeping in mind the convergence \eqref{converg}, we first build a 

First we demonstrate the lemma's assertions for $\tau\le \theta_1^-\wedge\tau_P$.
For any positive $\nu$, for the process $v^\nu(\tau)$ we define stopping times 
${\theta^{\pm}_n}^\nu$ and $\tau_P^\nu$  as the stopping times 
 $\theta_n^\pm$ and $\tau_P$ for the process
$I(\tau)=I({v}^\nu(\tau))$.

Clearly
\begin{equation}\label{diffeo}
\text{the mapping $\{[I]\ge\delta\} \times \T^N\ni (I^N, \vp^N)\mapsto v^N$ is a diffeomorphism
}
\end{equation}
(on its image in $\R^{2N}$). 
So if  the process $v^\nu(\tau)$ is such that for 
  $0\le \tau\le  {\theta_1^-}^\nu\wedge \tau_P^\nu$ its first $N$ components
meet  equations \eqref{5.1}${}_{1\le k\le N}$, then
% if and only the pair $(I^{\nu N}, \vp^{\nu N})$
%meets the equations 
% (\eqref{5.2}, \eqref{5.3})${}_{1\le k\le N}$. Secondly, %f  $v^N$ 
%satisfies  \eqref{5.1}${}_{1\le k\le N}$,  then 
 the process 
$(I^N,\Phi^{(N)})({v}^\nu)$ satisfies 
 \begin{equation}\label{5.2buona}
\begin{split}
dI^\nu_k(\tau)& =({v_k}^\nu\cdot P_k)({v}^\nu)\,d\tau +
b_k^2\,d\tau  +b_k( {v_k}^\nu\cdot
d\bb^k)\ ,\quad  k\le N\ ,\\
d\Phi^\nu_j(\tau)&=\sum_{k\ge 1}s^{(j)}_k\Bigl[
\frac{ (iv^\nu_k\cdot P_k(v^\nu))} {|v^\nu_k|^2}
%G_k({v}^\nu)
\,d\tau +
|v_k^\nu|^{-1}b_k\Big(\frac{iv_k^\nu}{|v_k^\nu|}\cdot
d\bb^k(\tau)\Big)\Bigr]\ , \quad  j\le J\ 
\end{split}
\end{equation}
(see Section \ref{s5.1}). 

Now we will pass to a limit (as $\nu_\ell\to0$) in equations \eqref{5.2buona}. To do this 
we will make use of the following averaging lemma, which is proved below in Section~\ref{sez:dim}, closely following the infinite-dimensional version of the Khasminski argument \cite{Khas68}, given 
in \cite{KP08}.

\begin{lemma}\label{l5.media}
For any $N\ge1$ and for 
$0\le \tau\le \theta_1^-$ the law $(I,\Phi^{(N)})\circ \cQ^0$
%(I_P^N(\tau),\Phi_P^{\,(N)}(\tau))$
%one has $\cQ_P(I_P^N(\tau),\Phi_{P\,(N)}(\tau))=
equals that of $(I, \Phi^{(N)})(I,V)$, where the process $(I,V)(\tau)$ is such that for
 $0\le \tau\le \theta_1^-$ its component $(I^N, \Phi^{(N)})$ is 
 a weak solution of the system of averaged equations
 \eqref{5.21}${}_{k\le N}$,  \eqref{5.22}${}_{j\le J}$, and  $(I(0),V(0))=(I(v(0)),V(v(0)))$. 
%In these equations  processes $I_k(\tau)$ and
%$ {V_j}(\tau)$ with  $k> N$, $j> J$, are such that $\cD(I^{0N}, I^{>N}, \Phi_{0(N)}, \Phi_{>J(N)})=\cQ_P$.
\end{lemma}

Using this lemma  we  assume that the process $(I'_P, V'_P)(\tau)$ satisfy equations 
\eqref{5.21},  \eqref{5.22}${}_{j\le J}$ for $0\le\tau\le\theta_1^-\wedge\tau_P$. 

For $0\le\tau\le\theta_1^-\wedge \tau_P$ and for $k=1,\dots,N$ in view of \eqref{diffeo}
a process 
 $v^N_k$ satisfies \eqref{5.effgal} if and only if its action $I_k=I(v^N_k)$ and  angle
   $\vp_k=\vp(v^N_k)$ meet equations
\begin{equation}\label{_1}
dI_k=v_k\cdot R_k\,d\tau+b_k^2\,d\tau+ b_kv_k\cdot  d\bb^k(\tau),
\end{equation}
\begin{equation}\label{_2}
d\vp_k=|v_k|^{-2}(iv_k\cdot R_k)+|v_k|^{-2}b_k(iv_k\cdot  d \bb^k(\tau))
\end{equation}
(cf. \eqref{5.2}, \eqref{5.3}). 
Let us write the stochastic part in \eqref{_1} as 
$$
b_k\sqrt{2I_k}\, \Re (e^{-i\vp_k}d\bb^k)=: b_k\sqrt{2I_k}\, \Re d\tilde\bb^k,
$$
where $\tilde\bb^k=\int e^{-i\vp_k}d\bb^k$. This  is a new standard complex Wiener process, independent from 
$\bb^l$ with $l\ne k$. The stochastic part of  \eqref{_2} is 
$$
|v_k|^{-1}b_k\Re(-ie^{-i\vp_k}d\bb^k)= (2I_k)^{-1/2} b_k\, \rm{Im} \, d\tilde\bb^k.
$$
Using the new Wiener processes we may write the $k$-th component of the 
 $v^N$-eq.  \eqref{5.effgal} as 
%choose in \eqref{5.21}, \eqref{5.22} \ $\beta^k=\Re \tilde\bb^k$. Then 
%$v_k$ is a weak solution of equation \eqref{5.effgal}, written as 
\begin{equation}\label{_3}
dv_k=R_k(v)\,d\tau +b_k 
 d \bb^k=  R_k(v)\,d\tau +b_ke^{i\vp_k}
 d\tilde\bb^k.
\end{equation}
We have seen that $v_k$ satisfies equation 
\eqref{_3}=\eqref{5.effgal}
if and only if 
\begin{equation}\label{_4}
dI_k=v_k\cdot R_k\,d\tau+b_k^2\,d\tau+ b_k  \sqrt{2I_k} \Re  d\tilde\bb^k,
\end{equation}
\begin{equation}\label{_5}
d\vp_k=|v_k|^{-2}(iv_k\cdot R_k)\, d\tau + (2I_k)^{-1/2}  b_k \, \rm{Im} \,
 d\tilde\bb^k \, ,
\end{equation}
where $v_k\cdot R_k=\lan v_k\cdot P_k\ran_\Lambda(I,V)$, 
$iv_k\cdot R_k=\lan iv_k\cdot P_k\ran_\Lambda(I,V)$ and
the components $I^N=I^N(v^N), V^{(N)}=V^{(N)}(v^N)$ are expressed in terms of the solution $I^N,\vp^N$
of \eqref{_4}, \eqref{_5}. 

Recall that the process $(I(\tau) , V^J(\tau))=(I'_P(\tau),V^{'(N)}_P(\tau))$ satisfies equations 
 \eqref{5.21},  \eqref{5.22}${}_{j\le J}$.   Since   $\lan v_k\cdot P_k\ran_\Lambda=v_k\cdot R_k$, we  write the first $N$ equations in \eqref{5.21} in the form \eqref{_4}. 
  For $\tau\le \theta^{-}_1\wedge \tau_P$ processes $\Phi_j$ with $j\le J(N)$ satisfy 
 equations \eqref{5.22} which we write as 
 %now take the form 
\begin{equation}\label{FFF}
d\Phi_j=
\sum_k s^{(j)}_k\Big( \frac{\lan iv_k\cdot P_k\ran_\Lambda(I,V)}{2I_k}\,d\tau +
\frac{b_k}{2\sqrt{I_k}}\,\text{Im}\, d\tilde\bb^k\Big),\qquad j\le J. 
\end{equation}
That is, the process $(I^{'N}_P(\tau),V^{'(N)}_P(\tau))$ is a solution of the system 
 \eqref{_4}${}_{k\le N}$, \eqref{FFF}.

Let us turn to equations \eqref{_5}${}_{k=1,\dots,N}$. 
  Since 
$
\delta\le [I(\tau)] \le |I(\tau)|_{h_I}
\le P
$
for $\tau\le \theta^{-}_1\wedge \tau_P$,  then for such $\tau$ the equations  form a Lipschitz 
system of functional
stochastic equations in $\T^N$ (see  \cite{KaSh, Kry03}). 
So they have  a unique strong solution $\vp^N(\tau)$, 
satisfying  $\vp^N(0)=\theta^N$.
Thus we obtain a process
 $v^N(\tau)=(I^N(\tau), \vp^N(\tau))$ 
 which satisfies \eqref{_3}${}_{k=1,\dots,N}$. Since  $v^N$ satisfies \eqref{_3}, then $I^N(v^N)$ satisfies 
  \eqref{_4}${}_{k\le N}$, while $\Phi^{J(N)}$ satisfies \eqref{FFF}. Since for $\tau\le \theta^-_1$ 
  the system  \eqref{_3}${}_{k=1,\dots,N}$, \eqref{FFF} has Lipschitz coefficients, then
  $$
  I_P^{'N}=I^N(v^N),\quad V^J(v^N)=V^{'(N)}_P.
  $$
  That is, the process $(I,V,v^N)$ satisfies \eqref{prop} for $\tau\le \theta_1^-\wedge \tau_P$.
  Extension of $v^N(\tau)$ to  the segment $[\theta_1^-\wedge \tau_P,\theta_1^-]$ is trivial,
and the process $\,  {}^{-,\, 1\!}\Upsilon$ is constructed. 
\qed
\medskip

\textbf{b)} Now we extend $\PP^-_1$ to a measure $\PP_1^+=\cD({}^{+,1} \Upsilon)$,
$\ {}^{+,1} \Upsilon(\tau)=(I,V,v^N)(\tau)$,  where $v^N$
is stopped at time $\theta_1^+$.

Let us denote by $\Theta=\Theta^{\theta^-_1}$ the operator which stops
any continuous trajectory $\eta(\tau)$ at time $\tau=\theta^-_1$. That
is, replaces it by $\eta(\tau\wedge \theta^-_1)$.

As in the proof of Lemma \ref{l.2.5}, we represent  the laws $\PP^-_1$ and
$\cD(v^\nu_P)$ by distribution of processes
$((I'_P(\tau),V'_P(\tau)), {v}^{'N}_P(\tau))$ and ${v}^{'\nu}_P$  such that $v_P^{'N}$ is stopped at
$\tau = \theta_1^-$ and 
\begin{equation*}
\begin{split}
& \cD\big((I,V)({v'_P}^\nu(\cdot)) \big)\strela
\cD\big( (I'_P,V'_P)(\cdot) \big)  =\cQ_P \quad \mbox{as }
  \nu=\nu_\ell\to 0, \\
%\mbox{ in } \cH_{I,V}\mbox{ a.s.,}\\
&(I,V)({v_P'}^N) =  ({I'_P}^N,{V'_P}^N) \quad\mbox{for
}\tau\le \theta^-_1\ .
\end{split}
\end{equation*}
As before, we will usually 
 drop the primes and the sub-index $P$. 

Since ${v}^\nu(\tau,\omega)$, $0\le \tau\le T$, is a diffusion
process, we may replace it by a continuous process
$w^\nu(\tau;\omega, \omega_1)$ on an extended probability space
$\Omega\times \Omega_1=(\omega,\omega_1)$, where $\Omega_1$ is a second
copy of the space $\ho$, such that
\begin{enumerate}
\item $\cD(w^\nu(\tau))=\cD({v}^\nu(\tau))$ ;
\item for $\tau\le \theta^-_1=\theta^-_1(\omega)$ we have
  $w^\nu={v}^\nu $ (in particular, then  $w_P^\nu$ is independent
  of $\omega_1$).
\item for $\tau\ge \theta^-_1$ the process $w^\nu$ depends on
  $\omega$ only through the initial data $w^\nu(\theta^-_1,\omega,
  \omega_1) = {v}^\nu(\theta^-_1,\omega)$. For a fixed $\omega$ it
  satisfies \eqref{5.11} with suitable Wiener processes $\bb^j$'s
  defined on the space $\Omega_1$.
\end{enumerate}

The process $w^\nu$ is fast: $\frac{d}{d\tau} w^\nu\sim\nu^{-1}$.
Using a construction from \cite{KP08, K10},
we now couple it with a slow process $\widetilde w^{\nu N}$ in $\C^N$ such that the actions 
and resonant combinations of the angles for the former and the latter coincide:

\begin{lemma}\label{l.couple}
There exists a measurable  process 
$(\bar w^\nu, \widetilde w^{\nu N})(\tau)\in h\times \C^N$, $0\le\tau\le T$, 
such that 

(i)  $\cD  \bar w^\nu(\cdot) =\cD  w^\nu(\cdot)$;

(ii) $(I^N, V^{(N)})(\widetilde w^{\nu\, N})=(I^N,V^{(N)})(\bar w^\nu)$ for $\tau\ge
  \theta^-_1$;
  
  (iii)  $ \widetilde w^{N}(\theta^-_1)=
   \bar w^{\nu\, N}(\theta^-_1)=
   {v}^{N}(\theta^-_1) $
   a.s. in $\Omega_1$;
 % \footnote{Notice  that, at variance with the result of \cite{KP08}, we claim that
    %also the $\Phi$'s coincide (see Appendix~\ref{app:KP08} below).} 
    
(iv) the law of the process $\widetilde w^{\nu \,
  N}(\tau)$, $\tau\ge \theta^-_1,$ is that of an Ito process $v^N\in\R^N$, 
\begin{equation}\label{5.mod}
d v^N=B^N(\tau) d\tau+a^N(\tau) d\beta^N(\tau)\ ,
\end{equation}
where for every $\tau$ the vector $B^N(\tau)$ and the $N\times N$-matrix
$a^N(\tau)$ satisfy $\nu$- and $\delta$--indepen\-dent estimates
\begin{equation}\label{5.modstime}
\left|B^N(\tau)\right|\le C\ , \quad C^{-1}\, \rm{Id}\le
a^N\cdot (a^N)^T(\tau) \le C \,\rm{Id} \quad \mbox{a.s.,}
\end{equation}
with some $C=C(P,N)$.
\end{lemma}
The lemma is proved below in Section \ref{s.pr_couple}.

Next, for $\nu=\nu_\ell$ consider the process
$$
\xi^\nu_P(\tau)= \left((I^\nu_P,V^\nu_P)= (I,V)(\bar w^\nu(\tau)), \;
\chi_{\tau\le \theta_1^-}{v}^N+\chi_{\tau\ge\theta_1^-}\tilde
w^{\nu\, N}\right)\ , \quad 0\le\tau\le T\ .
$$
Due to \eqref{gafa3.10} and (iii) the family of laws
$\{\cD( \xi^{\nu_\ell}_P),\ell\ge 1\}$ is tight in the space
  $C([0,T];\cH_{I,V}\times \C^N)$. Consider any limiting measure $\Pi$
  (corresponding to a suitable subsequence $\nu'_\ell\to 0$) and represent
  it by a process $\tilde \xi_P(\tau)= ((\tilde I_P,\tilde V_P)(\tau),
  \tilde v^N_P(\tau))$, i.e., $\cD(\tilde \xi_P)=\Pi$. Clearly,

(iv) $\cD(\tilde \xi_P)|_{\tau\le \theta_1^-}=\PP^-_1$ ;

(v) $\cD(\tilde I_P,\tilde V_P)=\cQ_P$.
\smallskip

Since each measure $\cD(\xi^\nu_P)$ is supported by the closed set,
formed by all trajectories satisfying $(I,V)(v^N) \equiv
(I^N,V^{(N)})$, then the limiting measure $\Pi$ also is.  That is, the process $\tilde \xi_P$ satisfies

(vi) $(I,V)(\tilde v^{N}_P(\tau)) =
(\tilde I^N_P(\tau),\tilde V^{(N)}_P(\tau))$ a.s.
\smallskip

Moreover, the law of the limiting process $\tilde v^N_P(\tau)$,
$\tau\ge \theta^-_1$, is that of an Ito process of the form  \eqref{5.mod},
\eqref{5.modstime}. (Note that for $\tau\ge \theta^-_1$ the process
$\tilde v^N_P$ \emph{is not} a solution of \eqref{5.effgalP}).

Now we set
$$
\PP^+_1=\Theta^{\theta^+_1}\circ \cD(\tilde \xi_P)\ .
$$

\textbf{c)} The constructed measure $\PP_1^+$ gives us the
distribution of a  process\\ $((I,V)(\tau),v^N(\tau))$ for $\tau\le
\theta_1^+$. Next we repeat step a) with the  initial data\\
$((I,V)(\theta_1^+), v^N(\theta_1^+))$ and iterate the
construction.

Since the sequence $\theta_n^\pm$ stabilises at
$\tau=T$ after a finite random number of steps a.s., then  the
sequence of measures $\PP_n^\pm$ converges to a limiting measure
$\PP_\delta$ on $\ho$. 

Consider the natural process  $\xi_\delta (\tau)=((I_\delta,V_\delta), v^N_\delta)(\tau)$
on  $(\widehat\Omega, \{\hat\cF_\tau\}, \PP_\delta)$. By construction it satisfies the 
properties (a)-(d). We have obtained a $\delta$-solution of \eqref{5.effgalP}. 
\medskip

\textbf{Step 4.} \textit{Limit $\delta\to 0$.}
Due to 1--3 (see Step 3.b)) the set of measures $\{\PP_\delta,0<\delta\le 1\}$ is
tight in $\widehat\Omega$. 
Let $\PP_P$ be any limiting measure as $\delta\to 0$. 

\begin{lemma}\label{lemmadelta}
The measure $\PP_P$ may be representes as a law of a process 
$ \xi=(I,V,v^{0N})$ which 
satisfies the properties (a) and (b) above and is such that 

i) for $\tau\le\tau_P$ the process 
 $\xi$ is a weak solution of equations \eqref{5.21}, \eqref{5.effgal}.

ii) Let $0\le\tau_1<\tau_2\le T$ be stopping times, defined in terms of the component 
$I(\tau)$, such that \eqref{stop} holds.  Then for $\tau_1\le\tau\le \tau_2$  process $\xi$
is a weak solution of equations  \eqref{5.21},  \eqref{5.22}${}_{j\le J(N)}$.
\end{lemma}

\begin{proof}
The  properties (a) and (b) are invariant under the weak limit, so since 
 they hold for $\PP_\delta$, then they also hold for  $\PP_P$.

Since for $\tau\le\tau_P$  the process 
$(I_\delta, V_\delta, v_\delta^N)$ is a weak solution of \eqref{5.21},  \eqref{5.effgal}
outside the set $\{\tau: [I(\tau)]<\delta\}$ which shrinks to the empty set when $\delta\to0$, then
the limiting process  is a weak solution of \eqref{5.21},  \eqref{5.effgal}. This follows in a traditional
way by applying the method of martingale solutions \cite{Yor74, MR99},
  see  Lemma~3.4    of \cite{KP08} or \cite{K10}.

If $\delta\le \tfrac12 \delta_*$, then the random segment $[\tau_1, \tau_2]$ is contained in 
$\Lambda(\delta)$ (see Fig.~1).  Since the process $(I_\delta, V_\delta, v_\delta^N)$  
satisfies (d) (see below \eqref{sets}), then  it is a weak solution of 
 \eqref{5.21},  \eqref{5.22}${}_{j\le J(N)}$ on the segment $[\tau_1,\tau_2]$, if  $\delta\le \tfrac12 \delta_*$.
 Passing to the limit as $\delta\to0$ we recover ii).
\end{proof}
\medskip

\textbf{Step 5.} \textit{Limit $P\to \infty$.}
Due to (a), relations \eqref{gafa3.10} and
Lemma~\ref{lemmadelta} the set of measures $\PP_P$, $P\in\N$, is
tight. Consider any limiting measure $\PP^N$ for this
family. Repeating the proof of Lemma~\ref{lemmadelta}, we find that
$\PP^N$ solves the martingale problem for eq. 
 \eqref{5.effgal}, \eqref{5.21}, and satisfies assertion ii) of Lemma~\ref{lemmadelta}. 
  It still
satisfies (a) and (b) of step 3.
% Let $((I(\tau),\Phi(\tau)),v^N(\tau))$ be
%a process with law $\PP^N$ such that $v^N(\tau)$ solves
%\eqref{5.effgal}. Denote by $^Nv(\tau)$ the process $(v^N(\tau),
%{\cal V}^{>N}_\theta (\tau))$ and denote by $\mu^N$ its law in the space
%$C([0,T], h)$.
\medskip

\textbf{Step 6.} \textit{Limit $N\to \infty$.}
Identifying $\widehat\Omega=\widehat\Omega_N$ with  the corresponding Galerkin  subspace 
of  $\bar\Omega=\cH_{I,V}\times C([0,T];h)$ we see that by virtue of \eqref{2.5} and \eqref{5.effgal}
the set of measures $\PP^N, N\ge1$, is tight in $\bar\Omega$. Let 
 $N_n\to \infty$ be a sequence such that $\PP^{N_n} \strela \PP^0$. Then $\PP^0$ satisfies (a) and (b) 
 for any $N$, i.e. 
 $$
 (I,V)(\tau)=(I,V)(v(\tau))\quad \PP^0\,\text{- a.s.}
 $$
Repeating again the proof of Lemma~\ref{lemmadelta} we see that $\PP^0$
is a martingale solution of the system \eqref{5.21}, \eqref{5.eff} and satisfies assertion ii)
of Lemma~\ref{lemmadelta}, for any $N\ge1$. 

%In particular, $I(\tau)=I(v(\tau))$ and $V(\tau)=V(v(\tau))$.
 Denote $\mu=\cD^{\PP^0}(v(\cdot))$. Then $\cQ^0=\cD^{\PP^0}(I,V)=\cD^\mu(I,V)(v)$, 
 where $v(\tau)$  is a weak solution of \eqref{5.eff}.  
\medskip

\textbf{Step 7.} \textit{End of the theorem's proof. }
Since $v(\tau)$, constructed above, is a weak solution of \eqref{5.eff}, then due to
Lemma~\ref{l.uniq} and the Yamada-Watanabe argument (see \cite{KaSh, Yor74, MR99}), weak 
and strong solutions for \eqref{5.eff} both exist and are unique. That is, the limit 
in \eqref{5.88} does not depend on the sequence $\nu_\ell\to0$, and the theorem is  proven.
\medskip

\textbf{Step 8.} \textit{Proof of Proposition \ref{p.slow}.} By Step 6, the measure $\PP^0$ is such that,
from one hand, the corresponding distribution of the process $(I, V)(\tau)$ equals $\cQ^0$, from another 
hand, it  satisfies  item ii) of Lemma~\ref{lemmadelta}.  
 It means that $\cQ^0$ satisfies the assertion of the proposition. 
\qed

\subsection{Proof of Lemma~\ref{l5.media}}\label{sez:dim}
Since the process $I^\nu(\tau)$ and the $I$-component of measure
$\cQ^0$ satisfy estimates \eqref{gafa3.10} and, accordingly, $\tau^P\to T$
a.s. with respect to $\cQ^0$, then it is sufficient to prove the lemma 
for $\cQ^0$ replaced by $\cQ_P$ and for $0\le\tau\le\theta_1^-\wedge\tau_P$,
for any $P$.

Let us denote $F^1_k= v_k\cdot P_k+b^2_k$ and
$F^2_k=\frac{ i v_k\cdot P_k}{ 2I_k}$ (see equations \eqref{5.2}, \eqref{5.3}),
 and   for $\tau\in[0,T]$ consider the processes 
\begin{equation*}
\begin{split}
N^{\nu_\ell}_k(\tau)&=I^{\nu_\ell}_k( \tau\wedge \theta^{\nu_\ell} ) -
 \int_0^{ \tau\wedge \theta^{\nu_\ell} }  %{\theta^{\nu_\ell} } 
\langle F^1_k\rangle_\Lambda(I^{\nu_\ell}(s),
\Phi^{\nu_\ell}(s)) \,ds\ ,\quad  k\le N\ ,\\
M^{\nu_\ell}_j(\tau)&=\Phi^{\nu_\ell}_j(\tau\wedge \theta^{\nu_\ell})
-\sum_{i\ge   1}s^{(j)}_i \int_0^{ \tau\wedge \theta^{\nu_\ell} }  
\langle F^2_i\rangle_\Lambda
(I^{\nu_\ell}(s), \Phi^{\nu_\ell}(s)) 
\,ds \ , \quad  j \le J(N)\ ,
\end{split}
\end{equation*}
where $I^\nu_j(\tau)=I_j(v^\nu(\tau)), \Phi^\nu_j(\tau)=\Phi_j(v^\nu(\tau))$. 
Due to \eqref{5.2} and  \eqref{5.ris} we can write $N^{\nu_\ell}_k$ as
$$
N^{\nu_\ell}_k(\tau)=   \widetilde N^{\nu_\ell}_k(\tau)+   \overline  {N}^{\,\nu_\ell}_k(\tau),
$$
where $ \widetilde N^{\nu_\ell}_k=I^{\nu_\ell}_k(\tau\wedge \theta^{\nu_\ell}) -\int_0^{ \tau\wedge \theta^{\nu_\ell} }
F^1_k(I^{\nu_\ell},\vp^{\nu_\ell})\,d\tau$ is a $\cQ_P$ martingale and
$$
 \overline  {N}^{\,\nu_\ell}_k(\tau) = \int_0^{ \tau\wedge \theta^{\nu_\ell} }  \Big( F_k^1-
\langle F^1_k\rangle_\Lambda \Big) (I^{\nu_\ell}(s),
\Phi^{\nu_\ell}(s)) \,ds.
$$
Similar we write 
$$
M^{\nu_\ell}_j(\tau)= \tilde M^{\nu_\ell}_j(\tau)+   \overline  {M}^{\,\nu_\ell}_j(\tau),
$$
where $\tilde M^{\nu_\ell}_j$ is obtained by replacing  $\lan F_j^2\ran_\Lambda$ by $F_j^2$ 
in the formula for $M^{\nu_\ell}_j$, and $ \overline  {M}^{\,\nu_\ell}_j$ is the disparity. We claim that 
\begin{equation}\label{conv}
\E\sup_{0\le\tau \le T}  | \overline  {N}^{\,\nu_\ell}_k(\tau) |
+ \E\sup_{0\le\tau \le T}  | \overline  {M}^{\,\nu_\ell}_j(\tau) |
\to 0 \quad \mbox{as }
\nu_\ell \to 0\ .
\end{equation}
This assertion  follows
 from the following lemma:
 
\begin{lemma}\label{l2.3ihp}
Let $\tilde F=F_k^1=v_k\cdot P_k+b_k^2$ or $\tilde F=iv_k\cdot P_k$ (so this is 
 a polynomial function of degree not bigger than
$m$),  and $ G (v)= \tilde F(v)I_k^p$, for a fixed $p\le0$ and 
 $k\le N$. Then 
\begin{equation}\label{2.14ihp}
\E \max_{0\le \tau\le  {\theta_1^-}^\nu\wedge \tau_P^\nu} \left|\int_0^{\tau }
\Big(  G(I^\nu(s),\vp^\nu(s))-\langle
G\rangle_\Lambda(I^\nu(s),\Phi^\nu (s))\Big) ds  \right| \to 0\quad \mbox{as }
\nu\to 0\ .
\end{equation}
\end{lemma}
\noindent
{\it Proof. } For this proof we adopt a notation from \cite{KP08}. Namely, 
we denote by $\vk(t)$ various functions of $t$ such that $\vk\to0$ 
as $t\to\infty$, and denote by $\vk_\infty(t)$ functions, satisfying  
$\vk(t)=o(t^{-N})$ for each $N$. We write $\vk(t,R)$ to indicate that $\vk(t)$ 
depends on a parameter $R$. Besides for events $Q$ and $O$ and  a
random variable  $f$ we write $\PP_O(Q)=\PP(O\cap Q)$ and 
$\E_O(f)=\E(\chi_O\, f)$. We will also abbreviate $ {\theta_1^-}^\nu\wedge \tau_P^\nu=:  \theta^\nu$. 
\medskip

For $\tau\le \theta^\nu$ the function 
 $G$ is bounded by a constant, independent from $\nu$ but 
  depending on  $P$ and $\delta$ (see Steps~1 and 2 in Section~\ref{s5.22}). 
  Below in the proof the dependence of this and other constants on $P$ and $\delta$ will not be  indicated. 

In order to work with the functions $\tilde F=\tilde F_k$ with $k\le N$, we will 
control their arguments 
 $v_j$ for  $j\le M$, where $M$ is a suitable number,   bigger than $N$. 
 Similar to the above, we denote by $v^M$, $I^{\nu M}$ and
$\vp^{\nu M}$, the vectors, formed by the first $M$ components of the corresponding
infinite vectors,  %$v$, $I^\nu$, $\vp^\nu$; 
and denote by by $\Phi^{\nu (M)}$ the vector
formed by the first $J(M)$ components of $\Phi^\nu$.  Recall that $h=h^r$, where $r\ge d/2+1$,
while  the smoothness $r>d/2$ is sufficient for our constructions. 
Since for
$\tau\le \tau_P^\nu$ one has $|v|_{h^r}\le P$, we have
$|v-v^M|_{h^{r-1/2}} \le C M^{-1/(2d)}$, and since the function 
$G$ is Lipschitz on the space 
$h^{r-1/2}$ uniformly on the bounded set, then the l.h.s. of \eqref{2.14ihp}
is smaller than $C M^{-1/(2d)} + \mathfrak A^\nu_M$, with
$$
\frak A^\nu_M:= \E \max_{0\le \tau\le \theta^\nu} \left|\int_0^{\tau}
\left(  G(I^{\nu M}(s),\vp^{\nu M}(s))-\langle
G\rangle_\Lambda^M(I^{\nu M}(s),\Phi^{\nu (M)} (s))\right)ds
\right|\ ,
$$
where $\lan\cdot\ran^M_\Lambda$ is defined in \eqref{defmedia} (so that 
$
\lan G(I^M,\vp^M)\ran^M_\Lambda
$
is a function of $I^M$ and $\Phi^{(M)}$). 

Consider a partition of $[0,T]$  by the points
$$
\tilde \tau_n=\tilde \tau_0+nL,\quad 0\le n\le K\sim T/ L. 
$$
where $\tilde \tau_{K}$ is the last point $\tilde \tau_n$ in $[0,T)$.
  The diameter $L$ of the partition is 
$$
L=\sqrt\nu, 
$$
and the non-random phase $\tilde \tau_0\in[0,L)$ will be chosen later.
  Denoting
\begin{equation}\label{5.17}
\eta_l= \int_{\tilde \tau_l}^{\tilde \tau_{l+1}}\chi_{\tau\le \theta^\nu}\Big(G(I^{\nu M},
\vp^{\nu M})-   \lan G\ran^M_\Lambda(I^{\nu M},\Phi^{\nu (M)})\Big)  ds,\quad 
0\le l\le K-1,
\end{equation}
we see that 
\begin{equation}\label{5.170}
\mathfrak A^\nu_{M} \le LC+\E\sum_{l=0}^{K-1}|\eta_l|,
\end{equation}
so it remains to estimate $\sum\E |\eta_l|$. 
We have
\begin{equation*}
 \begin{split}
&\  |\eta_l|   \le \left|
\int_{\tau_l}^{\tau_{l+1}}\Big(G(I^{\nu M}(s),
\vp^{\nu M}(s))-   G (I^{\nu M} (\tau_l), \vp^{\nu M} 
(\tau_l)+\nu^{-1}\Lambda^m(s-\tau_l
 )) \Big) ds\right|\\
&+ \left|
\int_{\tau_l}^{\tau_{l+1}}\Big(
 G (I^{\nu M} (\tau_l), \vp^{\nu M}
(\tau_l)+\nu^{-1}\Lambda^M(s-\tau_l))-
 \lan G\ran^M_\Lambda(I^{\nu M} (\tau_l), \Phi^{\nu (M)}(\tau_l) )
 \Big) \, ds\right|\\
& + \left|
\int_{\tau_l}^{\tau_{l+1}}\Big(
 \lan G\ran^M_\Lambda(I^{\nu M} (\tau_l),\Phi^{\nu (M)}(\tau_l)) - \lan
 G\ran^M_\Lambda(I^{\nu M} (s), \Phi^{\nu (M)}(s) )
 \Big) \, ds\right| \\
 &=:\Upsilon^1_l+\Upsilon^2_l+\Upsilon^3_l\ ,
\end{split}
\end{equation*}
where we have put $\tau_l=\tilde \tau_l\wedge \theta^\nu$.
To estimate the quantities $\Upsilon^{1,2,3}_l$ we first optimise the choice
of the phase $\tau_0$. Consider the events $\cE_l$, $1\le l\le K$,
\begin{equation}\label{5.151}
\cE_l = \{ I_k^\nu(\tau_l)\le \eps\;\text{for some}\;  N<k\le M\}\ , \quad \mbox{where
} \eps \ge \nu^a, \quad a=1/10\ .
\end{equation}
Since for each $k$ by Lemma \ref{l5.1} we have
$$
\int_0^L\sum_{n=0}^K \PP(I_k^\nu(\bar\tau_n)\le\eps) \,d\bar\tau_0
=\int_0^T \PP(I_k^\nu(\tau)\le\eps)  \,d\tau
=\vk(\eps^{-1})
$$
(here each $\bar\tau_n$ is regarded as a function of $\bar\tau_0$), then 
we can choose $\tau_0\in[0,L)$ in such a way that 
\begin{equation*} %\label{5.18}
K^{-1}\sum_{l=0}^{K-1}\IP(\cE_l)=\vk(\eps^{-1}; M).
\end{equation*}
For any $l$ consider the event 
$$
Q_l=\{\sup_{\tau_l\le\tau\le\tau_{l+1}}|I^\nu(\tau)-I^\nu(\tau_l)|_{h_I}\ge
C_1L^{1/3}\}\ .
$$
 It is not hard to verify using the Doob inequality that for a suitable choice of the constant 
 $C_1$  its probability satisfies
 $\IP(Q_l)\le \vk_\infty(L^{-1})$ 
(cf. \cite{KP08}). Setting
$$
\cF_l=\cE_l\cup Q_l\ ,\quad l=0,\dots, K-1, 
$$
we have that
\begin{equation*} %\label{5.20}
\frac1K\sum_{l=0}^{K-1}\IP(\cF_l)\le \vk(\eps^{-1};M)+
\vk_\infty(\nu^{-1/2})=: \tilde \vk\ .
\end{equation*}
Accordingly,
\begin{equation}\label{5.211}
\sum_{l=0}^{K-1}\left|(\E_{\cF_l})\Upsilon_l^j\right| \le
{C}{L}\sum_{l=0}^{K-1}\IP(\cF_l)\le C\tilde \vk:= \tilde
\vk_1\ ,\quad j=1,2,3. 
\end{equation}

If $\omega\in \cF_l^c$, then for $\tau\in[\tau_l,\tau_{l+1}]$ we have
that  $I^\nu_k(\tau)\ge\eps -C_1 L^{1/3}\ge \frac12
\eps$. This relation and \eqref{5.3}, \eqref{5.6} imply that  
\begin{equation*}
 \begin{split}
\IP_{\cF_l^c} \{ |\vp^{\nu M}(s)-(\vp^{\nu M}(\tau_l)+\nu^{-1}\Lambda^M(s-\tau_l)|
\ge \nu^a \;\;\text{for some}&\;\; s\in[\tau_l,\tau_{l+1}]
\}\\
&\le \vk_\infty(\nu^{-1};M)\ 
\end{split}
\end{equation*}
(cf. the estimate for $Q_l$). 
Accordingly, 
\begin{equation}\label{5.001}
  \sum_l  
\E_{\cF_l^c}  \Upsilon^1_l  \le C \nu^{1/6}+ C\nu^a+\vk_\infty(\nu^{-1};M).
\end{equation}
For the same reasons also 
\begin{equation}\label{5.002}
  \sum_l  
\E_{\cF_l^c}  % \left(\sum_l 
 \Upsilon^3_l  \le C \nu^{1/6}+C\nu^a+\vk_\infty(\nu^{-1};M) \ . 
\end{equation}
So it remains to estimate the expectation of $\sum\Upsilon^2_l$. For
any $\omega\in \cF_l^c$  abbreviate 
$$
 {\widetilde G}(\psi)= G(I^{\nu M}(t_l),\vp^{\nu M}(t_l)+\psi),\qquad
\psi\in\T^M\ ,
$$
where  in the r.h.s. $\psi$ is identified with  the vector
$(\psi,0,\dots)\in\T^\infty$. Then
$$
\lan \widetilde G(\psi)\ran_{\Lambda^M} = \lan G(I^{\nu M}(t_j), \vp^{\nu M}(t_l)\ran^M_\Lambda.
$$
  Denoting $t=\nu\tau$ we  write $\Upsilon^2_l$ as 
$$
\Upsilon^2_l=\left|\int_{\tau_l}^{\tau_{l+1}}{\widetilde G}(\nu^{-1}\Lambda^M(s-\tau_l))\,ds-
L\lan {\widetilde G}\ran^M_\Lambda \right|=L
\left|  \frac{\nu}{L} \int_{0}^{\nu^{-1}L} {\widetilde G}(\Lambda^M t)\,dt-\lan {\widetilde G}\ran^M_\Lambda
\right|.
$$
Since the function $F(\psi)$ is of degree $m$,  then by Lemma~\ref{l.aver} 
$$
\Upsilon^2_l\le L\vk(\nu^{-1}L;M,m,\eps,\Lambda)\ .
$$
Therefore
\begin{equation}\label{5.003}
 \sum_l  \E_{\cF_l^c}
%\left(\sum_l
\Upsilon^2_l  \le \vk(\nu^{-1/2};M,m,\eps,\Lambda). 
\end{equation}

Now \eqref{5.170},  \eqref{5.211} and  \eqref{5.001}-\eqref{5.003}
imply that the l.h.s. of \eqref{2.14ihp} is smaller than 
$$
CM^{-1/(2d)}+\vk(\nu^{-a};M)+ \vk(\eps^{-1};M)+
C\nu^a
+C\nu^{1/6}+\vk(\nu^{-1/2};M, m,\eps,\Lambda). 
$$
Choosing first $M$ large and next $\eps$ small and $\nu$ small in
such a way that \eqref{5.151} holds,  
 we make the quantity above  arbitrarily small. This proves the lemma.
\qed

\noindent 
{\it End of the proof of Lemma  \ref{l5.media}.
}
Since $\cD(I^{\nu}(\cdot),\Phi^{\nu}(\cdot))\strela Q_P$,  where $\nu=\nu_\ell$, then 
\begin{equation*}
\begin{split}
N_k(\tau)&=I_k( \tau\wedge \theta^-_1\wedge \tau_P ) -
\int_0^{\tau\wedge \theta^-_1\wedge \tau_P  } \langle F^1_k\rangle_\Lambda(I(s),
\Phi(s)) \,ds\ ,\quad  k\le N\ ,\\
M_j(\tau)&=\Phi_j(\tau\wedge \theta^{-}_1\wedge \tau_p)
-\sum_{k \ge   1}s^{(j)}_k \int_0^{\tau\wedge \theta^{-}_1\wedge \tau_p  }\langle F^2_k\rangle_\Lambda
(I(s), \Phi(s)) 
\,ds \ , \quad  j\le J\ ,
\end{split}
\end{equation*}
are $\cQ_P$ martingales.

Similar to \eqref{2.14ihp} one finds that
$$
\E \max_{0\le \tau \le T} \left|\int_0^{\theta^\nu}
\left(  G(I^\nu(s),\vp^\nu(s))-\langle
G\rangle_\Lambda(I^\nu(s),\Phi^\nu (s))\right)\right|^2 \to 0\quad \mbox{as }
\nu=\nu_\ell\to 0\ .
$$
Then, using the same arguments as before, we see that the processes
$N_{k_1}(\tau) N_{k_2}(\tau)- \int_0^{\tau\wedge \theta^-_1\wedge \tau_P} \langle 
A_{k_1k_2}\rangle_\Lambda d s$ are $\cQ_P$
martingales, where $A_{k_1k_2}$ denotes the diffusion matrix for the
system \eqref{5.21}, while a similar argument applies for $M_{j_1}M_{j_2}$
and $M_j N_k$. Taking into account the explicit form of the diffusion,
the proof is then concluded.
\qed

A simplified version of the argument, proving Lemma \ref{l2.3ihp}
 implies the following assertion: 

\begin{proposition}\label{r34}
Let $s\in\Z_0^\infty$ be such that $s\cdot\Lambda\ne0$ and 
$G:\R_+^M\times \T^{J(M)}\times S^1\to\R$ be a bounded Lipschitz-continuous function,
for some $M\ge1$. Then
\begin{equation*}
\begin{split}
\frak B^\nu:=
\E\max_{0\le\tau\le T} &\Big|
 \int_0^\tau \Big( G(I^{\nu M}(l), \Phi^{\nu (M)}(l),s\cdot \vp^\nu(l))-\\
 &\int_{S^1}  G(I^{\nu M}(l), \Phi^{\nu (M)}(l), \theta)\,\dbar\theta\Big)dl\Big|\to0\quad\text{as}\quad
 \nu\to0. 
 \end{split}
\end{equation*}
In particular, taking for $G$  Lipschitz functions on $S^1$ we get that 
$$
\frac1{t} \int_0^t \cD(s\cdot\vp^\nu(l))dl\strela d\theta\quad \text{as}\quad \nu\to0, 
$$
for any $t>0$. 
\end{proposition}

For a proof see Appendix B.

\subsection{Proof of Lemma~\ref{l.couple}}\label{s.pr_couple}
Below we follow \cite{KP08, K10}. 
First for any $\eps>0$ we construct an auxiliary 
 process $(\bar w^\eps, \widetilde w^\eps)$.
We  choose  $\bar w^\eps =w^\nu$, 
and build  $\widetilde w^\eps=\widetilde w\in\C^N$ as follows. We set 
$\widetilde w(\tm)=w^\nu(\tm)$ and for $\tau>\tm$ define $\widetilde w(\tau)$ as a solution of
the system 
\begin{equation}\label{^1}
d\widetilde w_k=e^{i (\widetilde\vp_k-\bar\vp_k) } P_k(\bar w)\, d\tau
+e^{i (\widetilde\vp_k-\bar\vp_k)} b_k\,\bb^k(\tau),\quad  k\le N,
\end{equation}
where $\widetilde\vp_k=\vp (\widetilde w_k)$, etc.  Let us define a stopping  time $\tau^1>\tm$,
$$
\tau^1= \inf \{\tau\in[\tm, T]:  [I(\bar w(\tau))] \wedge [I(\widetilde w(\tau))]=\eps\}.
$$
Due to \eqref{^1} for $\tau\le \tau^1$ we have
\begin{equation*}
\begin{split}
d\tilde I_k&= \big((\widetilde w_k\cdot e^{i (\widetilde\vp_k-\bar\vp_k) }P_k(\bar w)\big)\,d\tau
+b_k^2\,d\tau  +b_k\big( \widetilde w_k\cdot e^{i (\widetilde\vp_k-\bar\vp_k) }\,d\bb^k\big)\\
&=\Big(\sqrt{2\tilde I_k}\,e^{i \bar\vp_k}  \cdot P_k(\bar w) +b_k^2\Big)\,d\tau  +
b_k \sqrt{2\tilde I_k}\,e^{i \bar\vp_k} \cdot d\bb^k,\qquad k\le N.  
\end{split}
\end{equation*}
Noting that the equation \eqref{5.2} for $\bar I_k$ may be written as 
$$
d\bar I_k= \big(\sqrt{2\bar I_k}\,e^{i \bar\vp_k}  \cdot P_k(\bar w) +b_k^2\big)\,d\tau  +
b_k \sqrt{2\bar I_k}\,e^{i \bar\vp_k} \cdot d\bb^k
$$
and that $\bar I_k(\tm)= \tilde I_k(\tm)$ for $k\le N$ we arrive at the relation
\begin{equation}\label{^2}
\tilde I^N(\tau)=\bar I^N(\tau)\quad\text{a.s.},
\end{equation}
for $\tm\le \tau\le \tau^1$. \footnote{Since these are two solutions of a Cauchy problem 
for a stochastic ODE which is Lipschitz in the domain $\{\eps\le I_j\le\tfrac12 P^2,\ 
\forall\,j\}$, and one of them  stays in this domain a.s.}
Now we define  next stopping time $\tau^2$ as 
$$
\tau^2 = \inf \{\tau\in[\tau^1,T]: [I(\bar w)(\tau)]=2\eps\},
$$
and for $\tau\in[\tau^1, \tau^2],\ j\le N$, set $\widetilde w_j(\tau)$ to be a rotation of  $\bar w_j(\tau)$ on  a $\tau$-independent  angle, chosen by continuity:
\begin{equation}\label{^3}
\widetilde w(\tau)=\Psi_{\theta^N}\bar w^N(\tau), \qquad \theta^N=\vp^N( \widetilde w(\tau^1))- 
 \vp^N( \bar w(\tau^1)).
\end{equation}
Now the process $\widetilde w(\tau)$ is defined till $\tau=\tau^2$ and still satisfies \eqref{^2}. 

Iterating the two steps above we construct a process $(\bar w^\eps, \widetilde w^\eps)(\tau)$ 
such that \eqref{^2}  holds for all $\tau\le T$. 

Let us consider another slow variable -- a 
 resonant combination of angles $\Phi_j=\vp\cdot s^{(j)}$. Literally repeating the 
proof of \eqref{^2} we get that $\Phi_j(\bar w^\eps(\tau))= \Phi_j (\widetilde w^\eps(\tau))$
a.s., for $\tau\le \tau^1$   and  $j\le J(N)$.\footnote{Certainly this is not true for individual phases 
$\vp_j(\bar w^\eps(\tau))$ since these are fast variable, while the phases of
 $\widetilde w^\eps(\tau)$ are slow.} Jointly with \eqref{^2} it implies that 
\begin{equation}\label{^4}
V_j(\bar w^\eps(\tau))= V_j (\widetilde w^\eps(\tau))\quad\text{a.s.},\;\;\forall\, j\le J,
\end{equation}
for $\tau\le\tau^1$.  Using \eqref{^3} we see that this relation holds for 
 $\tau \le\tau^2$.
 Iterating the argument we get that \eqref{^4} is valid till $\tau=T$.

 The constructed process $(\bar w^\eps, \widetilde w^\eps)(\tau)$, $\eps>0$,
 satisfies  (i), (ii) and satisfies (iii), unless $[\bar I(\tau)]\le 2\eps$. By Lemma~\ref{l5.1} the latter event 
 becomes improbable when $\eps$ converges to zero, and it not hard to check 
 (see  \cite{KP08, K10})  that  a 
 limiting in law, as $\eps\to0$, process $(\bar w^\nu, \widetilde w^\nu)(\tau)$ exists and
 satisfies (i)-(iii). $\qed$
 \medskip
 
 \noindent
 {\it Remark.} Applying to the Ito equation \eqref{5.mod}, satisfied by the process $\tilde w^\nu$,  Krylov's theorem 
 from \cite{Kry77} (also  see in \cite{KS}),  we recover the assertion of Lemma~\ref{l5.1}.

\appendix
\section{Ito's  formula  in  complex variables}\label{a1}

Let $\{\Omega,\cF,\PP;\cF_t\}$ be a filtered probability space and  $z(t) \in \C^N$
be a complex Ito process  on this space of the form 
\begin{equation*} %\label{eq:ODEcompl}
dz_k= a_k(t)  dt+ b_k(t) d \bb^k(t)\ ,\quad 1\le k\le N\ ,
\end{equation*}
where $a_k,b_k$, are $\{\cF_t\}$-adapted complex processes. We assume that the 
 processes satisfy ``usual" 
growth conditions, needed to apply the Ito formula (these conditions 
 are clearly met each  time when we use the formula),
 and $\bb^k(t)$ are standard
independent complex Wiener processes.
The result below   easily follows from the usual (`real') Ito's formula.

\begin{lemma}\label{l.a1}
  Let $f:\C^N\to \C$ be a $C^2$-function. Then  
\begin{equation}\label{ito}
\begin{split}
df(z(t))=
 \sum_k\Big( a_k\frac{\partial f}{\partial z_k}&+\bar a_k \frac{\partial
  f}{\partial  \bar z_k} + 2|b_k|^2 \frac{\partial^2f}{\partial
  z_k\partial \bar   z_k} \Big)dt\\
  &+ \sum_k\Big(b_k\frac{\partial f}{\partial
  z_k} d\bb^k(t)+\bar b_k \frac{\partial 
  f}{\partial  \bar z_k}d \bar \bb^k(t)\Big) \ .
  \end{split}
\end{equation}
\end{lemma}

\section{Proof of Proposition~\ref{r34}}
For this proof, as in Section~\ref{sez:dim}, 
we denote by $\vk(t)$ various functions of $t$ such that $\vk\to0$ 
as $t\to\infty$, and denote by $\vk_\infty(t)$ functions, satisfying  
$\vk(t)=o(t^{-N})$ for each $N$.  For events $Q$ and $\cal O$ and  a
random variable  $f$ we write $\PP_{\cal O}(Q)=\PP({\cal O}\cap Q)$ and 
$\E_{\cal O}(f)=\E(\chi_{\cal O}\, f)$. Without lost of generality we assume that $|G|\le1$
and Lip$\,G\le1$. 
\medskip

Let us denote by $R$  a suitable function of
$\nu$ such that $R(\nu)\to \infty$ as $\nu\to 0$, but
$$
\nu R^n\to 0\quad \mbox{as }\nu\to 0\,, \quad \forall\, n\ .
$$
Denote, moreover, by $\Omega_R=\Omega^\nu_R$ the event
$\ 
\Omega_R= \left\{\sup_{0\le\tau \le T} |v^\nu(\tau)|_r\le
R\right\}\ .
$
Then, by \eqref{gafa3.10}, $\PP(\Omega^c_R) \le \vk_\infty(R)$
uniformly in $\nu$.

Taking into account the boundedness  of $G$, we get that
\begin{equation*}
\begin{split}
\mathfrak{B}^\nu\le \vk_\infty(R)+ \E_{\Omega_R} \max_{0\le \tau \le
  T} & \Big|
 \int_0^\tau \Big( G\big(I^{\nu M}(l), \Phi^{\nu (M)}(l),s\cdot \vp^\nu(l)\big)\\
 &-\int_{S^1}  G(I^{\nu M}(l), \Phi^{\nu (M)}(l),
 \theta)\,\dbar\theta\Big)dl\Big| \ .
\end{split}
\end{equation*}

As in  the proof of Lemma~\ref{l2.3ihp},  consider a partition of
$[0,T]$ by the points
\begin{equation}\label{tau}
\tau_n=\tau_0+nL,\quad 0\le n\le K\sim T/ L. 
\end{equation}
where $\tau_{K}$ is the last point $ \tau_n$ in $[0,T)$.
The diameter $L$ of the partition is 
$
L=\sqrt\nu, 
$
and the non-random phase $ \tau_0\in[0,L)$ will be chosen later.
  Denoting
\begin{equation}\label{5.17.1}
\eta_n= \int_{\tau_n}^{
  \tau_{n+1}}\Big(G(I^{\nu M},\Phi^{\nu (M)},s\cdot \vp^\nu)-
\int_{S^1}  G(I^{\nu M}, \Phi^{\nu (M)},
 \theta)\,\dbar\theta\Big) dl,\quad 
0\le l\le K-1,
\end{equation}
we see that 
\begin{equation}\label{5.170.1}
\mathfrak B^\nu \le \vk_\infty(R)+CL +\E_{\Omega_R}\sum_{n=0}^{K-1}|\eta_n|,
\end{equation}
so it remains to estimate $\sum\E_{\Omega_R} |\eta_n|$. 
We abbreviate 
$$
{\hat G}(\psi;l)= G(I^{\nu M}(l),\Phi^{\nu(M)}(l),\psi)\ ,
\quad  \psi\in S^1\ ,
$$
so that we have 
\begin{equation*}
 \begin{split}
  |\eta_n|  & \le \left|
\int_{\tau_n}^{\tau_{n+1}}\left({\hat G}(s\cdot \vp^\nu(l);l)) -   {\hat G} \big( s\cdot
\vp^\nu(\tau_n)+\nu^{-1} (s\cdot \Lambda)(l-\tau_n) ; \tau_n\big)\right) dl\right|\\
&+ \left| 
\int_{\tau_n}^{\tau_{n+1}}\left(
 {\hat G} \big(s\cdot \vp^\nu(\tau_n)+\nu^{-1} (s\cdot  \Lambda)(l-\tau_n) ; \tau_n\big)-
 \int_{S^1} {\hat G}(\theta;\tau_n )\, \dbar \theta
 \right) \, dl\right|\\
& + \left|
\int_{\tau_l}^{\tau_{l+1}}\Big(  \int_{S^1} {\hat G}(\theta;\tau_n )\, \dbar \theta
  -   \int_{S^1} {\hat G}(\theta;l )\, \dbar \theta
 \Big) \, dl\right| =:\Upsilon^1_n+\Upsilon^2_n+\Upsilon^3_n\ .
\end{split}
\end{equation*}

To estimate the quantities $\Upsilon^{1,2,3}_n$ we first optimise the choice
of the phase $\tau_0$. A crucial point here is that, if we set
$N:=M\vee \lc s \rc$, the function $G$ depends only on $v^N$. 
So we consider the events $\cE_n$, $1\le n\le K$,
\begin{equation}\label{5.151.1}
\cE_n = \{ I_k^\nu(\tau_n)\le \eps\;\text{for some}\;  k\le N\}, \quad \mbox{where
} \eps \ge  \nu^a, \quad a=1/10\ .
\end{equation}
Since for each $k$ by Lemma \ref{l5.1} we have
$$
\int_0^L\sum_{n=0}^K \PP(I_k^\nu(\bar\tau_n)\le\eps) \,d\bar\tau_0
=\int_0^T \PP(I_k^\nu(\tau)\le\eps)  \,d\tau
=\vk(\eps^{-1};R)
$$
(here each $\bar\tau_n$ is regarded as a function of $\tau_0=\bar\tau_0$, given by \eqref{tau}), 
then 
we can choose $\tau_0\in[0,L)$ in such a way that 
\begin{equation*} %\label{5.18}
K^{-1}\sum_{n=0}^{K-1}\IP(\cE_n)=\vk(\eps^{-1};R,N).
\end{equation*}
For any $n$ consider the event 
$$
Q_n=\{\sup_{\tau_n\le\tau\le\tau_{n+1}}|I^\nu(\tau)-I^\nu(\tau_n)|_{h_I}\ge
P_1(R)L^{1/3}\}\ ,
$$
where $P_1(R)$ is a suitable polynomial. It is not hard to verify
using the Doob inequality that  its probability satisfies
$\IP(Q_n)\le \vk_\infty(L^{-1})$  (cf. \cite{KP08}). Setting
$\ 
\cF_n=\cE_n\cup Q_n,
$
$n=0,\dots, K-1$, we have that
\begin{equation*} %\label{5.20}
\frac1K\sum_{n=0}^{K-1}\IP(\cF_n)\le \vk(\eps^{-1};R,N)+
\vk_\infty(\nu^{-1/2};N)=: \tilde \vk\ .
\end{equation*}
Accordingly,
\begin{equation}\label{5.211.1}
\sum_{n=0}^{K-1}\left|(\E_{\cF_n\cap \Omega_R})\Upsilon_n^j\right| \le
{C}{L}\sum_{n=0}^{K-1}\IP(\cF_n)\le C\tilde \vk:= \tilde
\vk_1\ ,\quad j=1,2,3. 
\end{equation}

If $\omega\in \Omega_R\backslash \cF_n$, then for
$\tau\in[\tau_n,\tau_{n+1}]$ we have 
that  $I^\nu_k(\tau)\ge\eps -P_1(R) L^{1/3}\ge \frac12
\eps$. This relation and \eqref{5.3}, \eqref{5.6} imply that  
\begin{equation*}
 \begin{split}
\IP_{\Omega_R\backslash \cF_n} \{ |\vp^{\nu N}(l)-(\vp^{\nu M}(\tau_n)+\nu^{-1}\Lambda^N(l-\tau_n))|
\ge \nu^a \;\;\text{for some}&\;\; l\in[\tau_n,\tau_{n+1}]
\}\\
&\le \vk_\infty(\nu^{-1};R,N)\ 
\end{split}
\end{equation*}
(cf. the estimating of $\PP(Q_n)$). 
Therefore
\begin{equation*}
 \begin{split}
\IP_{\Omega_R\backslash \cF_n} \big\{ |s\cdot\vp^{\nu N}(l)-(s\cdot\vp^{\nu N}(\tau_n)&+\nu^{-1}(s\cdot\Lambda)(l-\tau_n))|
\ge \nu^a 
\\
&\text{for some}\;\; l\in[\tau_n,\tau_{n+1}]\big\}
\le \vk_\infty(\nu^{-1};R,N,s, \Lambda)\ .
\end{split}
\end{equation*}
and
\begin{equation*}
\IP_{\Omega_R\backslash \cF_n} \{ |\Phi^{\nu (M)}(l)-\Phi^{\nu (M)}(\tau_n)|
\ge \nu^a \;\;\text{for some}\;\; l\in[\tau_n,\tau_{n+1}]
\}\le \vk_\infty(\nu^{-1};R,N)\ .
\end{equation*}
Accordingly, 
\begin{equation}\label{5.001.1}
  \sum_l  
\E_{\Omega_R\backslash \cF_n}  \Upsilon^1_n  \le C(R) \nu^{1/6}+
C(R)\nu^a+\vk_\infty(\nu^{-1};R,N,s\cdot \Lambda).
\end{equation}
For the same reason also 
\begin{equation}\label{5.002.1}
  \sum_l  
\E_{\Omega_c\backslash \cF_n}  % \left(\sum_l 
 \Upsilon^3_n  \le C(R)
 \nu^{1/6}+C(R)\nu^a+\vk_\infty(\nu^{-1};R,N,s\cdot \Lambda) \ . 
\end{equation}

So it remains to estimate the expectation of $\sum\Upsilon^2_n$.
Denoting $t=\nu(l-\tau_n)$ we  write $\Upsilon^2_n$ as 
\begin{equation*}
\begin{split}
\Upsilon^2_n&=\left|\int_{\tau_n}^{\tau_{n+1}}{\hat G}(s\cdot
\vp^\nu(\tau_n)+\nu^{-1}(s\cdot \Lambda)(l-\tau_n);\tau_n) \,dl-
L\int_{S^1} {\hat G}(\theta;\tau_n)\,\dbar \theta \right|\\
&=L
\left|  \frac{\nu}{L} \int_{0}^{\nu^{-1}L} {\hat G}(s\cdot
\vp^\nu(\tau_n)+s\cdot\Lambda t;\tau_n)\,dt- \int_{S^1}
{\hat G}(\theta;\tau_n)\,\dbar \theta
\right|.
\end{split}
\end{equation*}

Let us  expand ${\hat G}(\psi;\tau_n)$  as a Fourier series
${\hat G}(\psi)=\sum g_k e^{ik\psi}$, where each $g_k$ is a random variable and
$\ g_0=
\int_{S^1}
{\hat G}(\theta;\tau_n)\,\dbar \theta\,
$
 (we discard  the dependence
on $\tau_n$, which is   fixed thought the argument). Then 
$$
\left| \frac1T \int_0^T {\hat G}(\psi_0+t(s\cdot \Lambda))\,dt-g_0\right|\le \eps
\quad \forall\, T\ge T_\eps\ ,
$$
for a suitable non-random $T_\eps$. Indeed, for each nonzero $k$, one has
$$
\left| \frac1T \int_0^T e^{ik(\psi_0+t(s\cdot \Lambda))}
\,dt\right|\le  \frac{2}{T |s\cdot \Lambda|}\ ,
$$
so that
\footnote{By the Bernstein theorem, $\sum_{k=1}^\infty |g_k|\le C$,
where the constant $C=C({\hat G})$ is finite if the function ${\hat G}(\psi)$ is  Lipschitz-continuous. 
The proof of the theorem (e.g., see \cite{Zyg1}, Section~VI.3) easily implies that $C$ 
depends only on the Lipschitz constant of ${\hat G}$, which equals 1 in our case.}
$$
\left| \frac1T \int_0^T {\hat G}(\psi_0+t(s\cdot \Lambda))\,dt-g_0\right|\le
\frac{2}{T|s\cdot \Lambda|} \sum |g_k|\le \frac{2C}{T |s\cdot \Lambda| } \ .
$$
We have thus proved  that
$\ 
\Upsilon^2_n\le L\vk(\nu^{-1}L;R,N,\eps,s\cdot\Lambda)\ .
$
Therefore
\begin{equation}\label{5.003.1}
 \sum_l  \E_{\Omega_R\backslash \cF_l^c}
%\left(\sum_l
\Upsilon^2_n  \le \vk(\nu^{-1/2};R,N,\eps,s\cdot\Lambda). 
\end{equation}

Now \eqref{5.170.1},  \eqref{5.211.1} and  \eqref{5.001.1}-\eqref{5.003.1}
imply that $\mathfrak{B}^\nu$ is bounded by 
$$
 \vk_\infty(R)+\vk(\nu^{-a};R,N)+ \vk(\eps^{-1};R,N)+
C(R)\nu^a
+C(R)\nu^{1/6}+\vk(\nu^{-1/2};R,N, \eps,s\cdot\Lambda). 
$$
Choosing first $R$ large and next $\eps$ small and $\nu$ small in
such a way that \eqref{5.151.1} holds,  
 we make the quantity above  arbitrarily small. This proves the required convergence.
 
 The second assertion of the proposal follows from the first one since to check the weak 
 convergence of measures on a complete metric space it suffice to take for test-functions 
 the Lipschitz functions.
\qed

\bibliography{meas}
\bibliographystyle{amsalpha}
\end{document}